\documentclass[12pt]{article}
\usepackage{amsmath}
\usepackage{amsthm}
\usepackage{graphicx}
\usepackage{enumerate}
\usepackage[numbers]{natbib}
\usepackage{bibunits}
\usepackage{url} 

\newcommand{\blind}{1}

\usepackage{amssymb }
\usepackage{enumitem}
\usepackage{algorithm2e}
\usepackage{chapterbib}
\usepackage[colorlinks=true, allcolors=blue]{hyperref}
\usepackage{titletoc}

\DeclareMathOperator*{\argmax}{arg\,max}

\DeclareMathOperator{\Tr}{Tr}

\DeclareUnicodeCharacter{008D}{}
\DeclareUnicodeCharacter{0080}{}
\DeclareUnicodeCharacter{0093}{}

\newtheorem{theorem}{Theorem}[section]
\newtheorem{lemma}[theorem]{Lemma}

\newtheorem{proposition}[theorem]{Proposition}
\newtheorem{corollary}[theorem]{Corollary}

\newtheorem{definition}[theorem]{Definition}

\usepackage{setspace}
\usepackage{booktabs} 
\usepackage{caption} 
\usepackage{siunitx} 
\usepackage{threeparttable}
\usepackage{bbm}
\usepackage{xcolor}
\usepackage{sectsty}

\begin{document}


\def\spacingset#1{\renewcommand{\baselinestretch}%
{#1}\small\normalsize} \spacingset{1}


\if1\blind
{
  \title{\sf Sparse Bayesian Multidimensional Item Response Theory}
  \author{Jiguang Li\footnote{Jiguang Li is a $3^{rd}$-year doctoral student in Econometrics and Statistics at the Booth School of Business of the University of Chicago},\,\, Robert Gibbons\footnote{Robert Gibbons is the Blum-Reise Professor of Statistics at the University of Chicago}\,\, and  Veronika Ro\v{c}kov\'{a}\footnote{Veronika Ro\v{c}kov\'{a} is the Bruce Lindsay Professor of Econometrics and Statistics in the Wallman Society of Fellows}}
  \maketitle
} \fi

\if0\blind
{
   \title{\sf Sparse Bayesian Multidimensional Item Response Theory}
  \maketitle
  \medskip
} \fi

\begin{abstract}
Multivariate Item Response Theory (MIRT) is sought-after widely by applied researchers looking for interpretable (sparse) explanations underlying response patterns in questionnaire data.  There is, however, an unmet demand for such sparsity discovery tools in practice. Our paper develops a Bayesian platform for binary and ordinal item MIRT which requires minimal tuning and scales well on large datasets due to its parallelizable features. Bayesian methodology for MIRT models has traditionally relied on MCMC simulation, which cannot only be slow in practice, but also often renders exact sparsity recovery impossible without additional thresholding. In this work, we develop a scalable Bayesian EM algorithm to estimate sparse factor loadings from mixed continuous, binary, and ordinal item responses. We address the seemingly insurmountable problem of unknown latent factor dimensionality with tools from Bayesian nonparametrics which enable estimating the number of factors. Rotations to sparsity through parameter expansion further enhance  convergence and interpretability without identifiability constraints. In our simulation study, we show that our method reliably recovers both the factor dimensionality as well as the latent structure on high-dimensional synthetic data even for small samples. We demonstrate the practical usefulness of our approach on three datasets: an educational assessment dataset, a quality-of-life measurement dataset, and a bio-behavioral dataset. All demonstrations show that our tool yields interpretable  estimates, facilitating interesting discoveries that might otherwise go unnoticed under a pure confirmatory factor analysis setting. 
\end{abstract}

\noindent%
{\it Keywords:}  EM algorithm; Unified skew-normal distribution; Spike-and-Slab LASSO; Categorical factor analysis; Bayesian nonparametrics
\vfill

\newpage
\spacingset{1.8} 
\section{Introduction} \label{sec:intro}
\vspace{-0.2cm}

Multidimensional Item Response Theory (MIRT) is indispensable to educational and psychological measurement research. It plays a dominant role in evaluating latent cognitive abilities of respondents, classifying subpopulations, and uncovering item characteristics \cite{bock2021item}. MIRT models have been utilized to assess individual differences in personality traits and attitudes \cite{Fraley2000AnIR}, improve health outcomes in clinical patient-reported outcome research \cite{FORERO2013790}, make scientific discoveries from intricate bio-behavioral data \cite{Stan2020}, and enhance testing efficiency and accuracy in computerized adaptive testing in behavioral health \cite{doi:10.1146/annurev-clinpsy-021815-093634}. 

In the standard two-parameter MIRT framework \cite{bock2021item}, we observe a binary matrix $Y \in \mathbbm{R}^{N \times J}$, where each element $Y_{ij}$ represents whether subject $i$ answered item $j$ correctly.  Assuming the number of factors $K$ is known, we denote the factor loading matrix as $B \in \mathbbm R^{J \times K}$, the intercept as $D \in \mathbbm R^{J}$, and subject $i$'s latent trait as $\theta_i \in \mathbbm R^{K}$. The two-parameter MIRT model aims to capture the following data generating process:
\setlength{\belowdisplayskip}{6pt} \setlength{\belowdisplayshortskip}{6pt}
\setlength{\abovedisplayskip}{6pt} \setlength{\abovedisplayshortskip}{6pt}
 $$Y_{ij} | B_j, \theta_i, d_j \sim \text{Bernoulli(}g(B_j'\theta_i + d_j)) ,$$
for subjects $i = 1, \cdots, N$ and items $j= 1, \cdots, J$, where $B_j$, as a column vector, is the $j$-th row of $B$, $d_j$ is the $j$-th element of the intercept $D$, and function $g(.)$ is the link function. 

The most common approach to estimating MIRT models is based on the principle of marginal maximum likelihood (MML) estimation \cite{DarrellBock1972}. Following the development of the EM algorithm \cite{em}, Bock and Aitkin \cite{em_mirt} derive a computationally efficient procedure to estimate the item parameters using deterministic Gaussian quadrature and extend the EM principle to the exploratory MIRT setting. Another line of MIRT estimation roots in the idea of Monte Carlo Expectation Maximization algorithm (MCEM), in which the E-step is often replaced with Monte-Carlo methods \cite{MCEM}. The popular Metropolis-Hastings Robbins-Monro (MHRM) algorithm  performs stochastic imputation with the Metropolis-Hastings sampler in the E-step and Robbins-Monro stochastic approximation in the M-step \cite{mhrm_exploratory}. 

The Joint Maximum Likelihood (JML) estimation method has gained popularity for estimating latent variable models in large-scale discrete datasets. For example, a constrained JML estimator demonstrates remarkable scalability and effectiveness in estimating high-dimensional exploratory MIRT models \cite{jml1}. In a confirmatory context, \cite{jml2} proposed a variation of the JML estimation method within a generalized latent factor model framework, establishing the structural identifiability of latent factors under a double-asymptotic setting. Additionally, the JML estimation approach has been successfully extended to the discrete latent variable and the Cognitive Diagnostic Models (CDMs) framework \cite{jml3}.

MCMC sampling coupled with data augmentation strategies of Albert and Chibs \cite{albert_chib}, or P\'{o}lya-Gamma variables \cite{pg, pg_mirt} has long been the only resort for Bayesian statisticians to estimate MIRT models. However, Durante \cite{Durante_2019} recently showed that for probit regression with a standard Gaussian prior on the coefficient parameter, the posterior follows a unified-skew normal distribution \cite{unified_skew_normal}, which can be easily sampled from with the minimax tilting methods \cite{Botev_2016}. This discovery has made our derivation of the latent factor posterior distributions tractable under the probit factor model and motivates us to advocate for a Bayesian EM approach.  Unlike the E-step in MCEM or MHRM algorithms which rely on MCMC sampling, our E-step sampling can be efficiently parallelized for each observation, and the M-step can also be decomposed into solving $J$ independent penalized probit regressions, which can be solved by a standard \texttt{glmnet} estimation algorithm \cite{glmnet}.

By explicitly characterizing latent factor posterior distributions, this paper develops a scalable Bayesian framework for estimating sparse exploratory MIRT models, with the aim of recovering an interpretable representation of the factor loading matrix $B$ without making assumptions on its zero allocations. The process of finding a sparse solution can be facilitated by the idea of parameter expansion \cite{px-em}, which enables rapid exploration of the parameter space of equivalent observed data likelihoods and is more robust against poor initialization. Our proposed framework and estimation procedures can also be further generalized beyond binary datasets to accommodate mixed responses. Parallel research in the CDM framework has also received significant attention recently: In CDMs, the latent structure, represented by the Q-matrix, is typically provided by experts. However, the Q-matrix can often be misspecified in practice. Consequently, sparsity recovery through the estimation of the Q-matrix has been explored in many recent works \cite{slcm, qmat_dina, de_la_Torre2016}.

Besides its computational efficiency and versatility to accommodate mixed data type, one primary advantage of our Bayesian approach is its flexibility in prior specification. In many exploratory MIRT settings, the number of latent factors and the exact sparsity level of the loading matrix are typically unknown and are often misspecified. To incorporate our prior beliefs amid such uncertainty, we impose a hierarchical Indian Buffet Process (IBP) prior \cite{ibp} coupled with the Spike-and-Slab LASSO (SSL) prior \cite{SSL} on the factor loading matrix, following the idea in \cite{bfa}. This allows the algorithm to automatically learn the latent dimensionality and produce sparse loading estimates. Additionally, we establish general identification conditions and provide theoretical guarantees for recovering the true number of latent factors and ensuring the consistency of the loading matrix. When the number of observations is relatively small, which is often the case in social science and behavior health research, our Bayesian approach can yield more reliable estimates by leveraging such prior information, resulting in a more sparse solution. 

The rest of this paper is organized as follows. Section \ref{sec:motivation} offers a concrete motivating example exhibiting the potential of our proposed approach to estimating sparse Bayesian MIRT models. Section \ref{sec:model} provides a detailed exposition of the data generation process, identification conditions, and the theoretical justifications behind our proposed hierarchical model. We then show the derivations of our parameter-expanded EM algorithm in Section \ref{sec:em}. Section \ref{sec:grm_ordinal} further demonstrates the generality of our proposed methodology by accommodating mixed data types. In Section \ref{sec:empirical}, the effectiveness of our proposed method is substantiated through its application to diverse scenarios: a challenging synthetic dataset, a large sample educational assessment dataset, a quality-of-life measurement dataset, and a challenging ordinal bio-behavioral dataset. We conclude the paper in Section \ref{sec:discussion} and provide additional details in the Appendix. 

\vspace{-0.7cm}
\section{Exact Sparsity Recovery} \label{sec:motivation}
\vspace{-0.3cm}

Prior to delving into technical details, we provide a simple motivating example to illustrate the usefulness of our Bayesian approach when the factor loading matrix is sparse and its dimensionality uncertain. To this end, we generated four-factor, two-parameter binary item response data, and focused on a challenging high-dimensional $N < J$ case, in which $50\%$ of items load on an arbitrary pair of latent factors.  As shown in the leftmost subplot in Figure \ref{fig:plot1}, the true loading matrix is relatively sparse - the light pink areas represent zeros, while the red areas represent ones. Although simple looking, this loading structure can be very hard to recover in practice, as there exists a substantial overlap among latent factors.  

For benchmarking, we consider the other three state-of-the art exploratory MIRT methods, where no assumptions are made for the zero allocations of the loading matrix : (1) the MHRM algorithm \cite{mhrm_exploratory}, (2) a customized Gibbs sampler with an adaptive spike-and-slab normal prior \cite{ir_ss} (derivations in Appendix A), and (3) the exploratory JML estimation method \cite{jml1}. The MHRM and the JML method implementations are available in the \texttt{mirt} and the \texttt{mirtjml} R packages \cite{JSSv048i06}.

\begin{figure}[t]
    \centering
    \includegraphics[width=0.7\textwidth, height=0.5\textwidth]{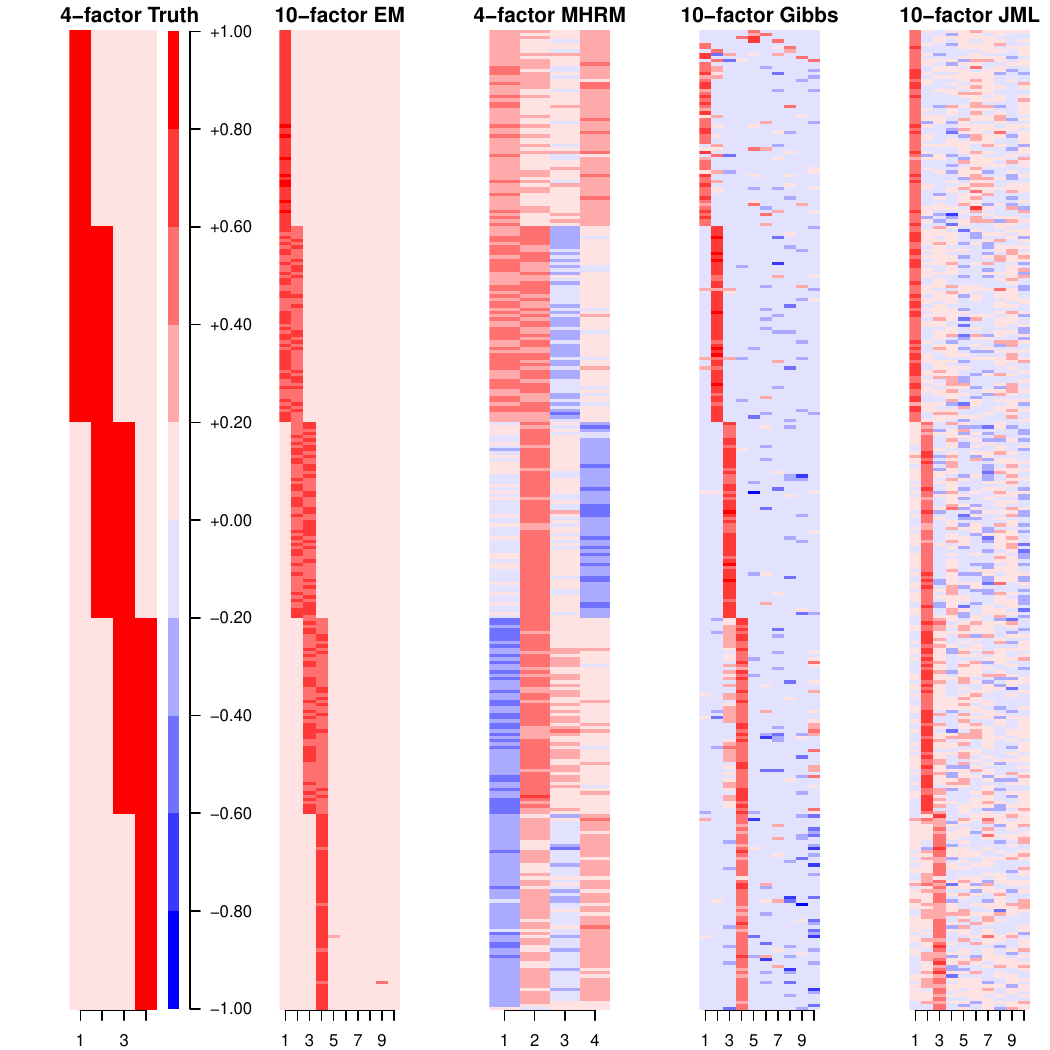}
    \caption{Factor Loading Matrix Estimation Comparison}
    \label{fig:plot1}
\end{figure}

To demonstrate that our approach is capable of learning the dimensionality and of producing a sparse solution, we purposefully mis-specified (overshot) the number of factors $K^*=10$ for our EM algorithm. As shown in Figure \ref{fig:plot1}, our PXL-EM algorithm recovered the true overlapping latent structure, and correctly identified the true number of factors as $4$, despite us purposefully overshooting the number of factors. For the MHRM algorithm, even given the true number of factors as $4$, its estimated loading matrix remained structureless.

To fit our customized sparse Gibbs sampler, we lower triangularized the factor loading matrix for identification purpose \cite{lower_trig}. While the Gibbs sampler identified four prominent latent factors with stripe-like structure, it did not recover the substantial overlapping components among the factors, and was incapable of outputing exact zero values. For the JML estimation method, although it converges significantly faster compared to the other approaches, it only identified three factors and failed to recover the overlapping components.

This simple motivating example effectively highlights the reasoning behind our endorsement of the Bayesian EM approach to MIRT over alternatives.

\vspace{-0.7cm}
\section{Nonparametric Bayesian MIRT} \label{sec:model}
\vspace{-0.3cm}

We draw inspiration from Durante's recent work \cite{Durante_2019}, which characterizes the posterior distribution of coefficient parameters in a standard probit regression. This insight guides us towards a standard normal cumulative distribution link function $\Phi(.)$ to model item response data. For simplicity, we start by considering the binary case: 
 \begin{equation}\label{eqn:2.3}
     Y_{ij} | \theta_i, B_j, d_j \sim \text{Bernoulli}(\Phi(B_j' \theta_i + d_j)). 
\end{equation}

Building upon the Bayesian factor analysis framework introduced by Ro\v{c}kov\'{a} and George \cite{bfa}, we deploy the spike-and-slab LASSO (SSL) prior \cite{SSL} on the loading matrix to prompt sparsity. The sparsity assumption is essential, as it alleviates model identification issues by anchoring estimation on sparse solutions, leading to more interpretable outcomes. Additionally, the SSL prior reflects many real-world applications where factor loadings are inherently sparse, and therefore has the potential to improve out-of-sample prediction. Our SSL approach employs a two-point mixture of Laplace distributions which can adaptively produce exact zeros, and does not necessitate an additional thresholding step. For each element $B_{jk}$ in the loading matrix, we couple the SSL prior with a binary variable selection variable $\gamma_{jk}\in\{0,1\}$. Let $\lambda_{0k} >> \lambda_1 > 0$, and let $\psi(\beta | \lambda) = \frac{\lambda}{2} \exp\{- \lambda |\beta|\}$ denote the Laplace prior with mean $0$ and variance $\frac{2}{\lambda^2}$. We consider the following hierarchical prior:
    \begin{equation*}
     \pi(B_{jk} | \gamma_{jk}, \lambda_{0k}, \lambda_1) \sim (1-\gamma_{jk})\psi(B_{jk} | \lambda_{0k}) + \gamma_{jk}\psi(B_{jk} | \lambda_{1}). 
    \end{equation*}

 To avoid committing to a fixed number of factors, we deploy the infinite-dimensional Indian Buffet Process (IBP) prior with intensity parameter $\alpha > 0$ on the sparsity pattern of the loading matrix $B$. A larger $\alpha$ value would favor a less sparse loading matrix estimation. This prior serves as variable selection tool and accounts for uncertainty in the latent space dimensionality.  As observed in \cite{bfa}, posterior simulation can be greatly facilitated if we consider the stick-breaking representation of IBP \cite{stick_breaking}. Consequently, researchers simply need to choose a large enough truncation-level $K^*$ as an upper bound for the number of latent factors. In the case where estimated loading matrix contains no sparse columns, one should consider increasing $K^*$ to avoid underestimation of the true number of factors. For each item $j= 1, \cdots, J$, and each latent dimension $k= 1, \cdots, K^{*}$, the stick-breaking representation of our IBP prior can be expressed as follows:
\begin{equation*}
v_l \overset{\text{i.i.d}}{\sim} \text{Beta}(\alpha, 1), \text{ } c_k = \prod_{l=1}^{k} v_l, \text{ } \gamma_{jk} | c_k \sim \text{Bernoulli}(c_k). 
\end{equation*}

The data generation process is finalized upon specifying priors for both the latent factor $\theta_i$ and the intercept $d_j$. In the context of exploratory factor analysis, it is customary to assume a prior of zero-mean multivariate Gaussian with an identity covariance matrix for the latent factors, since our goal is to unveil the latent structure of factor loadings. We further consider a Laplacian prior with small parameter $\lambda_1$ for the intercept parameter $d_j$. The combination of SSL and IBP priors on the loading matrix enhances the flexibility and interpretability of our model, which also has exhibited notable success in the nonparametric linear factor model settings \cite{bfa}.

\vspace{-0.5cm}
\subsection{Probit Identifiability}
\vspace{-0.1cm}

One of the key advantages of adopting a probit link $\Phi(.)$ is its ability to simplify the establishment of identifiability results in categorical factor analysis. In contrast, the convolution of Gaussian and logistic random variables induced by a logistic link often renders many classic proof techniques within linear factor analysis intractable. As observed in \cite{pf-ident}, a necessary and sufficient condition to check the identifiability of a probit factor model is to examine the tetrachoric correlations as follows:

\begin{proposition}[\cite{pf-ident}] \label{prop:id}
Let $A, B \in \mathbbm R^{J \times K}$ be the factor loading matrices, $D_1, D_2 \in \mathbbm R^J$ be the intercepts, and the latent factor $\theta$ be the zero-mean multivariate Gaussian with covariance matrix $\Sigma$. Then two sets of parameters $(A, D_1)$ and $(B, D_2)$ define the same probit factor model if and only if their thresholds and tetrachoric correlations are equal:
\begin{equation} \label{eq:id1}
    \frac{D_{1j}}{\sqrt{A_j' \Sigma A_j+1}}= \frac{D_{2j}}{\sqrt{B_j' \Sigma B_j+1}} \quad  \forall j,
\end{equation}
\begin{equation} \label{eq:id2}
    \frac{A_{j_1}' \Sigma A_{j_2}}{\sqrt{A_{j_1}' \Sigma A_{j_1}+1} \sqrt{A_{j_2}' \Sigma A_{j_2}+1}}= \frac{B_{j_1}' \Sigma B_{j_2}}{\sqrt{B_{j_1}' \Sigma B_{j_1}+1} \sqrt{B_{j_2}' \Sigma B_{j_2}+1}} \quad  \forall j_1 \neq j_2 .
\end{equation}
\end{proposition}

By Proposition \ref{prop:id}, establishing identifiability of a given probit factor model is equivalent to showing that equations \ref{eq:id1} and \ref{eq:id2} induce a unique set of model parameters $(B, D)$. Recall in a linear setting, one needs to show $AA'+\Sigma_1 = BB' + \Sigma_2$ yields the unique set of model parameters $(B, \Sigma)$, where $\Sigma$ is a diagonal matrix of variances. Though linear identifiability results are not directly applicable under the probit link, the presence of a comparable equation structure suggests the possibility of adapting classical linear identifiability principles to the probit factor model setting. In particular, \cite{pf-ident} leverages Proposition \ref{prop:id} to establish identifiability conditions for confirmatory probit bifactor and two-tier models. 

In the exploratory linear factor analysis setting, the most common strategy to resolve the rotational invariance identifiability issue is to constrain the loading matrix to possess of a positive lower triangularized (PLT) form \cite{d66cb2ac-a36f-336e-bfae-7af42fd7021f, 10.1093/oso/9780198526155.003.0053}. More recently, \cite{Fr_hwirth_Schnatter_2023} establishes linear factor model identifiability results with generalized lower triangular (GLT) structure.

\begin{definition}
  For factor loading matrix $B \in \mathbbm R^{J \times K}$ with full column rank $K$ and columns $k \in \{1, \cdots K\}$, we define the \textbf{pivot row} $l_k$ as the row index of the first non-zero factor loading in column $k$, such that $B_{i, k} = 0$ for all $i < l_k$.  We say $B$ has an \textbf{ordered GLT} structure if $l_1 < \cdots < l_K$, and $B_{l_k, k} > 0$ for all $k$.
\end{definition}

First observe that the popular positive lower triangular loading constraint is an instance of the GLT structure. As suggested in \cite{Anderson1956StatisticalII, Fr_hwirth_Schnatter_2023}, while PLT and GLT structures are effective in resolving rotational indeterminacy, they alone are insufficient for identifying linear factor models unless the variance matrix $\Sigma$ is also precisely determined. For our proposed probit MIRT model, we first provide a sufficient identification condition for the intercept $D$, analogous to Theorem 5.1 of Anderson and Rubin \cite{Anderson1956StatisticalII}, and then establish similar sufficient identifiability condition for the loading matrix $B$ under the GLT structures, provided that the intercept $D$ is uniquely identifiable. The proof is available in Appendix B.1.

\begin{theorem} \label{thm:identification}
 Consider a probit factor model with a full column rank factor loading matrix $B =[B_1, \cdots, B_J]' \in \mathbbm R^{J \times K}$ and intercept $D \in \mathbbm R^J$. Define matrix $\tilde{B} := [\tilde{B_1}, \cdots, \tilde{B_J}]'$, where $\tilde{B_{j}} = \frac{B_j}{\sqrt{\|B_j\|^2 + 1}}$. Then:
 \vspace{-0.1cm}
 \begin{enumerate}[label=\Alph*.]
    \item A sufficient condition for identification of $D$ is that if any row of $\tilde{B}$ is deleted, there remain two disjoint submatrices of rank $K$.
    \vspace{-0.35cm}
    \item An ordered GLT structure is uniquely identified, provided that the intercept $D$ is identified. 
 \end{enumerate}
\end{theorem}
\vspace{-0.2cm}
As noted in \cite{Fr_hwirth_Schnatter_2023}, the identifiability conditions for linear factor models usually involve two stages \cite{Anderson1956StatisticalII} : the identification of factor loading matrix (hindered by rotational indeterminacy) and the identification of variance, though the majority of statistics literature tends to focus on the former case. Theorem \ref{thm:identification} hints the proofs of probit model identifiability results are likely to follow a similar process as in a linear factor setting, where variance identification is replaced by the intercept identification. 

Once we know the factor loading matrix with the GLT structure is uniquely identified, a natural question is whether we can achieve posterior consistency of $B$ under such a  structural assumption. Notably, recent work by \cite{Ma2020OnPC} advocates for a ``$\sqrt{n}$-orthonormal factor model" approach under the SSL-IBP prior linear factor model setup, and establishes consistency of $B$ up to sign permutation, condition on the feature allocation matrix $\Gamma$ possessing the GLT structure. Their approach can  also help address the "magnitude inflation problem" in high-dimensional setting, effectively mitigating sensitivity to the prior. 

\vspace{-0.5cm}
\subsection{Posterior Asymptotics}
We provide robust theoretical justification for our proposed nonparametric MIRT model, demonstrating its capability to recover the true number of factors and achieving the desired posterior contraction rate for the factor loading matrix. The effective dimension of the loading matrix $K(B)$ under the SSL-IBP($\{\lambda_{0k} , \lambda_1, C\}$) prior is defined as the largest column index $K$ such that $|B_{jk}| < \frac{1}{\lambda_{0k} - \lambda_1} \log [\frac{\lambda_{0k}(1-c_k)}{\lambda_1 c_k}]$ for any $j \in \{1, \cdots, J\}$ and $k >K$ \cite{bfa}. We denote $B_{0n} \in \mathbbm{R}^{J_n \times \infty}$ as the true loading matrix with $k_{0n}$ nonzero columns, $\Theta_{0n} \in \mathbbm{R}^{n \times \infty}$ as the latent factor matrix, and $\eta_{0ij} = B_{0j}' \theta_{0i}$. For $1\leq k \leq k_{0n}$, let $S_{nk} = \sum_{j=1}^{J_n} \mathbbm{I} (B_{jk} \neq 0)$, $\bar{S}_n = \frac{1}{k_{0n}} \sum_{k=1}^{k_{0n}} S_{nk}$, and assume each column of $B_{0n}$ exhibits $S_n$-sparsity. 

To study the asymptotic properties of the posterior, we make the following assumptions: (A) $n < J_n$ and $\lim_{n \to \infty} \sqrt{\frac{S_n k_{0n} \log^2 J_n}{n}} = 0$; (B)  $\|B_{0n}\|_2 < \sqrt{\bar{S}_n}$, where $\|B_{0n}\|_2$ is the spectral norm of the submatrix formed by the first $k_{0n}$ columns $\|B_{0n, k_{0n}}\|_2$; and (C) $\gamma_n(B_{0n}, \Theta_{0n}) \precsim \log J_n$, where $\phi$ is standard normal PDF, and 
    $$\gamma_n (B_{0n}, \Theta_{0n}) := 1 + \max_{\{i \leq n, j \leq J_n\}} \frac{\phi^2(\eta_{0_{ij}})}{\Phi(\eta_{0ij}) (1-\Phi(\eta_{0ij}))}.$$
Assumptions A and B are similar to the assumptions made in \cite{bfa}, in which A restricts our analysis on the high-dimensional case, and B holds with high probability \cite{bfa}. Assumption $C$ is an additional sparsity assumption on $(B_{0n}, \Theta_{0n})$, analogous to a special case of assumption 1 appeared in \cite{glm_contraction} under the sparse generalized linear model setting. Given that $\frac{\phi^2(\eta_{0_{ij}})}{\Phi(\eta_{0ij}) (1-\Phi(\eta_{0ij}))}$ grows linearly in $\eta_{0_{ij}}$ by the Mills ratio while $\max_{i,j} \eta_{0ij}$ roughly grows in the log scale, Assumption $C$ is relatively mild. 

Theorem \ref{thm:factor-dimensionality} shows the average posterior probability of our exploratory MIRT model overestimating the true number of factors, scaled by the true sparsity level, approaches zero asymptotically. This result is analogous to Theorem 3 in \cite{bfa} within a linear setting.

\begin{theorem} \label{thm:factor-dimensionality}
Under assumptions (A)-(C), consider the sparse Bayesian MIRT model defined in equation (\ref{eqn:2.3}), where the true loading matrix $B_{0n}$ has $k_{0n}$ nonzero leftmost columns. For $1\leq k \leq k_{0n}$, let $\sum_{j=1}^{J_n} \mathbbm{I} (B_{jk} \neq 0) = S_{nk}$ and  $\lambda_{0k} \geq \frac{2J_n^2 k_{0n}^3 n}{S_{nk}}$. Then with the prior choice $B \sim \text{SSL-IBP}(\{\lambda_{0k}\}, \lambda_1, \alpha_n)$, where $\alpha = \frac{1}{J_n}$,  $\lambda_1 < e^{-2}$, and with some constant $C_0 >0$:
\vspace{-0.3cm}
\begin{equation*}
    E_{B_0} P[K(B) > C_0 k_{0n} S_n | Y^{(n)}] \xrightarrow[n \rightarrow \infty]{} 0.
\end{equation*}
\end{theorem}

The corresponding proof can be found in Appendix B.2. In the linear factor model setting \cite{bfa}, one may easily integrate Gaussian factor over the likelihood to study the posterior distribution of the covariance matrix. This approach is no longer possible under the probit link. Consequently, we introduce assumption (C) to control the Kullback-Leibler divergence and its variation between the posterior density and the true density.

The following theorem demonstrates the posterior distribution of the factor loading matrix $B$ contracts towards the true factor loading matrix $B_0$ in terms of the root-average-squared Hellinger distance. For any two loading matrices $B_1$ and $B_2$, recall the average of the squares of the Hellinger distances is defined as $H_n^2(B_1, B_2) = \frac{1}{n} \sum_{i=1}^n H^2(f_{B_{1,i}}, f_{B_2, i})$, where $H^2(f_{B_{1,i}}, f_{B_{2,i}}) = \int (\sqrt{f_{B_1, i}} - \sqrt{f_{B_2,i}})^2 d \nu$ represents the squared Hellinger distance between the densities of the binary vector $Y_i$ induced by $B_1$ and $B_2$ respectively.

\begin{theorem} \label{thm:pos-contraction}
Assuming the true loading matrix $B_{0n}$ has $k_{0n}$ columns, each of which is $S_n$-sparse. Under assumptions (A)-(C), and with the same SSL-IBP prior choice as in Theorem \ref{thm:factor-dimensionality}, these exists a constant $C_0'$ such that
 \begin{equation*}
        E_{B_0} \Pi_{n} \left\{ B: H_n (B, B_0) > C_0' \sqrt{\frac{k_{0n}S_n \log^2 J_n}{n}} | Y^{(n)}  \right\}  \xrightarrow[n \rightarrow \infty]{} 0.
\end{equation*}
\end{theorem}

The Hellinger posterior contraction rate for the regression parameter in generalized linear models with a fixed design was recently shown to be $\sqrt{\frac{S_0 \log J_n}{n}}$, where $S_0$ is the size of the true support \cite{glm_contraction}. Theorem \ref{thm:pos-contraction} provides similar contraction guarantee but with an extra $\sqrt{\log J_n}$ factor, reflecting the cost of the growing number of factors $k_{0n}$ and the support size $S_n$. The proof of Theorem \ref{thm:pos-contraction} is detailed in Appendix B.2. In Appendix E.2, we further empirically verify the consistency of $B$ with a challenging high-dimensional example as described in Section \ref{subsec:ibp_loading}.

\vspace{-0.6cm}

\section{Parameter-Expanded EM for Sparse MIRT} \label{sec:em}
\vspace{-0.3cm}
Our proposed EM algorithm is guided by a recent groundbreaking result from Durante \cite{Durante_2019}, who fully characterized the posterior distribution of coefficients in the standard probit regression model. This result enabled us to explicitly derive the latent factor posterior distributions, resulting in highly efficient E-step sampling. 
\vspace{-0.5cm}
\subsection{Algorithm Setup} \label{subsec:em_setup}
\vspace{-0.2cm}
For simplicity, we let $K$ represent the truncated level of the  IBP prior $K^*$. Additionally, let $\Theta = [\theta_1, \cdots, \theta_N]' \in \mathbbm{R}^{N \times K}$ be the matrix of latent factors, $C = [c_1, \cdots, c_k] \in \mathbb R^{K}$ be the vector of ordered inclusion probabilities, $D = [d_1, \cdots, d_J]' \in \mathbb R^{J}$ be the intercept parameter,  $B \in \mathbbm R^{J \times K}$ be the factor loading matrix, and $\Gamma \in \mathbb R^{J \times K}$ be the binary latent allocation matrix. The goal of our EM algorithm is to estimate the model parameters $\Delta := (B, C, D)$. Given the initialization $\Delta^{(0)}$, the $(m+1)^{st}$ step of the EM algorithm finds $\Delta^{(m+1)}$ by solving the following optimization problem:
\begin{equation*} 
\Delta^{(m+1)} = \text{argmax}_{\Delta} \mathbbm E_{\Gamma, \Theta | Y, \Delta^{(m)}}  [\log \pi(\Delta, \Gamma, \Theta|Y))]  := \text{argmax}_{\Delta} Q(\Delta). 
\end{equation*}
Let $\langle X \rangle$ represent $\mathbbm E_{\Gamma, \Theta| Y, \Delta}(X)$, the conditional expectation given the observed data $Y$ and model parameters $\Delta$. Due to the hierarchical structure of our model, the optimization problem can be simplified by recognizing that the model parameters $(B, D)$ and $C$ are conditionally independent given $(\Gamma, \Theta)$. Thus, we can decompose $Q(\Delta)$ into the sum of $Q_1(B,D)$ and $Q_2(C)$, where
$$Q_1(B,D)  \propto \langle \log \pi(Y | B, D,  \Theta) \rangle + \langle \log \pi(B| \Gamma) \rangle + \langle \log \pi(D) \rangle, \quad Q_2(C)  \propto \log(\langle \Gamma \rangle | C) + \log \pi (C).$$
In the MIRT literature, it is customary to assume that item responses for each individual are conditionally independent given their multidimensional latent traits. This assumption enables us to write $Q_1(B, D)$ as $ \sum_{j=1}^{J} Q_j(B_j, d_j)$, where
\begin{equation}\label{eqn:3.6}
Q_j(B_j,d_j) = \sum_{i=1}^{N} \langle \log \Phi((2Y_{ij}-1)(B_j' \theta_i +d_j)) \rangle - \sum_{k=1}^{K} |B_{jk}|(\lambda_1  \langle \gamma_{jk} \rangle  + \lambda_0(1- \langle  \gamma_{jk} \rangle )) - \lambda_1 |d_j|.  \\
\end{equation}
This highlights the primary components of our EM algorithm. In the E-step, we compute the conditional expectations $ \langle \log \Phi((2Y_{ij}-1)(B_j' \theta_i +d_j)) \rangle $ and $\langle \gamma_{jk} \rangle$. In the M-step, we separately maximize $Q_1(B,D)$ and $Q_2(C)$. Efficient maximization of $Q_1(B, D)$ is achievable by decomposing it into $J$ distinct objective functions, allowing for parallel optimization.
\vspace{-0.5cm}
\subsection{The E-Step} \label{subsec:e-step}
\vspace{-0.2cm}
The complexity of the above-presented EM algorithm primarily stems from the computation of the conditional expectation term $\langle \log \Phi((2Y_{ij}-1)(B_j' \theta_i +d_j)) \rangle$, given the challenge of identifying latent factor posterior distributions. In Bayesian linear factor model \cite{bfa}, one can interchange the expectation and the linear functional form, thereby enabling the computation of the posterior distribution $\langle \theta_i \rangle$ through a form of ridge regression. However, in our MIRT setting, relocating the expectation $\langle . \rangle$ within the $\log \Phi(.)$ function would lead to maximizing an upper bound, constrained by the Jensen's inequality, and would not be valid even in an approximation sense. 
\vspace{-0.3cm}
\subsubsection{Characterize Posterior Distribution of Latent Factors}
To compute $\langle \log \Phi((2Y_{ij}-1)(B_j' \theta_i +d_j)) \rangle$, we propose to directly sample from the posterior distribution of $\theta_i$ and perform Monte Carlo integration. This strategy can yield significantly enhanced computational speed when compared to MCEM-type algorithms relying on MCMC sampling. The advantage of our proposed E-step becomes particularly prominent when dealing with a large number of Monte Carlo samples in high latent dimensions, as MCMC sampling is not amenable to parallelization, thereby resulting in increased computational costs.  The efficient sampling of latent factors is facilitated by a recent result by Durante \cite{Durante_2019}, who derives the posterior distribution of coefficient parameters in a probit regression as the unified skew-normal distribution \cite{unified_skew_normal}, which can be sampled efficiently:

\begin{definition} \label{def:usn}
Let $\Phi_{J} \left\{V; \Sigma \right\}$ represent the cumulative distribution functions of a J-dimensional multivariate Gaussian distribution $N_{J}\left(0_{J}, \Sigma \right)$ evaluated at vector $V$. A K-dimensional random vector $z \sim \operatorname{SUN}_{K, J}(\xi, \Omega, \Delta, \gamma, \Gamma)$ has the unified skew-normal distribution if it has the probability density function:
$$
\phi_{K}(z-\xi ; \Omega) \frac{\Phi_{J}\left\{\gamma+\Delta^{\mathrm{T}} \bar{\Omega}^{-1} \omega^{-1}(z-\xi) ; \Gamma-\Delta^{\mathrm{T}} \bar{\Omega}^{-1} \Delta\right\}}{\Phi_{J}(\gamma ; \Gamma)}.
$$
Here,  $\phi_{K}(z-\xi ; \Omega)$ is the density of a $K$-dimensional multivariate Gaussian with expectation $\xi=\left(\xi_{1}, \ldots, \xi_{K}\right)' $, and a K by K covariance matrix $\Omega =\omega \bar{\Omega} \omega$, where $\bar{\Omega}$ is the correlation matrix and $\omega$ is a diagonal matrix with the square roots of the diagonal elements of $\Omega$ in its diagonal.  $\Delta$ is a $K$ by $J$ matrix that determines the skewness of the distribution, and $\gamma \in \mathbbm R^J$ control the flexibility in departures from normality.

In addition, the $(K+J) \times(K+J)$ matrix $\Omega^{*}$, having blocks $\Omega_{[11]}^{*}=\Gamma, \Omega_{[22]}^{*}=\bar{\Omega}$ and $\Omega_{[21]}^{*}=\Omega_{[12]}^{*'}=\Delta$, needs to be a full-rank correlation matrix.
\end{definition}

 In the context of our EM algorithm, a key difference from \cite{Durante_2019}  rises in that the intercept $D$ is held fixed during the E-step, and each $d_j$ is distinct for every item. Fortunately, it remains feasible to deduce the posterior distribution of $\theta_i$ as an instance of the unified skew normal distribution: 

\begin{theorem}\label{thm:3.2}
Suppose the factor loading matrix $B$ and the intercept $D$ are known. If $y=\left(y_{1}, \ldots, y_{J}\right)' \in \mathbbm R^J $ is conditionally independent binary response data from the probit MIRT model defined in equation (\ref{eqn:2.3}), and $\theta_i \sim N(\xi, \Omega)$ assumed to be $K$-dimensional Gaussian prior, then
\vspace{-0.2cm}
$$(\theta_i \mid y, B, D) \sim \operatorname{SUN}_{K, J}\left(\xi_{\mathrm{post}}, \Omega_{\text {post }}, \Delta_{\text {post }}, \gamma_{\text {post }}, \Gamma_{\text {post }}\right),$$
with posterior parameters
\vspace{-0.3cm}
\begin{gather*}
  	\xi_{\mathrm{post}} = \xi, \quad \Omega_{\text {post }}= \Omega,  \quad \Delta_{\text {post }}= \bar{\Omega} \omega D_1^{\mathrm{T}} S^{-1}, \\
	\gamma_{\text {post }} = S^{-1}(D_1 \xi + D_2), \quad \Gamma_{\text {post }} =S^{-1}\left(D_1 \Omega D_1^{\mathrm{T}}+\mathbbm I_{J}\right) S^{-1},
   \end{gather*}
where $D_1 = \text{diag}(2y_1-1, \cdots, 2y_J-1) B$ with $j$-th row as $D_{1j}$,  $D_2 = \text{diag}(2y_1-1, \cdots, 2y_J-1)D \in \mathbbm R^{J}$,  and $S = \text{diag}\{(D_{11}' \Omega D_{11} +1)^{\frac{1}{2}}, \cdots,  (D_{1J}' \Omega D_{1J} +1)^{\frac{1}{2}} \} \in \mathbbm R^{J \times J}.$ 
\end{theorem}

Theorem \ref{thm:3.2} is useful as we can now entirely characterize the latent factor posterior distribution. The significance of it goes beyond leading to a tractable E-step in our proposed algorithm, as it has direct applications to psychological measurements: given any $J$ well-calibrated items and for any new observation $y \in \mathbbm R^J$, Theorem \ref{thm:3.2} immediately provides an analytic form of the latent trait posterior distribution. In consequence, no extra Newton step or MCMC simulation are needed for latent factor inference. 

Given the latent factor posterior distributions, the E-step computation becomes straightforward via Monte Carlo integration. As the posterior distribution of $\theta_i$ is not identical to that of a probit regression, we can instead consider the following sampling strategy:
\begin{corollary} \label{cor:4.3}
If $\theta_i | y, B, D$ has the unified skew-normal distribution from Theorem \ref{thm:3.2} , then
\vspace{-0.8cm}
$$
(\theta_i | y, B, D) \stackrel{\mathrm{d}}{=} \xi+\omega\left\{V_{0}+\bar{\Omega} \omega D_1' \left(D_1 \Omega D_1'+ \mathbbm{I_{J}}\right)^{-1} S V_{1}\right\}, \quad\left(V_{0} \perp V_{1}\right),
$$
with $V_{0} \sim N \left\{0, \bar{\Omega}-\bar{\Omega} \omega D_1'\left(D_1 \Omega D_1'+\mathbbm{I_{J}}\right)^{-1} D_1 \omega \bar{\Omega}\right\}$, and $V_{1}$ from a zero mean $J$-variate truncated normal with covariance matrix $S^{-1}\left(D_1 \Omega D_1'+\mathbbm{I_{J}}\right) S^{-1}$ and truncation below \\$-S^{-1}(D_1 \xi +D_2)$. 
\end{corollary}
The proofs of Theorem \ref{thm:3.2} and Corollary \ref{cor:4.3} can be found in Appendix C. 
\subsubsection{E-step Computation} \label{subsec:posterior_sampling}
Corollary \ref{cor:4.3}  provides us with a direct approach to assess $\langle \log \Phi((2Y_{ij}-1)(B_j' \theta_i +d_j)) \rangle$. This involves sampling a sufficiently large number of $\theta_i$ values from its posterior distribution and subsequently performing Monte Carlo integration. Assuming we require $M$ samples for each latent factor $\theta_i$, we can generate individual samples $\theta_i^{(m)}, m = 1, \cdots, M$, by following these three steps:
\vspace{-0.1cm}
\begin{itemize}
    \item Step 1: Sample $V_{0}^{(m)} \sim N \left\{0, \bar{\Omega}-\bar{\Omega} \omega D_1'\left(D_1 \Omega D_1'+\mathbbm{I_{J}}\right)^{-1} D_1 \omega \bar{\Omega}\right\} \in \mathbbm R^{K}$. 
    \item Step 2: Leveraging the  minimax tilting methods of Botev \cite{Botev_2016}, sample $V_{1}^{(m)}$ from a zero mean $J$-variate truncated normal with covariance matrix $S^{-1}\left(D_1 \Omega D_1'+\mathbbm{I_{J}}\right) S^{-1}$ and truncation below $-S^{-1}(D_1 \xi +D_2)$. 
    \item Step 3: Compute $\theta_i^{(m)} = \xi+\omega\left\{V_{0}^{(m)}+\bar{\Omega} \omega D_1' \left(D_1 \Omega D_1'+\mathbbm{I_{J}}\right)^{-1} S V_{1}^{(m)}\right\}.$
\end{itemize}

Finally, we need to compute the conditional distribution $\langle \gamma_{jk} \rangle$ as
$$\langle \gamma_{jk} \rangle = P(\gamma_{jk} =1 | \Delta) = \frac{c_k \psi (B_{jk} | \lambda_1)}{c_k \psi (B_{jk} | \lambda_1) + (1-c_k) \psi(B_{jk} | \lambda_0)}.$$

\vspace{-0.3cm}
\subsection{The M-step} \label{subsec:m-step}
\vspace{-0.2cm}

As discussed in Section \ref{subsec:em_setup}, the M-Step of our algorithm can be decomposed into two independent optimization problems. To maximize $Q_1(B,D)$, we can maximize each $Q_j(B_j, d_j)$ as defined in \eqref{eqn:3.6} in parallel. Suppose we obtain $M$ samples of $(\theta_i^{(1)}, \cdots, \theta_i^{(M)})$ from the E-Step for each subject $i$, we can maximize the following surrogate objective function:
\begin{equation*}
 \frac{1}{M} \sum_{i=1}^{N} \sum_{m=1}^{M} \log \Phi((2Y_{ij}-1) (B_j' \theta_i^{(m)} +d_j)) - \sum_{k=1}^{K} |B_{jk}|(\lambda_1  \langle \gamma_{jk} \rangle  + \lambda_0(1- \langle  \gamma_{jk} \rangle )) -\lambda_1 |d_j|. 
\end{equation*}
Hence maximizing $Q_j(B_j,d_j)$ is equivalent to estimating a penalized probit regression with an intercept. Specifically, this regression involves $N \times M$ response variables, which result from replicating the $j$-th column of response $Y_j$ for $M$ times. The corresponding design matrix consists of $M$ identical copies of the latent factor matrix $\Theta$. For each latent dimension $k$, we apply a distinct $l_1$ penalty to $B_{jk}$, with a weight of $(\lambda_1 \langle \gamma_{jk} \rangle + \lambda_0(1- \langle \gamma_{jk} \rangle ))$. Solving this penalized regression can be efficiently achieved using standard off-the-shelf packages, such as the 'glmnet' package \cite{glmnet} in R.

To maximize $Q_2(C)$, we need to solve the following constraint linear program:
\begin{equation*}
\begin{aligned}
\max_{C} \quad & \sum_{j=1}^{J} \sum_{k=1}^{K} [ \langle \gamma_{jk} \rangle \log c_k + (1- \langle \gamma_{jk} \rangle) \log(1-c_k)] + (\alpha-1) \log(c_{K})  \\
\textrm{s.t.} \quad & c_k - c_{k-1} \leq 0, k = 2, \cdots, K,\\
  &0 \leq c_k \leq 1, k =1, \cdots, K .   \\
\end{aligned} 
\end{equation*}
Note the optimization problem is convex, and hence can be solved with standard convex optimization package. As observed in \cite{bfa}, the optimization process can be accelerated if we sort the factors by their binary order before optimizing the problem.
\vspace{-0.3cm}
\subsection{The PXL-EM Algorithm}
We present the parameter-expanded version of our EM algorithm (PXL-EM) to expedite convergence and enhance exploration of the parameter space. The updates for our PXL-EM algorithm are almost as identical to vanilla EM algorithm described above, except we now need an additional auxiliary parameter $A$, which serves to rotate the loading matrix $B$ towards a sparser configuration (PXL-EM). A brief summary of our PXL-EM algorithm can be found in the table of Algorithm \ref{alg:pxl-em}. The derivation and the justification of the factor rotation step is provided in Appendix D.

\RestyleAlgo{ruled}
\SetKwComment{Comment}{/* }{ */}
\begin{algorithm}[hbt!]
\caption{PXL-EM Algorithm for Sparse MIRT}\label{alg:pxl-em}
\KwData{Item Response $Y \in \mathbbm{R}^{N \times J}$}
Initialize model parameters $B= B^{(0)}, C= C^{(0)}, D= D^{(0)}$ ; \\
\While{\text{PXL-EM has not converged}}{
    \tcc{E-Step}
    \For{$i=1$ \KwTo $N$}{  
        Draw $\theta_i^{(1)}, \cdots, \theta_i^{(M)}$ as described in Section \ref{subsec:posterior_sampling} \Comment*[r]{Update Latent Factors}
     }
     
    $\langle \gamma_{jk} \rangle = \frac{c_k \psi (B_{jk} | \lambda_1)}{c_k \psi (B_{jk} | \lambda_1) + (1-c_k) \psi(B_{jk} | \lambda_0)}$ \Comment*[r]{Update Latent Indicators} 
    \tcc{M-Step}
    \For{$j=1$ \KwTo $J$}{  
        $B_j, d_j = \argmax Q_j(B_j, d_j)$, as described in Section \ref{subsec:m-step} \Comment*[r]{Update Loadings}
     }
    $C= \argmax Q_2(C)$ as described in Section \ref{subsec:m-step} \Comment*[r]{Update IBP Weights} 
    \tcc{Rotation-Step}
    $A = \frac{1}{N\cdot M} \sum_{i=1}^{N} \sum_{m=1}^{M}  \theta_{i}^{(m)} \theta_i^{(m)'}$ ;
    $B = B * A_{L}$, where $A_{l}A_{L}'=A $; 
}
\end{algorithm}
\vspace{-0.6cm}
\section{Extension to Mixed Data Types} \label{sec:grm_ordinal}
\vspace{-0.2cm}
We demonstrate the extension of our proposed methodology, illustrated for binary datasets, to item response with mixed data type. For each observation $Y_i \in \mathbbm R^J$, we allow $Y_{ij}$ to be categorized as continuous, binary, or ordinal. Specifically, let $J_C$ be the set of continuous responses, $J_B$ be the set of binary responses, $J_{O}$ be the set of ordinal responses, then we have $J = J_B \cup J_C \cup J_O$. let $B_{S} \in \mathbbm R^{|S| \times K}$ be the sub matrix of $B$ indexed by set $S$. Assuming the binary responses $j \in J_B$ are generated by the standard MIRT model as defined in equation (\ref{eqn:2.3}), we consider the following data generating process for the continuous and ordinal items:
\begin{itemize}
    \item For each continuous feature $j \in J_c$, we assume it has been centered, and assign a noninformative inverse gamma prior on its unknown variance $\sigma_{j}^2 \sim \text{IG}(\eta_1, \eta_2)$. Then 
        \begin{equation} \label{eq:y-c}
        Y_{i} | \theta_i, B_{J_c}, \Sigma \sim N(B_{J_c} \theta_i, \Sigma), \quad \Sigma:=\text{diag}(\{\sigma_j^2 \}_{j \in J_c}) .
        \end{equation}
    \item For each ordinal feature $j$ that takes values in $\{0, 1, \cdots, L_j\}$, we assume it has been generated from the standard graded response model \cite{grmodel}:
    \begin{equation} \label{eq:y-o}
     P(Y_{ij}= l | \theta_i, B_j, d_j) = P(Y_{ij} \leq l) - P(Y_{ij} \leq l-1) = \Phi(B_j'\theta_i +d_{j,l}) -  \Phi(B_j'\theta_i +d_{j,l-1}),
    \end{equation}
     where we need to estimate $L_j$-dimensional intercept $d_j$ such that $- \infty = d_{j,-1} < d_{j, 0} \leq \cdots \leq d_{j, L_j-1} < d_{j,L_j} = \infty$. 
\end{itemize}

To derive a tractable E-step, we will show the latent factor posterior distribution still belongs to an instance of $\operatorname{SUN}$ under the mixed item response types. For ease of notation, we fix an observation $i$ and write $Y:= Y_i \in \mathbbm R^{J}$. Define $J_{O_1}$ as the set of ordinal features for which $y_{ij} \notin \{0, L_j\}$, and $J_{O_2} = J_O \setminus J_{O_1}$. Let $D_{J_B}$ and $Y_{J_B}$ denote the sub vectors of $D$ and $Y$ indexed by $J_B$ respectively We further define the following quantities:
\vspace{-0.2cm}
\begin{itemize}
    \item For any n-dimensional vector $v$, let $\{v=0\}$ denote the n-dimensional binary vector, in which the $i$-th component is $1$ iff $v_i=0$. Let
    \vspace{-0.3cm}
    \begin{align*}
        &D_1 = \text{diag}(2Y_{J_B}-\mathbbm{1}_{J_B}) B_{J_B}, \quad D_2 = \text{diag}(2Y_{J_B}-\mathbbm{1}_{J_B}) D_{J_B}, \\
        &D_3 = \text{diag}(2\{Y_{J_{O_2}} = 0\} - \mathbbm{1}_{J_{O_2}}) B_{J_{O_2}}, \quad D_4 = \text{diag}(2\{Y_{J_{O_2}} = 0\} - \mathbbm{1}_{J_{O_2}}) D_{J_{O_2}},
    \end{align*}
    where $D_{J_{O_2}}$ has its $j$-th component as $d_{j,0}$ when $y_j=0$, and as $d_{j, L_{j}-1}$ when $y_j = L_j$. 
    
    \item Let $B_{J_{O_1}}^{(i)}$ be the $i$-th row vector of $B_{J_{O_1}}$,  and $D_{j}^{(i)}$ be the $i$-th component of the intercept vector of the $j$-th feature in $J_{O_1}$. Define 
    \vspace{-0.3cm}
    \begin{align*}
        \tilde B_{J_{O_1}} &= [B_{J_{O_1}}^{(1)'}, -B_{J_{O_1}}^{(1)'}, \cdots, B_{J_{O_1}}^{(|J_{O_1}|)'}, -B_{J_{O_1}}^{(|J_{O_1}|)'}]' \in \mathbbm R^{2|J_{O_1}| \times K}, \\
        \tilde D_{J_{O_1}} &= [D_{1}^{(y_{j})}, -D_{1}^{(y_{j}-1)} \cdots,D_{|J_{O_1}|}^{(y_{j})}, -D_{|J_{O_1}|}^{(y_{j}-1)} ] \in \mathbbm R^{2|J_{O_1}|}.
    \end{align*}
    \item Let $\overline{\mathbbm{I_{J_{O_1}}}}$ be a $2|J_{O_1}|$ by $2|J_{O_1}|$ block diagonal matrix, each block of which is the $2$ by $2$ matrix $-\mathbbm{1}_2 \mathbbm{1}_2' + 2\mathbbm{I_2}$.
\end{itemize}
\vspace{-0.3cm}

We are now ready to characterize the conditional posterior distribution for each latent factor $\theta_i$ as an instance of the unified skew-normal distribution: 
\vspace{-0.1cm}
\begin{theorem}\label{thm:mixed}
Suppose the model parameters $(B, D, \Sigma)$ are known. Let $y=\left(y_{1}, \ldots, y_{J}\right)' \in \mathbbm R^J $ be a vector of conditionally independent response with mixed data types, each component of which is generated through either equations \ref{eqn:2.3}, \ref{eq:y-c}, or \ref{eq:y-o}.  Assuming latent factor prior $\theta_i \sim N(\xi, \Omega) \in \mathbbm{R}^{K}$. Then
\vspace{-0.3cm}
\begin{equation*}
    (\theta_i \mid y, B, D, \Sigma) \sim \operatorname{SUN}_{K, \bar J}\left(\xi_{\mathrm{post}}, \Omega_{\text {post }}, \Delta_{\text {post }}, \gamma_{\text {post }}, \Gamma_{\text {post }}\right)
\end{equation*}
\vspace{-0.2cm}
with posterior parameters
\begin{align*}
 &\xi_{\mathrm{post}} = \Omega_{\text {post }}[\tilde{B}_{J_c}' \tilde{B}_{J_c} \hat{\theta_i} + \Omega ^{-1} \xi], \quad \Omega_{\text {post }}= (\tilde{B}_{J_c}' \tilde{B}_{J_c} + \Omega^{-1})^{-1},  \quad \Delta_{\text {post }}= \bar{\Omega}_{\text {post }} \omega_{\text {post }} \bar{D}^{'} S^{-1},\\
 & \hspace{6em} \gamma_{\text {post }} = S^{-1}(\bar{D} \xi_{\text {post}} + \bar{V}), \quad \Gamma_{\text {post }} =S^{-1}\left(\bar{D} \Omega_{\text {post }} \bar{D}^{'}+\overline{\mathbbm{I_{\bar{J}}}}\right) S^{-1},
\end{align*}
where $\tilde{B}_{j_c} = \Sigma^{-\frac{1}{2}}B_{J_c}$, $\hat{\theta_i} = (\tilde{B}_{j_c}' \tilde{B}_{j_c})^{-1} \tilde{B}_{j_c}' \Sigma^{-\frac{1}{2}} y_{\raisebox{-0.5mm}{$\scriptscriptstyle J_c$}}$, $\bar J = |J_B| + 2|J_{O_1}| + |J_{O_2}|$,  $\bar{D} = [D_1', D_3', \tilde B_{J_{O_1}}']' \in \mathbbm{R}^{\bar{J} \times K}$ with $\overline{D_{j}} \in \mathbbm R^{k}$ as its $j$-th row, 
\vspace{-0.2cm}
$$S = \text{diag}\{(\overline{D_{1}}' \Omega_{\text {post }} \overline{D_{1}} +1)^{\frac{1}{2}}, \cdots,  (\overline{D_{\bar{J}}}' \Omega_{\text {post }} \overline{D_{\bar{J}}} +1)^{\frac{1}{2}} \} \in \mathbbm R^{\bar{J} \times \bar{J}}, $$, $\bar{V} = [ D_2', D_4',  \tilde{D}_{J_{O_1}}']' $,  and $\overline{\mathbbm{I_{\bar{J}}}}$ is a $\bar{J}$ by $\bar{J}$ diagonal block matrix, in which $\overline{\mathbbm{I_{\bar{J}}}}_{[1,1]} = \mathbbm{I_{|J_B|+|J_{O_2}|}}$ and $\overline{\mathbbm{I_{\bar{J}}}}_{[2,2]} = \overline{\mathbbm{I_{J_{O_1}}}}$.
\end{theorem}

Theorem \ref{thm:mixed} is particularly valuable because it demonstrates that sampling from the latent factor posterior distribution remains tractable under mixed item response data types, as it still conforms to an instance of the unified skew-normal distribution. This ensures that our methodology can be generalized to handle a mix of continuous, binary, and ordinal item response types. The proof of Theorem \ref{thm:mixed}, and the comprehensive EM procedure are provided in Appendix I.
\vspace{-0.7cm}
\section{Experiments} \label{sec:empirical}
\vspace{-0.5cm}
\subsection{Dynamic Posterior Exploration of EM} \label{subsec:calibration}
\vspace{-0.2cm}

A common critique of the Bayesian approach centers on the perceived arbitrariness in prior specification and the challenges in prior calibration. Here we introduce a systematic methodology to calibrate the priors in our EM algorithm, aiming to achieve optimal outcomes. The calibration of penalty parameters $\lambda_0$ and $\lambda_1$ of the SSL prior is of paramount importance for effectively exploring the posterior landscape. In \cite{bfa}, the authors adopt the idea of deterministic annealing \cite{UEDA1998271, JMLR:v11:yoshida10a} and propose to keep the value of $\lambda_1$ fixed at a reasonable level, and initialize $\lambda_0$ at a small value to quickly obtain an initial solution. They then iteratively refit the model by gradually increasing the value of $\lambda_0$ and reinitializing the model parameters from the previous iteration. Once the increase of $\lambda_0$ no longer affects the solution, the model has stabilized and its solution is ready for interpretation. We can readily apply this approach of  ``dynamic posterior exploration'' to our PXL-EM algorithm for MIRT models. In Section \ref{subsec:ibp_loading}, we will describe this calibration process of in detail.
\vspace{-0.5cm}
\subsection{IBP Loading Matrix} \label{subsec:ibp_loading}
\vspace{-0.2cm}
We randomly generated the true loading matrix from the Indian Buffet Process (IBP) prior with an intensity parameter $\alpha = 2$, resulting in a complex loading structure as shown in the leftmost subplot of Figure \ref{fig:plot4}. Since a draw from IBP prior would be infinite-dimensional, we applied a stick-breaking approximation to keep the strongest five factors, and reordered the rows and columns of the loading matrix to enhance visualization. We focused on a difficult high-dimensional scenario, with the number of observations $N=250$ being smaller than the number of items $J=350$.

To assess whether PXL-EM can recover the true number of factors, we deliberately overshot the number of factors as $K^*=10$. Keeping $\lambda_1=0.5$ fixed, we gradually increased $\lambda_0$ along the path $\{0.5, 1, 3, 6, 10, 20, 30, 40\}$ following the procedure described in Section \ref{subsec:calibration}. A visualization of the loading matrix estimation along this $\lambda_0$-path can be found in Appendix E.1.  Convergence was determined when $|B^{(t+1)} - B^{(t)}|_{\infty} < 0.06$. Table \ref{tab:ibp_px_em} documents how the estimation accuracy dramatically improved over time.

\begin{table}[t]
    \centering
    \renewcommand{\arraystretch}{1.0}
    \captionsetup{labelfont=bf, format=hang} 
    \caption{Dynamic Posterior Exploration of PXL-EM}
    \begin{tabular}{@{}l S[table-format=1.0] S[table-format=1.3] S[table-format=1.3] S[table-format=1.3] S[table-format=1.3] S[table-format=1.3] S[table-format=4]@{}}
        \toprule
        {$\boldsymbol{\lambda}_0$} & {Loading MSE} & {Intercept MSE}  & {FDR} & {FNR} & {Accumulated Time (s)} \\
        \midrule
        \hline
        1 & 0.078 & 0.035 & 0.666 & 0 & 554\\
        6 & 0.016 & 0.028 & 0.544 & 0  & 1433\\
        20 & 0.011 & 0.022 & 0.103 & 0  & 2173\\
        30 & 0.01 & 0.022  & 0.022 & 0 &  2450\\
        40 & 0.01 & 0.022  & 0.009 & 0 & 2682\\
        \bottomrule
    \end{tabular}
    
    \label{tab:ibp_px_em}
\end{table}

We also considered the MHRM algorithm, the sparse Gibbs sampler, and the JML estimation method as introduced in Section \ref{sec:motivation}. In Figure \ref{fig:plot4}, we observe that the Gibbs sampler was able to recover the general structure of the factor loading, but required further thresholding to output exact zeros. On the other hand, the MHRM algorithm suffered from rotational indeterminacy issues due to the absence of structural assumptions on the sparsity of the loading matrix. The JML estimation method converged amazingly fast compared to the other methods. While it captured the general IBP pattern, it only identified four factors, and missed several major components for each of the recovered factor. 
\begin{figure}[t]
    \centering
    \includegraphics[width=0.85\textwidth, height=0.55\textwidth]{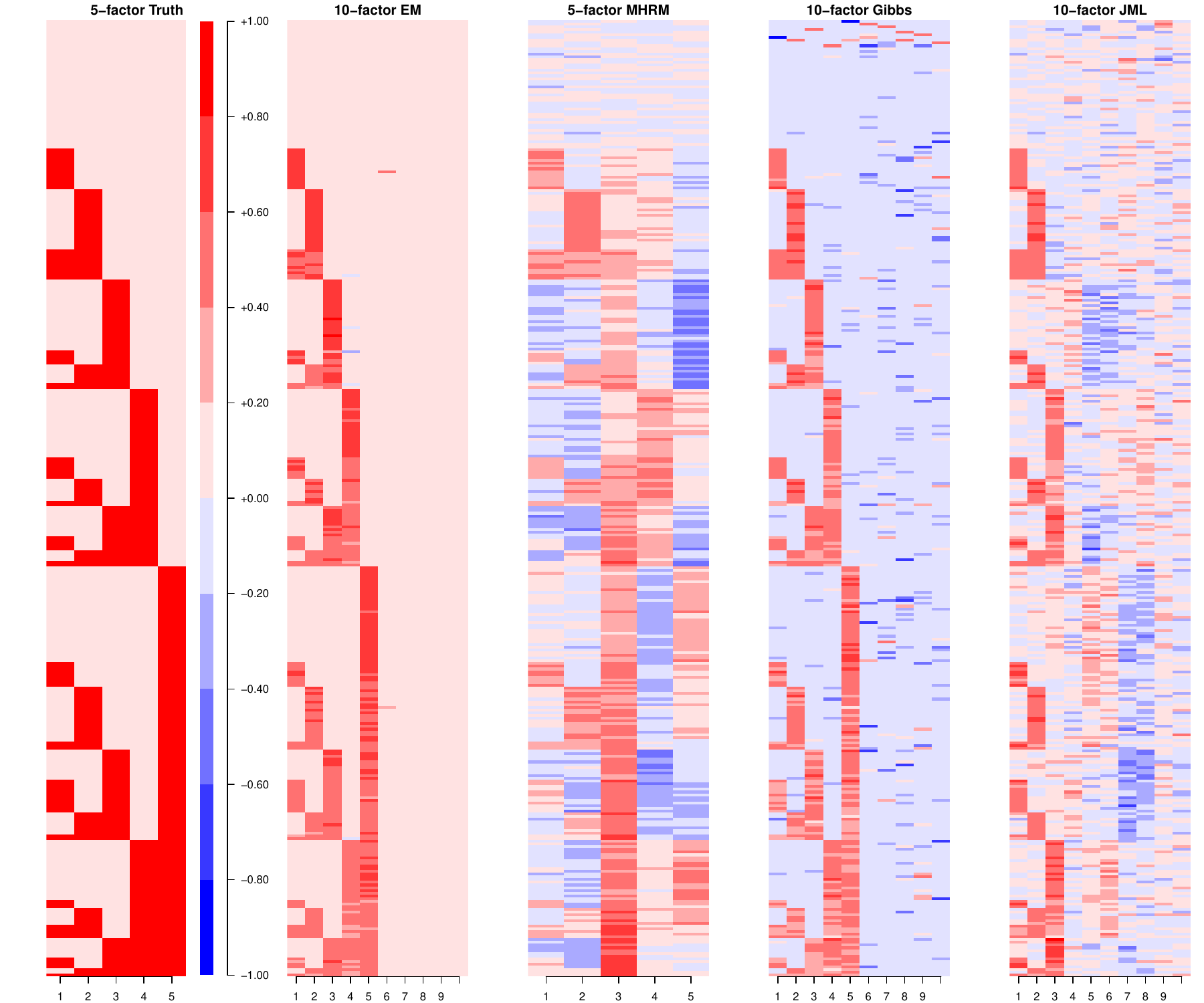}
    \caption{IBP Loading Matrix Estimation: Baseline Models}
    \label{fig:plot4}
\end{figure}

As shown in Table \ref{tab:ibp_models_comparison}, our PXL-EM algorithm achieved the best overall performance in the accuracy of loading estimation, as well as the recovery of true zeros. One might have noticed that PXL-EM was not dramatically fast in this experiment. It is essential to realize our PXL-EM had already attained matching performance when $\lambda_0=6$. The Gibbs sampler would not scale well with either the sample size or the truncated dimension. The availability of more CPUs can always effectively accelerate our EM algorithm, but may not speed up the Gibbs sampler. As the sample size $N$ increases, we can expect the estimation becomes increasingly accurate. This is due to the IBP loading structure possessing a GLT structure, making it identifiable by Theorem \ref{thm:identification}. In Appendix E.2 , we empirically verify consistency by repeating the same experiment with varying $N$ from $100$ to $1,000$.

\begin{table}[t]
    \centering
    \renewcommand{\arraystretch}{1.0}
    \captionsetup{labelfont=bf, format=hang} 
    \caption{Comparison of Model Performances}
    \begin{threeparttable}
    \begin{tabular}{@{}l S[table-format=1.3] S[table-format=1.3] S[table-format=1.3] S[table-format=1.3] S[table-format=4]@{}}
        \toprule
        {Algorithm} & {Loading MSE} & {Intercept MSE} & {FDR} & {FNR} & {Time (Seconds)} \\
        \midrule
        PXL-EM & 0.01 & 0.022 & 0.009 & 0.000 & 2682 \\
        10-factor Gibbs & 0.018 & 0.037 & 0.542 \tnote{*}   & 0.000 \tnote{*}  & 2385 \\
        JML & 0.165 & 0.079 & \multicolumn{1}{c}{NA} & \multicolumn{1}{c}{NA} & 4 \\
        MHRM & 0.214 & 0.043 & \multicolumn{1}{c}{NA} & \multicolumn{1}{c}{NA} & 1848 \\
        \bottomrule
    \end{tabular}

    \smallskip
    \footnotesize
    \begin{tablenotes}
    \item[*] FDR and FNR for Gibbs samplers were computed after thresholding.
    \end{tablenotes}
    \end{threeparttable}
    \label{tab:ibp_models_comparison}
\end{table}

\vspace{-0.3cm}
\subsection{Massachusetts DESE Data} \label{subsec:dese_experiment}
\vspace{-0.2cm}

To highlight the advantages of our approach in educational assessment,  we utilized student item response data from Massachusetts Department of Elementary and Secondary Education (DESE). The DESE datasets encompass students' item response to both the English and the Mathematics exam items administrated in Massachusetts in 2022. We focus on the grade 10 students who participated in both the English and Math exams, resulting in a total sample size $N=60,918$. We omitted items with unconventional scores, such as the writing scores in the English exams for simplicity. This resulted in a dataset with 21 English items and 32 math items eligible for analysis. \footnote{Exam questions are available here: \url{https://www.doe.mass.edu/mcas/2022/release/}.}

Considering that a certain level of English skill is necessary to solve math items, but not necessarily vice versa, we anticipated that math items load on both the English and math factors. However, the math factors for the English items should be negligible. To assess whether our approach would successfully recover such a sparse structure, we fit the DESE dataset with our PXL-EM algorithm, and set the truncated dimension as $10$. Setting $\lambda_1=0.1$, we deployed the procedure of dynamic posterior exploration with convergence claimed when $|B^{(t+1)} - B^{(t)}|_{\infty} < 0.04$. Despite the large sample size, the PXL-EM algorithm converged in just $3,483$ seconds, whereas the competing Bayesian method using Gibbs sampling took over $15$ hours. This highlights the computational efficiency of our E-step in large N setting. For the frequentist approaches, the JML and MHRM estimation algorithms took 2,055 and 8,225 seconds, respectively, to fit the data.

Figure \ref{fig:plot5} presents the estimated factor loading matrix for each model. We've labeled each row to indicate its corresponding item type as either English or math. Notably, our PXL-EM algorithm demonstrates a distinct pattern: high factor loadings for the secondary factor commence at item m$22$, aligning precisely with our expectations. Impressively, all math items exhibit substantial positive loadings on the second factor, yet the second factor for the English items are negligible, rendering our estimated loading matrix considerably more interpretable compared to the other models. In contrast, both the Gibbs sampler and the JML estimation algorithm managed to identify two prominent components, but neither of them successfully captured the distinction between the math and the english items. More importantly, for all the competing exploratory methods, their loading estimates for the extra dimensions $3$-$10$ significantly deviate from zero, making it very difficult to determine either the number of factors, or the potential latent structure. 

\begin{figure}[t]
    \centering
    \includegraphics[width=0.7\textwidth, height=0.45\textwidth]{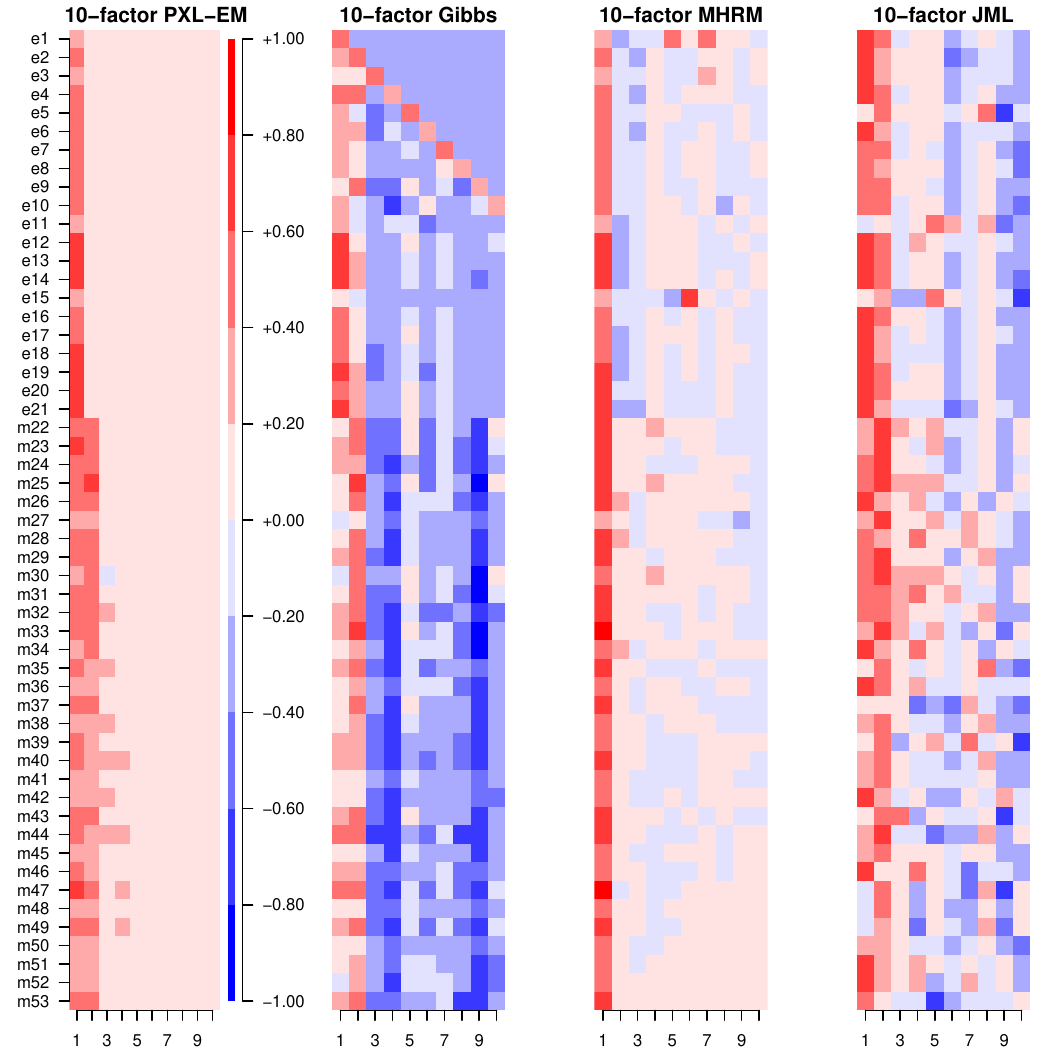}
    \caption{DESE Loading Matrix Estimation}
    \label{fig:plot5}
\end{figure}

Observant readers might notice that our estimated loadings for the math items are noisier than the  English ones, and are not entirely zero across dimensions $3$ to $4$. This is intuitive as DESE classifies these 32 math items into four categories: algebra (13 items) , geometry (13 items), number \& quantity (3 items) and statistics (3 items). There were only $11$ nonzero estimates out of $32$ total math items for the fourth factor, among which only $4$ items exhibit loadings higher than $0.2$. Interestingly, all $3$ rare statistics items have loadings higher than $0.2$ on the fourth factor, which cannot be coincidental. These intriguing patterns suggest that the extra factors might be associated with the detection of rare skills, rather than simply being an artifact of estimation noise. Tables of exact factor loading estimates for each item can be found in the Appendix F.

This experiment emphasizes the benefits of our proposed exploratory MIRT framework in producing interpretable and sparse solutions, as well as detecting the nuances of rare latent math skills in a large $N$ setting. Our method eliminates the need for expert item categorization into rigid subdomains, and is able to make interesting and objective discoveries from data that might be unaware by experts. 

\vspace{-0.5cm}
\subsection{Quality of Life Measurement} \label{subsec:qol_experiment}
\vspace{-0.2cm}

We demonstrate the potential of our approach by applying it to a more intricate Quality of Life (QOL) dataset discussed in Lehman \cite{Lehman1988AQO}. This dataset comprises responses from $586$ individuals with chronic mental illnesses, who answered $35$ items related to their life satisfaction across seven distinct subdomains (family, finance, health, leisure, living, safety, and social). In the psychological measurement literature, Gibbons et al. \cite{bifactor_graded} showcased the advantages of the bifactor model in a confirmatory analysis. We replicated their analysis, and the estimated bifactor loading structure is visualized in the leftmost plot of figure \ref{fig:plot6}.
\begin{figure}[t]
    \centering
    \includegraphics[width=0.85\textwidth, height=0.5\textwidth]{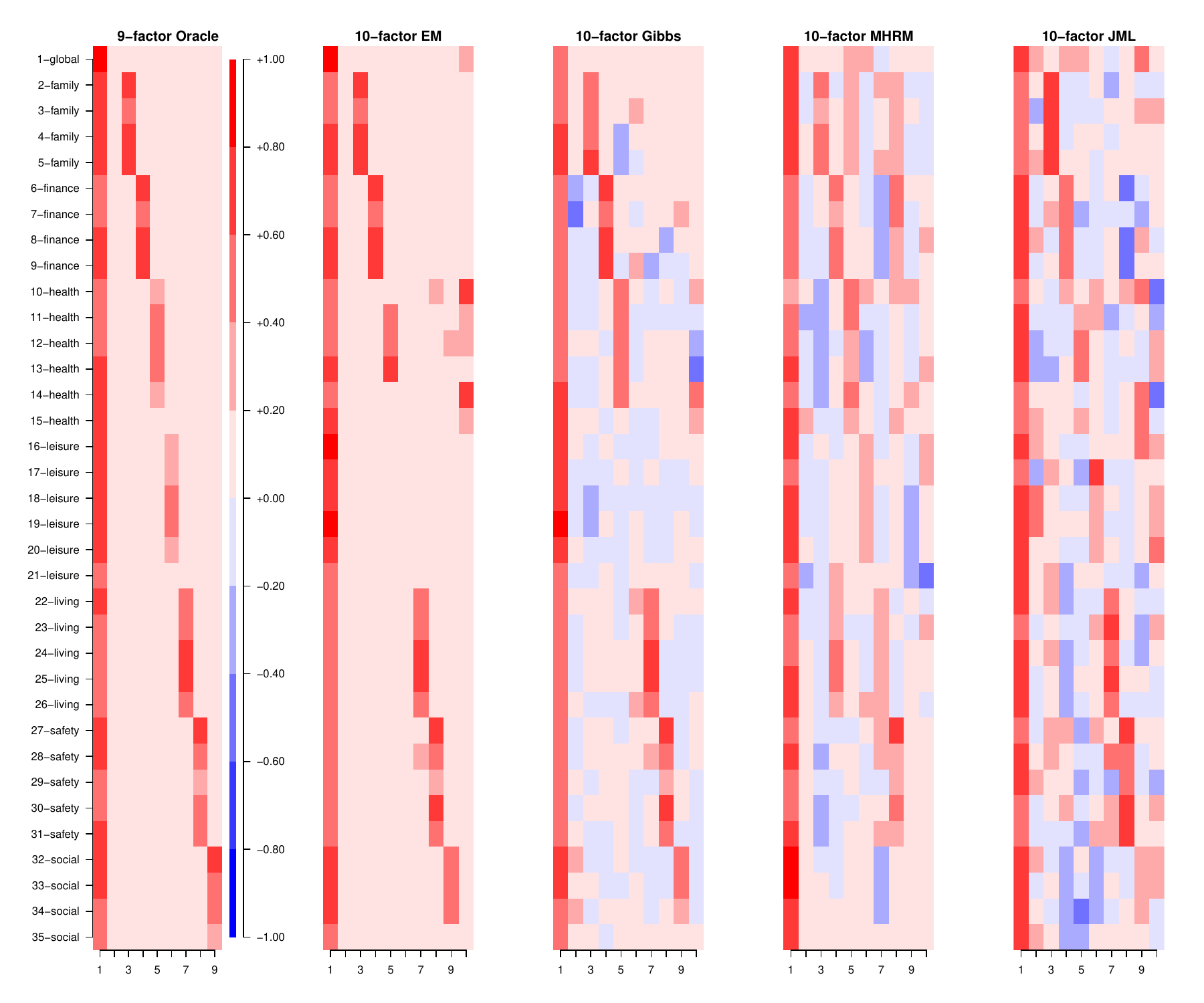}
    \caption{QOL Loading Matrix Estimation}
    \label{fig:plot6}
\end{figure}

Our goal is to determine whether our unconstrained PXL-EM algorithm can reveal a similar bifactor structure without any confirmatory guidance. This particular instance presents a significant challenge, as the algorithm must learn an $8$-dimensional latent space from a mere $35$ available items. We once again adopted an overshot latent dimension of $K^* = 10$ and conducted dynamic posterior exploration. In figure \ref{fig:plot6}, we observe a remarkable resemblance between our estimated latent structure and the bifactor loading pattern. Our algorithm successfully identifies the primary factor and accurately recovers $6$ out of the $7$ subdomains, only with the exception of the ``Leisure'' subdomain. It also correctly discerns the actual latent dimensionality, as dimensions $2$ and $6$ are entirely zeros. 

The loading matrix produced by our PXL-EM algorithm only differs from the oracle bifactor structure in two subdomains: the additional dimensionality for the ``Health'' items, and the absence of ``Leisure'' factor account for the FDR and FNR respectively as shown in Table \ref{tab:qol_table}. In contrast, both the MHRM algorithm and the Gibbs sampler utilized all $10$ dimensions, hence overfitting the item response data and having a marginally lower in-sample reconstruction MSE. The JML estimation method converged dramatically fast for this small $N$ dataset, but obtained relatively high estimation MSE, and a noisy loading representation. In contrast, the pink areas visualized by our PXL-EM algorithm are exact zeros. We provide the detailed loading matrices and more discussions in Appendix G. 

\begin{table}[t]
    \centering
    \renewcommand{\arraystretch}{1.0}
    \captionsetup{labelfont=bf, format=hang}
    \caption{Model Performances for QOL Data}
     \begin{threeparttable}
    \begin{tabular}{@{}l S[table-format=1.3] S[table-format=1.3] S[table-format=1.3] S[table-format=4]@{}}
        \toprule
        {Algorithm} & {FDR\tnote{*}} & {FNR\tnote{*}} & {Reconstruction MSE} & {Time (Seconds)} \\
        \midrule
        PXL-EM & 0.092 & 0.157 & 0.102 & 411 \\
        Sparse Gibbs & 0.678 & 0.057 & 0.096 & 410 \\
        MHRM & NA & NA & 0.088 & 350 \\
        JML & NA & NA & 0.156 & 3 \\
        \bottomrule
    \end{tabular}
    \smallskip
    \footnotesize
    \begin{tablenotes}
    \item[*]  We performed thresholding for the Gibbs sampler, but not for the PXL-EM.
    \end{tablenotes}
    \end{threeparttable}
     \label{tab:qol_table}
\end{table}

Essentially, if the goal of exploratory analysis is to cluster items into distinct types, our PXL-EM algorithm clearly demonstrates distinct characteristics among all seven subdomains. This experiment showcases the power of our proposed EM approach to learn high-dimensional latent structures with very few items, highlighting its unique advantages compared to other state-of-the-art exploratory factor analysis strategies.

\vspace{-0.3cm}
\subsection{Ordinal bio-behavioral Data} \label{subsec:bsnip_experiment}

To highlight the potential of our methodology in uncovering new scientific insights, we apply it to a challenging ordinal bio-behavioral dataset with intricate latent structures. Recent research \cite{Stan2020} utilized the confirmatory MIRT bifactor model to identify strong negative correlations between cortical thickness and severity of clinical symptoms. Our goal is to determine whether our exploratory data-driven framework can make similar discoveries without the need to prespecify a loading structure and the number of factors, which typically require multiple stages of statistical analysis in a confirmatory setting.

Following the analysis in \cite{Stan2020}, we investigate the relationship between the 16 cerebral cortex regions (ordinal scale 1-5) and 51 symptom items across 5 subdomains: MADRS (depression; 10 items), PANSS-P (positive psychosis; 7 items), PANSS-N (negative psychosis; 7 items), PANSS-G (general psychosis; 16 items), and Young (mania; 11 items). The symptom items vary in ordinal scales from 1 to 8. We consider participants who have both complete cortical thickness items and clinical symptom items, resulting in a sample size of $N=575$. We set truncation level $K^*=10$, and deploy our PXL-EM algorithm via dynamic posterior exploration. The estimated loading matrix along with the ones produced by the other competing methods are presented in figure \ref{fig:plot6}. \footnote{We did not include the JML method this time because it does not seem to handle ordinal datasets.}

As indicated by the color palette, the red, light pink, and blue regions represent negative, zero, and positive loadings respectively. Notably, the loading matrix estimated by PXL-EM is significantly sparser and more interpretable compared to those produced by the other methods: the first 1-3 dimensions represent the primary bio-behavioral dimensions of interest, shared by both the cortical thickness items (first 16 rows) and the clinical symptom items (last 51 rows). Dimension 4 is an independent factor for cortical thickness, while dimensions 5-10 are shared by the 51 clinical symptoms.

\begin{figure}[t]
    \centering
    \includegraphics[width=0.75\textwidth, height=0.55\textwidth]{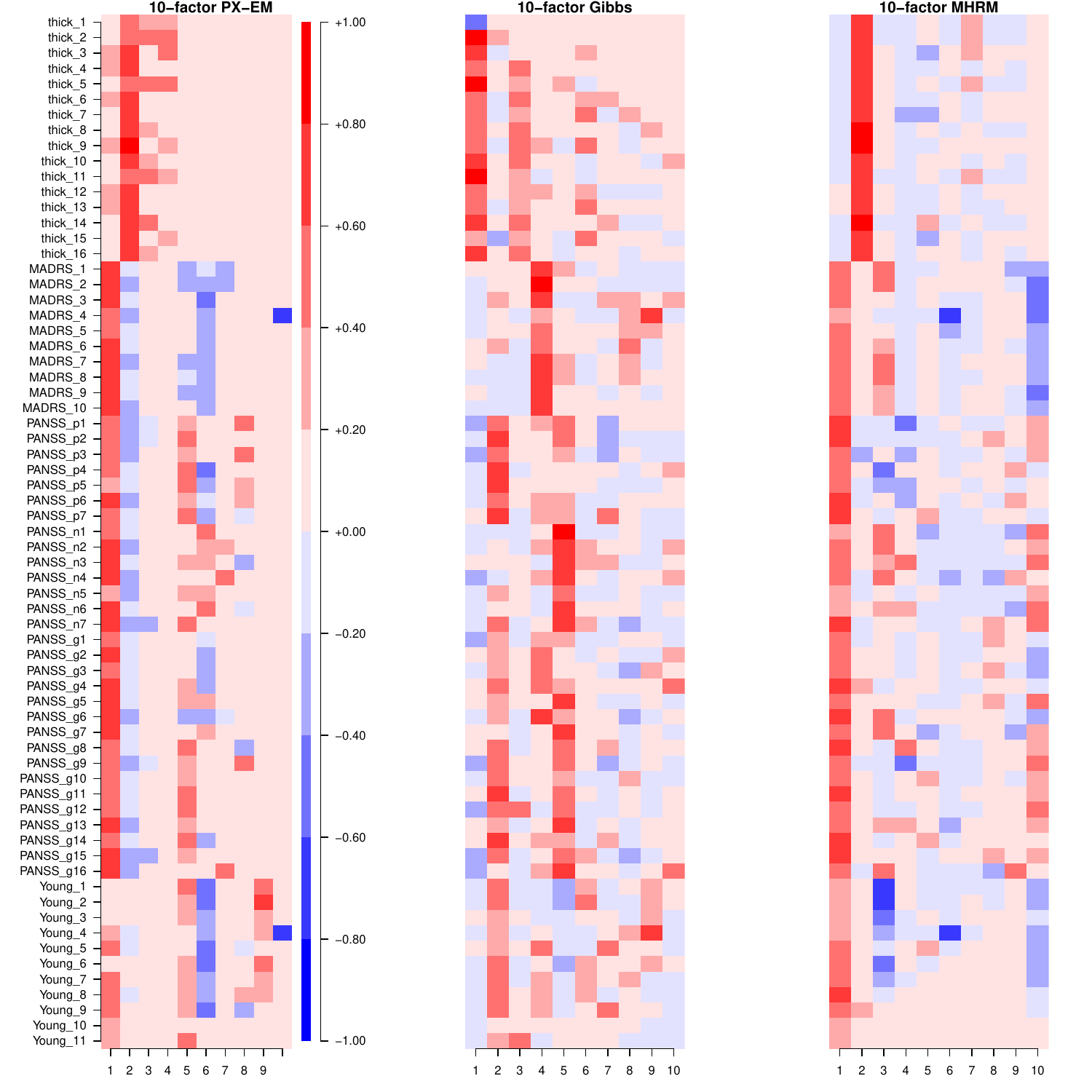}
    \caption{Estimated bio-behavioral Factor Loading Matrices}
    \label{fig:plot6}
\end{figure}

To investigate the relationship between cortical thickness and clinical psychotic manifestations, it is essential to identify the (bio-behavioral) factors that are shared by all the items. Such analysis is unfeasible based on the factor loadings produced by the Gibbs sampler and the MHRM algorithm, where all the items are associated with all the $10$ factors, because these methods neither produce exact zeros nor can determine the number of factors. In contrast, our PXL-EM algorithm simplifies this investigation by identifying only three bio-behavioral dimensions, with the remaining seven dimensions independently loaded by either the biological or clinical features. More importantly, the second bio-behavioral dimension clearly highlights the negative correlation between cortical thickness and the severity of clinical symptoms, as the cortical thickness items are predominantly red, while the clinical symptoms items are overwhelmingly blue. 

In addition to identifying the negative correlation, we are also able to pinpoint the exact symptoms that are highly correlated with reduced cortical thickness. The research in \cite{Stan2020} highlights seven symptom items with highly negative factor loadings lower than -0.25: delusions (PANSS-p1), hallucinatory behavior (PANSS-p3), suspiciousness (PANSS-p6), depression (PANSS-g6), unusual thought content (PANSS-g9), active social avoidance (PANSS-g16), and passive-apathetic social withdrawal (PANSS-n4). Remarkably, by examining the second bio-behavioral dimension, we identify six symptoms with factor loadings lower than -0.25, each of which is among the seven highlighted symptoms listed above. The only missing item is PANSS-g6, which still has a very negative loading of -0.24. We invite readers to refer to the tables in Appendix H for the detailed estimates and comparison.

This experiment demonstrates the generality and effectiveness of our approach in handling complex bio-behavioral ordinal datasets with nuanced latent structures. Not only can we generate a sparse and interpretable solution unattainable by other competing estimation methods, but our scientific discoveries also closely align with the established research findings. Moreover, our exploratory data-driven approach avoids making assumptions about the loading structure, which is particularly valuable in settings where such assumptions are too ambiguous to be made, as is often the case in scientific inquiries.

\vspace{-0.7cm}
\section{Discussion} \label{sec:discussion}
\vspace{-0.3cm}

This article advocates for a novel Bayesian framework for estimating exploratory sparse MIRT models. The combination of the Indian Buffet Process prior and the spike-and-slab LASSO prior allows our EM algorithm to simultaneously learn latent dimensionality and sparse representations of factor loadings, which have not been sufficiently addressed in the categorical actor analysis literature. The efficiency of the algorithm is driven by the adaptation of parameter expansion, as well as by an efficient sampling strategy during the E-step, where we explicitly derived latent factor posterior distributions. We further demonstrate the advantage of adopting a probit link for establishing identification results and provide theoretical guarantees for recovering the true number of factors and contraction towards the true factor loading matrix under our prior specification. Our approach is capable of accommodating mixed data types, and we showcase the effectiveness of our approach on a challenging synthetic dataset and three real-world datasets from different domains. In each experiment presented in Section \ref{sec:empirical}, our algorithm consistently generates sparse and interpretable estimations of the factor loading matrices, demonstrating its unique strength compared to other popular exploratory MIRT estimation frameworks.

While we focus on the estimation of the factor loadings and intercepts, our framework also has direct implications for latent trait measurements. For a well-calibrated set of $J$ items where the item parameters are known, Theorems \ref{thm:3.2} essentially provides an analytic form of the latent factor posterior distribution of any new observation $y \in \mathbbm R^J$. This avoids the additional computational need to perform MCMC sampling in order to conduct posterior inference on latent traits. Compared to the other MIRT estimation strategies, this can potentially save tremendous amount of computational time when the number of observations is large. Furthermore, many important properties of the unified skew-normal distribution such as the mean and covariance are readily available in \cite{sun}, making exact Bayesian inference possible. For more complicated objects of inferential interest, one may also consider the efficient posterior sampling strategies outlined in Section \ref{subsec:e-step} to compute credible intervals. In Appendix G.2, we conducted a comparative analysis of the estimated latent factors and their posterior variances between our proposed PXL-EM approach and the oracle bifactor model for the QOL dataset, highlighting the potential of our approach in posterior factor inference and psychological measurements.

Several promising research directions can shape the future of our approach. First, advancements in evaluating high-dimensional Gaussian cumulative distribution functions could eliminate the need for Monte-Carlo simulations in the E-step and factor rotation step. This improvement would not only speed up computation but also simplify Bayesian inference of latent factors following factor loading estimation. Moreover, we can enhance the flexibility of our Bayesian model. Recent work by Anceschi et al. \cite{doi:10.1080/01621459.2023.2169150} has demonstrated posterior conjugacy in probit and tobit models, allowing the prior for coefficient parameters to follow a unified-skew normal distribution (SUN). This adjustment could further refine the modeling capabilities of our approach, with implications for various applications.

While our methodology can handle mixed continuous, binary, and ordinal responses, it would be interesting to extend it to other item types such as count data, or to embed our Bayesian model into generalized mixed-models for longitudinal data. One can derive the latent factor posterior distribution using techniques similar to those in Theorem \ref{thm:3.2}, in order to create a feasible E-step to estimate various factor models of interest. 

\vspace{-0.7cm}
\section*{Acknowledgments}
\vspace{-0.4cm}

The authors sincerely thank the Associate Editor and four anonymous referees for their insightful comments and valuable suggestions, which have significantly enhanced the quality of this paper. Additionally, Robert D. Gibbons acknowledges support from the National institute on Aging (NIA-R56 AG084070), and Veronika Ro\v{c}kov\'{a} acknowledges support from the National Science Foundation (NSF-DMS 1944740).
\vspace{-0.7cm}
\section*{Disclosure Statement}
\vspace{-0.4cm}
The authors report there are no competing interests to declare.

\addtolength{\textheight}{-.3in}%

\spacingset{1.2}
\bibliographystyle{chicago}
\bibliography{reference_all}

\clearpage
\bigskip
\spacingset{1.0}

\begin{center}
{\Large{\textbf{Supplementary Material: Sparse Multidimensional Item Response Theory }}}
\end{center}

\appendix

\renewcommand{\baselinestretch}{1.8} \normalsize
\section{Gibbs Sampler with Adaptive Spike-and-Slab Prior} \label{subsec:gibbs}
We provide the derivation of our sparse Gibbs sampler in details. Following similar notations as in the paper: let $B \in \mathbbm{R}^{J \times K}$ be the factor loading matrix,  $Y \in \mathbbm{R}^{N \times J}$ be the item response data, $\Gamma \in \mathbbm{R}^{J \times K}$ be the latent allocation matrix, $D \in \mathbb R^{J}$ be the intercept parameters, and $\Theta = [\theta_1, \cdots, \theta_n]' \in \mathbbm{R}^{N \times K}$ be the matrix of latent factors. Additionally, we introduce the P\'{o}lya-Gamma augmentation matrix  $\Omega= \{w_{ij}\} \in \mathbbm{R}^ {N\times J}$,  in which each element  $w_{ij} \sim \text{PG}(1,0)$.

Many of our derivation steps are adapted from \cite{pg_mirt}. The key derivation idea is explored in \cite{pg}: for a logistic regression with covariate $x_i \in \mathbb R^J$, coefficient vector $\beta \in \mathbb R^J$, outcome variable $y \in \mathbb R^N$, and P\'{o}lya-Gamma random variables $w=[w_{1}, w_{2}, \ldots, w_{N}]$, we can express the posterior of $\beta$ as following:
$$
\begin{aligned}
p({\beta} \mid {w}, {y}) & \propto p({\beta}) \prod_{i=1}^{N} \exp \left\{k_{i} {x}_{i}^{T} {\beta}-\frac{w_{i}\left({x}_{i}^{T} \beta\right)^{2}}{2}\right\} \\
& \propto p({\beta}) \exp \left\{-\frac{1}{2}(z-{X} {\beta})' {\tilde \Omega}(z-{X} {\beta})\right\}
\end{aligned},
$$
where $k=(w_1-\frac{1}{2}, \cdots, w_N- \frac{1}{2})$, $z = (k_1/w_1, \cdots, k_N/w_N)$, and $\tilde \Omega = \text{diag}(w_1, \cdots, w_N)$, and $X \in \mathbbm R^{N\times J}$ is the data.

Our Gibbs sampler iteratively sample model parameters using the following four steps: 

\subsection{Gibb Sampling Latent Trait $\theta_i$}

Let $y_i \in \mathbbm{R}^J$ be student $i$'s item response ($i$-th row of $Y$), $k_i =[y_{i1} - \frac{1}{2}, \cdots, y_{iJ}- \frac{1}{2}]' \in \mathbbm{R}^J$, $\Omega_i = \text{diag}(w_{i1}, \cdots, w_{iJ})$, and $z_i= [k_{i1}/w_{i1},\cdots, k_{iJ}/w_{iJ}]'$.

Given factor loading matrix $B$, Intercept $D$, P\'{o}lya-Gamma augmentation matrix $\Omega$, item response data $Y$, the conditional posterior distribution for $\theta_i$ can be obtained as:
$$
\begin{aligned}
\pi(\theta_i | B, D, \Omega, Y) & \propto p(\theta_i) \exp \{-\frac{1}{2}(z_i-(B\theta_i + D))' \Omega_i (z_i-(B\theta_i + D)) \} \\
& \propto p(\theta_i) \exp \{-\frac{1}{2}((z_i-D) - B\theta_i))' \Omega_i ((z_i-D) - B\theta_i)) \}  .
\end{aligned}
$$

Here we have $p(\theta_i) \sim N(0, \mathbbm{I}_k)$, likelihood can also be treated as Gaussian density in $\theta_i$, it follows the posterior is another Multivariate Gaussian $N(\mu_{\theta_i}, \Sigma_{\theta_i})$ where
\vspace{-0.3cm}
\begin{gather*}
\Sigma_{\theta_i} = (B' \Omega_i B + \mathbbm{I_k})^{-1},\\
\mu_{\theta_i} = \Sigma_{\theta_i}(B'\Omega_i(z_i-D)).
\end{gather*}

\subsection{Gibb Sampling Intercept D}

Let $y_j \in \mathbbm{R}^N$ be all student's item response for question $j$, ($j$-th column of $Y$). $k_j =[y_{1j} - \frac{1}{2}, \cdots, y_{Nj}- \frac{1}{2}]' \in \mathbbm{R}^N$, $\Omega_j = \text{diag}(w_{1j}, \cdots, w_{Nj})$, and $z_j= [k_{1j}/w_{1j},\cdots, k_{Nj}/w_{Nj}]'$. Let $B_j \in \mathbbm{R}^k$ be the $j$-th column in $B$.

Given factor loading matrix $B$, latent traits $\Theta$, P\'{o}lya-Gamma augmentation matrix $\Omega$, item response data $Y$, the conditional posterior distribution for $d_j$ can be obtained as:
$$
\begin{aligned}
\pi(d_j | B, \Theta, \Omega, Y) & \propto p(d_j) \exp \{-\frac{1}{2}(z_j-(\Theta B_j + \mathbbm{1} d_j))' \Omega_j (z_i-(\Theta B_j + \mathbbm{1} d_j)) \} \\
& \propto p(d_j) \exp \{-\frac{1}{2}((z_j-\Theta B_j) - \mathbbm{1} d_j)' \Omega_j ((z_i-\Theta B_j) - \mathbbm{1} d_j)) \}  .
\end{aligned}
$$

Here we have $p(d_j) \sim N(0, 1)$, likelihood can also be treated as the product of $N$ normal density of $d_j$, which makes the overall posterior still have Gaussian density. To see this, write $V= Z_j-\Theta B_j$, we have
$$\exp \{-\frac{1}{2}((z_j-\Theta B_j) - \mathbbm{1} d_j)' \Omega_j ((z_i-\Theta B_j) - \mathbbm{1} d_j)) \} \propto \prod_{i=1}^n \exp \{- \frac{1}{2} w_{ij} (V_i-d_j)^2\}.$$

Hence each term in the likelihood function follows a normal density $N(V_j, \frac{1}{w_{ij}})$.

\subsection{Gibbs Sampling Loading}

\subsubsection{Baseline: Normal Prior on $B$}

We start by considering the case when only the normal prior is assumed for $B$. Let $B_{j(-k)} \in \mathbbm{R}^{k-1}$ be the vector $B_j$ excluding its $k$ component. Given intercept $D$, latent traits $\Theta$, P\'{o}lya-Gamma augmentation matrix $\Omega$, item response data $Y$, and $B_{j(-k)}$, the joint conditional posterior distribution for $B_{jk}$ can be obtained as:
$$
\begin{aligned}
\pi(B_{jk} | D, \Theta, \Omega, Y, B_{j(-k)}) & \propto p(B_{jk})  \exp \{-\frac{1}{2}(z_j-(\Theta_{-k} B_{j(-k)} + \Theta_k B_{jk} + \mathbbm{1} d_j))' \\
 & \qquad \Omega_j (z_j-(\Theta_{-k} B_{j(-k)} + \Theta_k B_{jk} + \mathbbm{1} d_j)) \} \\
& \propto p(B_{jk})  \exp \{-\frac{1}{2}((z_j- \Theta_{-k} B_{j(-k)} - \mathbbm{1} d_j) - \Theta_k B_{jk}))' \\
 & \qquad \Omega_j ((z_j-\Theta_{-k} B_{j(-k)} - \mathbbm{1} d_j) -  \Theta_k B_{jk}) \}  .
\end{aligned}
$$
As before, we expect the likelihood to have normal density as it can be written as the product of $N$ Gaussian densities: write $V= Z_j-\Theta_{-k} B_{j(-k)} -\mathbbm 1 d_j$ and $\Theta_{k}^{i}$ be the $i$-th element of the vector $\Theta_{k}$, we have:
\vspace{-0.3cm}
    \begin{align*}
    & \exp \{-\frac{1}{2}((z_j- \Theta_{-k} B_{j(-k)} - \mathbbm{1} d_j) - \Theta_k B_{jk}))'\Omega_j ((z_j-\Theta_{-k} B_{j(-k)} - \mathbbm{1} d_j) -  \Theta_k B_{jk}) \}   \\
    & \propto \prod_{i=1}^n \exp\{-\frac{1}{2}w_{ij}(V_i - \Theta_{k}^{i} B_{jk})^2\} \\
    & \propto \prod_{i=1}^n \exp\{-\frac{1}{2} (\Theta_{k}^{i})^2 w_{ij}(\frac{V_i}{\Theta_{k}^{i}} - B_{jk})^2\}.
    \end{align*}
    It follows each term $i$ in the likelihood part follows the normal density $N(\frac{V_i}{\Theta_{k}^{i}}, \frac{1}{w_{ij} (\Theta_{k}^{i})^2 })$.

\subsubsection{Adaptive Spike-and-Slab Gaussian Prior on $B$}

We consider imposing the adaptive Spike-and-Slab Gaussian prior from \cite{ir_ss} on factor loading matrix $B$. Note this is the Gibbs sampler we implemented to conduct all our experiments: 
\vspace{-0.3cm}
\begin{gather*}
B_{jk}|\gamma_{jk}, \tau_{jk}^2 \sim  N(0, \gamma_{jk} \tau_{jk}^2), \\
\gamma_{jk} | v_0, \theta \sim (1-\theta) \delta_{v_0}(.) + \theta \delta_1(.), \\
\tau_{jk}^2 | a_1, a_2 \sim \text{InverseGamma}(\alpha_1, \alpha_2),\\
\theta \sim \text{Uniform}[0, 1].
\end{gather*}
To sample $B_{jk}$ at step $i$, our Gibbs sampler would perform the following steps

\begin{itemize}
    \item Sample $B_{jk}^{(i)}$:
    $$
\begin{aligned}
\pi(B_{jk}^{(i)} | \gamma_{jk}^{(i-1)}, \theta^{(i-1)}, \tau_{jk}^{2(i-1)}, D, \Theta, \Omega, Y, B_{j(-k)}) & \propto p(B_{jk}^{(i)}| \gamma_{jk}^{(i-1)}, \tau_{jk}^2) p(\gamma_{jk}^{i-1} | \theta^{(i-1)}) p(\theta^{(i-1)}) \times \\
& \quad \exp \{-\frac{1}{2}((z_j- \Theta_{-k} B_{j(-k)} - \mathbbm{1} d_j)  - \Theta_k B_{jk}))' \\
 & \qquad   \Omega_j ((z_j-\Theta_{-k} B_{j(-k)} - \mathbbm{1} d_j) -  \Theta_k B_{jk}) \} .
\end{aligned}
$$
\item  Sample $\gamma_{jk}^{(i)}$:
\begin{align*}
&\pi(\gamma_{jk}^{(i)} =  1 | \beta_{jk}^{(i)}, \theta^{(i-1)}, \tau_{jk}^{(i-1)}) =  \\
& \qquad \frac{p(B_{jk}^{(i)}| \gamma_{jk}=1, \tau_{jk}^{(i-1)}, \theta^{(i)})p(\gamma_{jk}^{(i)} =1 | \theta)}{p(B_{jk}^{(i)}| \gamma_{jk}=1, \tau_{jk}^{(i-1)}, \theta^{(i)})p(\gamma_{jk}^{(i)} =1 | \theta) + p(B_{jk}^{(i)}| \gamma_{jk}=v_0, \tau_{jk}^{(i-1)}, \theta^{(i)})p(\gamma_{jk}^{(i)} =v_0 | \theta)}.
\end{align*}
\item Sample $\tau_{jk}^{(i)}$
\begin{align*}
    \pi(\tau_{jk}^2 | B_{jk}, \gamma_{jk}, \theta) & \propto p(B_{jk}| \tau_jk^2, \gamma_{jk}) p(\tau_{jk}^2) \\
    & \propto (\gamma_{jk} \tau_{jk}^2)^{-\frac{1}{2}} \exp\{-\frac{1}{2\gamma_{jk} \tau_{jk}^2} B_{jk}^2\} (\gamma_{jk}^2)^{-\alpha_1-1} \exp\{-\frac{\alpha_2}{\gamma_{jk}^2}\} \\
    & \propto (\tau_{jk}^2)^{-(\alpha_1+\frac{1}{2})-1} \exp\{-(\frac{B_{jk}^2}{2\gamma_{jk}}+\alpha_2) / \gamma_{jk}^2 \}.
\end{align*}
It follows we can sample $\tau_{jk}^{2(i)}$ from InverseGamma($\alpha_1+\frac{1}{2}$, $\frac{B_{jk}^2}{2\gamma_{jk}}+ \alpha_2$).

\item Sample $\theta^{(i)}$:
\begin{align*}
\pi(\theta^{(i)} | \{\gamma_{jk}^{(i)}\}) & \propto \theta^{\#\{\gamma_{jk}=1\}} (1-\theta)^{\# \{\gamma_{jk}= v_0\}} \\
& \propto \theta^{\sum_j \sum_k \gamma_{jk}} (1-\theta)^{JK- \sum_j \sum_k \gamma_{jk}} \\
& \propto \text{Beta}(\sum_j \sum_k \gamma_{jk}+1 , JK- \sum_j \sum_k \gamma_{jk} + 1).
\end{align*}
    
\end{itemize}

\subsection{Updating P\'{o}lya-Gamma Variables}
Based on the updates in \cite{pg}, we have $$\pi(w_{ij}|B, \Theta, D) \sim \text{PG}(1, B_j' \theta_i+d_j ).$$

\section{Sparse Bayesian MIRT: Model Properties}

\subsection{Proof of Theorem 3.3} \label{subsec:id_proof}

\begin{proof}
    \textbf{Part (a)}: Suppose two sets of model parameters $(A, D_1)$ and $(B, D_2)$ define the same probit factor model, we want to show $D_1 = D_2$ when condition (A) holds. By proposition 3.1, we must have 
    \begin{equation} \label{eqn:linearization}
        \tilde{A}\tilde{A'} + \Sigma_A = \tilde{B}\tilde{B'} + \Sigma_B,
    \end{equation}
    where $\Sigma_A= \text{diag}\{\frac{D_{1,j}}{\sqrt{\|A_j\|^2 +1}} - \frac{\|A_j\|^2}{\|A_j\|^2 + 1}\}_{j=1}^{J}$ and $\Sigma_A= \text{diag}\{\frac{D_{2,j}}{\sqrt{\|B_j\|^2 +1}} - \frac{\|B_j\|^2}{\|B_j\|^2 + 1}\}_{j=1}^{J}$. Observe if we can show the diagonal elements of $\tilde{A} \tilde{A'}$ and $\tilde{B} \tilde{B'}$ are equal, then we are done. To see why, note $\text{diag}(\tilde{A} \tilde{A'}) = \text{diag}(\tilde{B} \tilde{B'})$ implies $\frac{\|A_j\|^2}{\|A_j\|^2 + 1} = \frac{\|B_j\|^2}{\|B_j\|^2 + 1}$ for all $j$, which further implies $\|A_j\|^2= \|B_j\|^2$ for all $j$. In addition, $\text{diag}(\tilde{A} \tilde{A'}) = \text{diag}(\tilde{B} \tilde{B'})$ also implies $\Sigma_A = \Sigma_B$. It then follows we must have $D_{1,j} = D_{2,j}$ for all $j$ and the intercept is identified. Finally, to prove $\text{diag}(\tilde{A} \tilde{A'})$ is indeed equivalent to $\text{diag}(\tilde{B} \tilde{B'})$, note the row deletion condition is imposed on $\tilde{B}$ rather than on $B$, and hence we can  follow the same arguments as appeared in Theorem 5.1 of Anderson and Rubin \cite{Anderson1956StatisticalII} to finish the proof.  
    
    \textbf{Part(b)}: Since the intercept is uniquely identified, we consider two sets of model parameters $(A, D)$ and $(B, D)$ defining the same probit factor model. We first observe if $\tilde{A} = \tilde{B}$, it must be the case $A=B$. This is because for any two K-dimensional vectors $a$ and $b$, the equation $\frac{a}{\sqrt{\|a\|^2+1}} = \frac{b}{\sqrt{\|b\|^2+1}}$ implies $a=b$. Hence it is sufficient to prove the theorem by showing $\tilde{A}= \tilde{B}$. Since $D_{1,j}= D_{2,j}$, the equation in the Proposition 3.1  implies $\|A_j\| = \|B_j\|$, which further implies $\Sigma_A = \Sigma_B$ as defined in equation \ref{eqn:linearization}. It follows we must have $\tilde{A}{A'} = \tilde{B}{B'}$. Since $A$ and $B$ conforms to an ordered GLT structure, we know $\tilde{A}$ and $\tilde{B}$ also follows an ordered GLT structure. Then we can simply apply Theorem 1 of \cite{Fr_hwirth_Schnatter_2023} to obtain $\tilde{A} = \tilde{B}$, which implies $A=B$.    
\end{proof}

\subsection{Proof of Theorem 3.4} \label{subsec:proof-ld}

\begin{proof}
We denote the Kullback-Leibler divergence between two probability densities $f$ and $g$ with respect to a measure $\mu$ as  $K(f,g) = \int f \log (f/g) d\mu$, and the variation measure $V_{2,0}(f, g) = \int f |\log(f/g) -K(f,g)|^2 d \mu$.
Define event $E = \{B : K(B) > DS_n K_{0n} \}$ and the likelihood $f_{\pi}(y_{ij}) = [\Phi(B_j' \theta_i + d_j)]^{y_{ij}} [1-\Phi(B_j' \theta_i + d_j)]^{1- y_{ij}}$ under the prior distribution of $\Pi := (B, \Theta)$. Similarly, let $\Pi_0$ be the sequence of true parameters distributions. Then we have
\begin{equation} \label{eq:evidence-lb}
    P(E | Y^{n}) = \frac{\int_E \prod_{i=1}^N \prod_{j=1}^J f_{\pi}(y_{ij})/f_{\pi_0}(y_{ij}) d \Pi }{\int \prod_{i=1}^N \prod_{j=1}^J f_{\pi}(y_{ij})/f_{\pi_0}(y_{ij}) d \Pi} := \frac{N_n}{D_n}.
\end{equation}
We further define $A_n \in \sigma(Y^{n})$ (the sigma field generated by the observed data) as follows:
\begin{align} \label{eq:kl-concentration}
  A_n:= \{ B \in \mathbbm R^{J_n \times \infty}, \Theta \in \mathbbm R^{N \times \infty}:  \frac{1}{NJ_n} \sum_{i, j} K(f_{\pi_0}(y_{ij}), f_{\pi}(y_{ij})) \leq \frac{ K_{0n} S_n \log(J_n)}{N J_n},\\
   \frac{1}{N J_n} \sum_{i, j} V_{2,0}(f_{\pi_0}(y_{ij}), f_{\pi}(y_{ij})) \leq \frac{ K_{0n} S_n \log(J_n)}{N J_n}\}
\end{align}
Note while each component of observation $y_i \in \mathbbm R^j$ is not identically distributed, they are still independent. Hence the KL divergence term is additive and the variation measure term is subadditive up to $O(1)$ (section 8.3.1 of \cite{GhosalVaart2017}). When conditional on the event $A_n$,  we can apply Lemma 10 of \cite{noniid} to lower bound $D_n$, and obtain for any $c>0$:
\begin{equation} \label{eq:evidence-lb-2}
    P_{\Pi_0}^{(n)}(D_n \geq \pi(A_n) \cdot e^{-(1+c)  K_{0n} S_n \log(J_n) } ) \geq 1- \frac{1}{c^2  K_{0n} S_n \log(J_n)} \xrightarrow[]{n \to \infty} 1
\end{equation}

Furthermore, by Theorem $2$ of \cite{bfa}, we can upper bound the probability of event E such that $P(E) \precsim e^{-DK_{0n} S_n \log(J_n + 1)}$. Hence we have:
\begin{equation} \label{eq:con_pos}
    \mathbbm{E}_{\Pi_0} [P(E | Y^{n}) \mathbbm{1}(A_n)] \precsim \frac{e^{-DK_{0n} S_n \log(J_n + 1)}}{\pi(A_n) \cdot e^{-(1+c)  K_{0n} S_n \log(J_n) } }
\end{equation}
Observe if we can show $\pi(A_n) \gtrsim e^{- K_{0n} S_n \log(J_n)}$, then with $D$ large enough, the above equation will converge to $0$ and we will be done. To lower bound $\pi(A_n)$, we may compute the KL divergence and the variation terms directly and then bound their maximums in a similar fashion as in the proof of lemma 1 of \cite{glm_contraction}, but under a random design matrix: let $\eta_{ij} := B_j' \theta_i$, $\eta_{0ij} := B_{0j}' \theta_{0i}$, $h(x):= \Phi^{-1}(x)$, $b(x) := \log(1+e^x)$, and $\xi(x) := (h \circ b')^{-1} (x)$. Then it is easy to see for each $(i,j)$ pair, we have:
\begin{align} \label{eq:kl-explicit}
    & K(f_{\pi_0}(y_{ij}), f_{\pi}(y_{ij})) = (\xi(\eta_{0_{ij}}) - \xi(\eta_{{ij}}))b'(\xi(\eta_{0_{ij}})) - b(\xi(\eta_{0_{ij}})) + b(\xi(\eta_{{ij}})), \\
    & V_{2,0}(f_{\pi_0}(y_{ij}), f_{\pi}(y_{ij})) = b^{''}(\xi(\eta_{0_{ij}}))(\xi(\eta_{0_{ij}}) - \xi(\eta_{{ij}}))^2.
\end{align}
Observe that $b^{''}(\xi(\eta_{0ij})) \xi'(\eta_{0ij}) = (h^{-1})'(\eta_{0ij})$ by definition, if we conduct Taylor expansion at $\eta_{0ij}$ for the two terms above as functions of $\eta_{ij}$, we have
\begin{align} \label{eqn:klv_bound}
  & \max\{K(f_{\pi_0}(y_{ij}), f_{\pi}(y_{ij})), V_{2,0}(f_{\pi_0}(y_{ij}), f_{\pi}(y_{ij}))  \} \notag \\
  & \quad \leq b^{''} (\xi(\eta_{0ij}))(\xi'(\eta_{0ij}))^2 (\eta_{ij} - \eta_{0ij})^2 + O((\eta_{ij}-\eta_{0ij})^2) \notag \\
  & \quad = (h^{-1})^{'}(\eta_{0ij})(\xi'(\eta_{0ij})) (\eta_{ij} - \eta_{0ij})^2 + O((\eta_{ij}-\eta_{0ij})^2) \notag \\
  & \quad = 1 + \max_{\{i \leq n, j \leq J_n\}} \frac{\phi^2(\eta_{0_{ij}})}{\Phi(\eta_{0ij}) (1-\Phi(\eta_{0ij}))}  + O((\eta_{ij}-\eta_{0ij})^2).
\end{align}
Recall $\gamma_{n}(B_{0n}, \Theta_{0n}) :=  1 + \max_{\{i \leq n, j \leq J_n\}} \frac{\phi^2(\eta_{0_{ij}})}{\Phi(\eta_{0ij}) (1-\Phi(\eta_{0ij}))} $. We can then lower bound $\pi(A_n)$ as follows:
\begin{align*}
    \pi(A_n) & \geq  \pi \{B, \Theta: \gamma_n \max_{i,j} \{ (\eta_{ij} - \eta_{0ij})^2\} \leq \frac{ K_{0n} S_n \log(J_n)}{N J_n}  \} \\
    & \geq \pi\{B, \Theta: \|\Theta B' - \Theta_0 B_0'\|^2_{\infty} \leq \frac{ K_{0n} S_n \log(J_n)}{ \gamma_n N^2 J_n}  \} 
\end{align*}
Observe under the prior distribution $\Pi_0(.)$, if we condition on the event that the effective dimension $B$ is no larger than $K_0$, then each row of the matrix $(\Theta B' - \Theta_{0} B_{0}')$ follows the same distribution as $(B_{K_0} + B_{0, K_0})Z$, where $Z \sim N(0, \mathbbm{I}_{K_0})$. Hence we have $$\|(B_{K_0} + B_{0, K_0})Z\|_{\infty} \leq \|B_{K_0}\|_{\infty} \|Z\|_{\infty} + \|B_{0, K_0}\|_2 \|Z\|_2 \leq (\|B_{K_0}\|_F+\sqrt{\bar{S}_n J_n}) \|Z\|_{\infty},$$
where the first inequality is due to the combination of the triangle and the Cauchy-Schwarz inequalities, and the last inequality is due to restriction of the 2-norm on $B_{0n}$ by assumption B. Conditional on the prior concentration event $\|B\|_F < \sqrt{K_{0n}/n}$, we can further lower bound the expression as follows:
\begin{align*}
    \pi(A_n) \geq P(K \leq K_{0n}) P(\|B\|_F < \sqrt{K_{0n}/n}) \prod_{i=1}^n P(\|Z\|_{\infty}^2 \leq \frac{ K_{0n} S_n \log(J_n)}{ \gamma_n n^2 J_n (\sqrt{K_{0n} / n} + \sqrt{S_n J_n})^2 })
\end{align*}
With the choice $\alpha = \frac{1}{J_n}$, we can again lower bound $P(K\leq K_{0n})$ by Theorem 2 of \cite{bfa}:
$$P(K \leq K_{0n}) \geq (1-C_2e^{-K_{0n} \log(1+J_n)}),$$
where $C_2$ is a suitable constant. Furthermore, by Lemma G.5 of the Appendix in \cite{bfa}, we have the prior concentration of the true loading matrix under the Frobenius norm:
$$P(\|B\|_F \leq \sqrt{K_{0n}/n}) \gtrsim e^{-C_3 (K_{0n} S_n + 1)(1+ \log(J_n + 1))}$$
To lower bound the last product term, note each term can be lower bounded by the smallest density of $K_{0n}$ multidimensional Gaussian multiply by the area of a hypercube, with side length $\sqrt{q_n}$ , where $q_n := \frac{ K_{0n} S_n \log(J_n)}{ \gamma_n n^2 J_n (\sqrt{K_{0n} / n} + \sqrt{S_n J_n})^2 }$:
\begin{align*}
    \prod_{i=1}^n P(\|Z\|_{\infty}^2 \leq \frac{ K_{0n} S_n \log(J_n)}{ \gamma_n n^2 J_n (\sqrt{K_{0n} / n} + \sqrt{S_n J_n})^2 }) &\geq \prod_{i=1}^n q_n^{K_{0n}/2} \inf_{\|x\|_{\infty} \leq \sqrt{q_n}} (2 \pi)^{-K_{0n}/2} \exp \{-\frac{1}{2}x'x\} \\
    &\geq (\frac{q_n}{2 \pi})^{n K_{0n}/2} \exp \{-nq_n/2\} \\
    & \gtrsim e^{-(K_{0n} S_n + 1)(1+ \log(J_n + 1)) }
\end{align*}
To see why the last inequality is true for sufficiently large $n$, we can take $\log$ on both side. It is clear the second term $\frac{-nq_n}{2} \gtrsim -(K_{0n} S_n + 1)(1+ \log(J_n + 1) $ due to the large denominator term. To see why the first term $-\frac{nK_{on}}{2}\log(2\pi/q_n) \gtrsim - (K_{0n} S_n +1)(1+\log(J_n + 1))$, we can divide the fractional term and then exponentiate both side, and it is sufficient to show $\frac{K_{0n}S_n \log(J_n)}{2\pi \gamma_n n^2 J_n(\sqrt{K_{0n}/n} + \sqrt{S_n J_n})^2} \geq (J_n+1)^{-c_2}$ for sufficiently large n and constant $c_2$. Since $\gamma_n \lesssim \log(J_n)$ by assumption C, $J_n >n$ by assumption A, we know the left hand side is lower bounded by the order of $J_n^{-4
}$, and then with the choice $c_2 > 4$, the last inequality holds.

Finally, with sufficient large $n$,  both $P(\|B\|_F \leq \sqrt{K_{0n}/n})$ and the product terms above are lower bounded by the order $e^{- (K_{0n} S_n + 1)(1+ \log(J_n + 1)) }$. Hence, we have shown $\pi(A_n) \gtrsim e^{- K_{0n} S_n \log(J_n)}$, and with sufficient large choice of $D$, the display in (\ref{eq:con_pos}) approaches zero. 
\end{proof}

\subsection{Proof of Theorem 3.5} \label{subsec:posterior contraction}

\begin{proof}
  For simplicity, let $g_n := C_0 k_{0n} S_n$ and we introduce the following notations:
    \begin{equation*}
      \epsilon_n := \sqrt{\frac{S_n k_{0n} \log^2 J_n}{n}}, \quad e_n := S_n k_{0n} \log J_n, \quad t_n := C_2  e_n^2, \quad \delta_n := \frac{C_3 \epsilon_n}{J_n^{3/2} g_n (\sqrt{n}+\sqrt{g_n})}.  
    \end{equation*}
   Define event $E_n = \{B \in \mathbbm R^{J_n \times \infty} : K(B) \leq C_0 k_{0n} S_n | Y^{(n)} \}$. To prove the Theorem, we start by bounding the entropy number for a smaller subspace of $E_n$, and then argue the remaining complement space is negligible. Similar to the proof in \cite{pati}, we restrict our attention on the factor loading matrix $B$ with controlled $l_1$ norm and the cardinality of the support higher than certain threshold $\delta_n$ smaller than $H e_n$, where $H$ is a constant to be specified later in Lemma \ref{lem:sieve-complement}. Here we define the $l_1$ norm of matrix $\|B\|_1:= \sum_{j} \sum_k |B_{jk}|$, same as the $l_1$ norm for $\text{vec}(B)$, the vectorization of $B$. Denote $\text{Supp}_{\delta_n}(B_n)$ as the number of entries in $B_n$ with absolute value higher than $\delta_n$. We hence consider the sieve:
   \begin{equation}
        B_{n,k,S} := \{B_n \in \mathbbm R^{J_n \times k}: k \leq C_0k_{0n}S_n,  \text{Supp}_{\delta_n}(B_n) = S \leq H e_n,  \|B_n\|_1 \leq t_n \} \subset E_n.  
   \end{equation}
   Denote $B_n^* := \bigcup_{\substack{n, k, s}} B_{n, k, s} \subset E_n$, event $F_n := \{D_n \geq e^{-C_1 k_{0n} S_n \log(J_n)} \}$, where $D_n$ is the evidence term defined in equation (\ref{eq:evidence-lb}). For any $\epsilon > 0$,  we then have
    \begin{align*}
    & E_{B_0} \Pi_{n} \{ B: H_n (B, B_0) > \epsilon | Y^{(n)}  \} \\
    & \leq E_{B_0} \Pi_{n} \{ B \in E_n: H_n (B, B_0) > \epsilon | Y^{(n)}  \} \mathbbm{1}_{F_n} + E_{B_0} \Pi(E_{n}^{c} | Y^{(n)}) + P_0 E_{n}^c \\
    & \leq E_{B_0} \Pi_{n} \{ B \in B_n^*: H_n (B, B_0) > \epsilon | Y^{(n)}  \} \mathbbm{1}_{F_n} + E_{B_0} \Pi_{n} \{ B_n^{*c} \cap E_n | Y^{(n)}  \} \mathbbm{1}_{F_n} +  E_{B_0} \Pi(E_{n}^{c} | Y^{(n)}) + P_0 F_{n}^c \\
    & \xrightarrow[n \to \infty]{}E_{B_0} \Pi_{n} \{ B \in B_n^*: H_n (B, B_0) > \epsilon | Y^{(n)}  \} \mathbbm{1}_{F_n}.
    \end{align*}
    We show the expression of $E_{B_0} \Pi_{n} \{ B_n^{*c} \cap E_n | Y^{(n)}  \} \mathbbm{1}_{F_n}$ approaches $0$ in Lemma \ref{lem:sieve-complement}, an analogous of Lemma $9.3$ in \cite{pati}. The expressions $E_{B_0} \Pi(E_{n}^{c} | Y^{(n)})$ and $P_0 F_{n}^c$ also approach to $0$ by Theorem 3.4, and the evidence lower bound established in equation (\ref{eq:evidence-lb-2}).

    To study the asymptotic property of $E_{B_0} \Pi_{n} \{ B \in B_n^*: H_n (B, B_0) > \epsilon | Y^{(n)}  \} \mathbbm{1}_{F_n}$, we are interested in constructing a test $\varphi_n$ to better separate the neighborhood around the true $B_0$, and then applying Lemma 9 of \cite{noniid}. This motivates us to bound the entropy number for $B_n^*$ in terms of the Hellinger distance, by showing $\log N(\frac{\epsilon_n}{36}, B_n^*, H_n) \lesssim n \epsilon_n^2$. We hence consider a slightly larger space $\tilde{B_n}^* :=  \bigcup_{\substack{n, k, s}} \tilde{B}_{n,k,s}$, where
     \begin{equation*}
        \tilde{B}_{n,k,S} := \{B_n \in \mathbbm R^{J_n \times k}: k \leq C_0k_{0n}S_n,  \text{Supp}_{\delta_n}(B_n) = S \leq H e_n,  \|B_n\|_F \leq g_n t_n \}. 
   \end{equation*}
    We know $B_{n, k, s} \subset  \tilde{B}_{n,k,S}$ since $\|B_n\|_1 \leq t_n$ implies $\|B_n\|_F \leq \sqrt{k} \|B_n\|_2 \leq k \|B_n\|_1 \leq g_n t_n$. For arbitrary loading matrices $B_1, B_2 \in \tilde{B}_{n,k,S}$, and latent factor matrix $\Theta $, we define $\eta_{1,ij}:=B_{1j}' \theta_i$, $\eta_{2, ij}:= B_{2j}'\theta_i$, $\Theta_k,  B_{1,k}, B_{2,k}$ as the first $k$ columns of $\Theta$, $B_1$, and $B_2$ respectively. We can then upper bound the squared Hellinger distance by the KL divergence, which has been computed in equation (\ref{eq:kl-explicit}):
   \begin{align*}
        H_{n, \theta}^{2}(B_1, B_2) & \leq \sum_{i=1}^n K_{\theta}(B_1, B_2) \\
        & =  \sum_{i=1}^n \sum_{j=1}^{J_n} \frac{\phi^2(\eta_{1, ij})}{2 \Phi(\eta_{1, ij}) (1- \Phi(\eta_{1, ij}))} (\eta_{1, ij}-\eta_{2, ij})^2 + O((\eta_{1, ij}-\eta_{2, ij})^2 ) \\
        & \leq \gamma_n \|B_{1,k} \Theta_k' - B_{2,k} \Theta_k' \|_F^2 \leq \gamma_n  \|\Theta_k\|_2^2 \|B_{1,k} - B_{2,k}\|_F^2,  
   \end{align*}
   where $\gamma_n = \max_{i,j} 1 + \frac{\phi^2(\eta_{1, ij})}{2 \Phi(\eta_{1, ij}) (1- \Phi(\eta_{1, ij}))}$. Observe if $\|\Theta_k\|_2 \|B_{1,k} - B_{2,k}\|_F \leq (\gamma_n)^{-1/2} \epsilon_n$, we have the above expression $\gamma_n  \|\Theta_k\|_2^2 \|B_{1,k} - B_{2,k}\|_F^2 \leq \epsilon_n^2$. Since the 2-norm of Gaussian matrix $\|\Theta_k\|_2 \lesssim \sqrt{n} + \sqrt{k}$ (see section 4.4.2 of \cite{hdp}), there exits a constant $C_4$ such that
   $$N(\frac{\epsilon}{36}, \tilde{B}_{k, S, n}, H_n) \leq N \left( \frac{C_4 \epsilon_n /36}{(\gamma_n)^{\frac{1}{2}}(\sqrt{n} + \sqrt{k})}, \tilde{B}_{k, S, n} , \|\|_F \right).$$
Write $\xi_n := \frac{C_4 \epsilon_n /36}{(\gamma_n)^{\frac{1}{2}}(\sqrt{n} + \sqrt{k})}$, and we proceed to explicitly construct a $\xi_n$ - net for the space $\tilde{B}_{k, S, n}$ as follows: let $\{\bar{\theta_l}\}_{l=1}^L$ be a $\frac{\xi_n}{2}$-net of an Euclidean ball in $\mathbbm R^{s}$ with radius $g_n t_n$. Then by the standard result (see Corollary 4.2.13 of \cite{hdp}), we have $L \leq 
 \big(1+\frac{g_n t_n}{\xi_n} \big)^s$. In addition, since $\frac{\phi^2(\eta_{1, ij})}{2 \Phi(\eta_{1, ij}) (1- \Phi(\eta_{1, ij}))}$ grows linearly by the Mills ratio, and yet $\max_{ij} \eta_{ij} \lesssim \log (n \cdot J_n)$ grows only in log scale with controlled factor dimension, we have $\gamma_n \lesssim J_n$. This implies there exists a constant $C$ such that $\big(1+\frac{g_n t_n}{\xi_n} \big)^s \leq e^{CslogJ_n}$. For any $B \in \tilde{B}_{k, S, n}$, we consider its vectorization $\text{vec}(B)$, with $\text{vec}(B_s)$ corresponding to the components in $\text{vec}(B)$ where the elements have value larger than $\delta_n$. By definition of the $\frac{\xi_n}{2}$-net, there exists a $\theta_l \in \{\bar{\theta_l}\}_{l=1}^L$, such that $\|\text{vec}(B_s) - \bar{\theta_l}\|_2 \leq \frac{\xi_n}{2}$. We can then define a factor loading matrix $B^{(\bar{\theta_l})}$ such that $\text{vec}\big(B^{(\bar{\theta_l})}_s \big) = \theta_l$ and $\text{vec}\big(B^{\bar{(\theta_l})}_{s^c} \big) = 0$. It then follows $$\|B-B^{(\bar{\theta_l})}\|_F \leq \|\text{vec}(B_s) - \bar{\theta_l}\|_2 + \|\text{vec}(B_{S^c})\|_2 \leq \frac{\xi_n}{2} + \delta_n J_n g_n \leq \xi_n,$$
with sufficient small choice of $C_3$.

After finding the cover number for each sieve $\tilde{B}_{k, S, n}$, we are ready to bound the entropy number for the whole space $B_{n}^*$:
\begin{align*}
    N(\frac{\epsilon_n}{36}, B_n^*, H_n) \leq N(\frac{\epsilon_n}{36}, \tilde{B_n^*}, H_n) &\leq \sum_{k, s} N\left(\frac{C_4 \epsilon_n /36}{(\gamma_n)^{\frac{1}{2}}(\sqrt{n} + \sqrt{k})}, \tilde{B}_{k, S, n} , \|\|_F\right) \\
    & = \sum_{k: k \leq g_n} \sum_{s=0}^{He_n} \binom{J_n \cdot k}{s} \big(1 + \frac{t_n g_n}{\xi_n} \big)^s \\
    & \leq g_n \cdot (J_{n}e_n)^{He_n} \cdot e^{CHe_n \log J_n}  .
\end{align*}
Recall $e_n = S_n k_{0n} \log(J_n)$, it follows we have $\log N(\frac{\epsilon_n}{36}, B_n^*, H_n) \lesssim S_n {k_{0n}} \log^2{J_n} =  n \epsilon_n^2$. By Lemma 9 of \cite{noniid}, there exists a test $\varphi_n$, and some constant $C_5$ such that
$$E_{B_0} \varphi_n \leq \frac{1}{2} \exp \big\{C_5 n \epsilon_n^2 - \frac{n\epsilon^2}{2} \big\}, \quad \sup_{B \in B_n^*: H_n(B, B_0)>\epsilon} E_B (1-\varphi_n) \leq \exp\big\{\frac{-n\epsilon^2}{2} \big\}.$$
Hence we have
\vspace{-0.7cm}
\begin{align*}
    & E_{B_0} \Pi_{n} \{ B \in B_n^*: H_n (B, B_0) > \epsilon | Y^{(n)}  \} \mathbbm{1}_{F_n} \\
    & \leq E_{B_0} \Pi_{n} \{ B \in B_n^*: H_n (B, B_0) > \epsilon | Y^{(n)}  \} \mathbbm{1}_{F_n}(1-\varphi_n) + E_{B_0} \varphi_n \\
    & \leq \big\{ \sup_{B \in B_n^*: H_n(B, B_0)>\epsilon} E_{B}(1-\varphi_n)\big\} \exp \{C_1 S_n k_{0n} \log J_n \} + E_{B_0} \varphi_n \\
    & \leq \exp \big\{ \frac{-n\epsilon^2}{2} + C_1 S_n k_{0n} \log J_n  \big\} + \frac{1}{2} \exp \big\{C_5 n \epsilon_n^2 - \frac{n\epsilon^2}{2}\big\} \rightarrow 0.
\end{align*}
By assumption A and the fact that $\epsilon_n < \epsilon$ as $n \to \infty$, this completes the proof.

\end{proof}

\begin{lemma} \label{lem:sieve-complement}
Recall the events  $B_n^* = \{ |\text{supp}_{\delta_n}(B_n)| \leq H e_n, \|B_n\|_1 \leq t_n\}$, and $E_n = \{K(B) \leq C_0 k_{0n} S_n \}$, where $e_n = k_{0n}S_n \log J_n$, $t_n= C_2 e_n^2$.  there exist constants $H$ and $C_2$ such that
     $$E_{B_0} \Pi_{n} \{ B_n^{*c} \cap E_n | Y^{(n)}  \} \mathbbm{1}_{F_n} \xrightarrow[n \to \infty]{} 0.$$
\end{lemma}
\begin{proof}
Define events $G_n:=\{|\text{supp}_{\delta_n}(B_n)| > H e_n\}$, and $G_n':= \{|\|B_n\|_1 > t_n\}$. Similar to the proof of Lemma 9.3 of \cite{pati}, we will show:
\begin{align}
    E_{B_0} \Pi_{n} \{ G_n \cap E_n | Y^{(n)}  \} \mathbbm{1}_{F_n} \xrightarrow[n \to \infty]{} 0. \label{eq:first_g} \\
    E_{B_0} \Pi_{n} \{ G_n' \cap E_n | Y^{(n)}  \} \mathbbm{1}_{F_n} \xrightarrow[n \to \infty]{} 0. \label{eq:second_g'}
\end{align}
We start by proving equation (\ref{eq:first_g}) and the proof of equation (\ref{eq:second_g'}) can be proceeded similarly:
\begin{align*}
    E_{B_0} \Pi_{n} \{ G_n \cap E_n | Y^{(n)}  \} &= E_{B_0}\left\{\frac{\int_{G_n \cap E_n} \prod_{i=1}^N \prod_{j=1}^J f_{\pi}(y_{ij})/f_{\pi_0}(y_{ij}) d \Pi }{\int \prod_{i=1}^N \prod_{j=1}^J f_{\pi}(y_{ij})/f_{\pi_0}(y_{ij}) d \Pi} \mathbbm{1}_{F_n}\right\}  \\
    &\leq  \frac{\Pi_n(G_n \cap E_n)}{e^{c_1 k_{0n} S_n \log J_n}}  \leq  \frac{\Pi_n(G_n | E_n)}{e^{-c_1 k_{0n} S_n \log J_n}} . 
\end{align*}

For any $k \leq g_n$, we may write $|\text{supp}_{\delta_n}(B_n)| := \sum_{k} \xi_k$, where $\xi_k \sim \text{Binomial}(J_n, P(|B_{jk}| > \delta_n))$. Hence, $|\text{supp}_{\delta_n}(B_n)|$ follows a Poisson binomial distribution, since each column of $B_n$ has distinct prior calibration. Define $\mu := \sum_k E[\xi_k]$ and $t:= H S_n k_{0n} \log J_n$, the standard Azuma-Hoeffding concentration inequality (see section 2 of \cite{10.1214/22-STS852}) yields:
\begin{equation} \label{eq:posbin}
    P(|\text{supp}_{\delta_n}(B_n)| > t ) \leq \left(\frac{\mu}{t} \right)^{t}\left(\frac{g_n - \mu}{g_n-t}\right)^{(g_n-t)}, \quad \text{for any } \mu < t.
\end{equation}
We need to show the right hand side of equation (\ref{eq:posbin}) decays faster than $e^{-c_1 k_{0n} S_n \log J_n}$. Recall with the prior choices $\alpha=\frac{1}{J_n}$, $\lambda_{0k} \geq \frac{2J_n^2k^3 n}{S_n}$, $\lambda_1 < e^{-2}$ and $\delta_n= \frac{C_3 \epsilon_n}{J_n^{3/2} g_n (\sqrt{n}+\sqrt{g_n})}$,we have for any column $k$, 
\vspace{-0.5cm}
$$P(|B_{n, jk}| > \delta_n) \leq \frac{1}{J_n} \exp \{-\delta_n \lambda_1\} + \frac{J_n-1}{J_n} \exp\{-\delta_n \lambda_0\} \lesssim \frac{\log(J_n)}{J_n}.$$
This suggests we have $\mu = \sum_k E[\xi_k] \lesssim C_{0}k_{0n}S_n \log J_n$. Define $\rho := \frac{H S_n k_{0n} \log J_n}{\mu}$. Taking the log on the right hand side  of equation (\ref{eq:posbin}) yields
$$-Hk_{0n}S_n \log J_n \Big[\log \bigg\{\frac{\mu - C_0 S_nk_{0n}}{H S_n k_{0n} \log J_n - C_0 S_n k_{0n}} \cdot \rho \bigg\} \Big]  + S_nK_{0n} \log \frac{\mu - C_0 S_nk_{0n}}{H S_n k_{0n} \log J_n - C_0 S_n k_{0n}}.$$ With sufficient large of $H$, we can make the equation above smaller than $-c_1 k_{0n} S_n \log J_n$. The completes the proof for equation $(\ref{eq:first_g})$.

To prove the claim in equation (\ref{eq:second_g'}), we may argue similarly as before, and show $P(\|B_n\|_1 > t_n) \leq e^{-c_1 k_{0n} S_n \log J_n}$. Since $t_n = C_2 e_n^2$, and $e_n =k_{0n} S_n \log J_n$, it is sufficient to show the concentration 
\begin{equation*} \label{eq:bernstein}
    P(\|B\|_1 > t) \leq e^{-c \sqrt{t}} \quad \text{for } t > 1,
\end{equation*}
where $B$ is a $J$-dimensional random vector of spike-and-slab Laplacian variables. Conditional on the prior proportion $\theta \sim \text{Beta}(\frac{1}{J_n}, 1)$, each element of $B_j$ has sub-exponential norm $\|B_j\|_{\psi_1} = \frac{\lambda_1 + \theta(\lambda_0 - \lambda_1)}{2\lambda_1 \lambda_0}$. But this is the same scenario as in Lemma 7.4 of \cite{pati}. An application of the Bernstein tail bound for $B$ would verify equation (\ref{eq:second_g'}) for any $t > 1$. 

\end{proof}

\section{E-step for Binary Sparse MIRT} \label{sec: binary_proof}

\subsection{Proof of Theorem 4.2}

\begin{proof}
The likelihood term can be expressed as $$\prod_{j=1}^{J} \Phi (B_j' \theta_i + d_j)^{y_{j}} (1- \Phi(B_j' \theta_i +d_j))^{1-y_j} = \prod_{j=1}^{J} \Phi\{(2y_{j}-1)(B_j' \theta_i +d_j)\} = \Phi_J \{ (D_1 \theta_i + D_2) ; \mathbbm I_J\}.$$
The posterior then becomes
\begin{align*}
\pi(\theta_i | y, B, D) & \propto \phi_k(\theta_i - \xi ; \Omega) \Phi_J \{ (D_1 \theta_i + D_2) ; \mathbbm I_J\} \\
& = \phi_k(\theta_i - \xi ; \Omega) \Phi_J \{ S^{-1}(D_1 \theta_i + D_2) ; S^{-1}S^{-1}\} \\
& = \phi_k(\theta_i - \xi ; \Omega) \Phi_J \{ S^{-1}(D_1 \xi + D_2) + S^{-1} D_1 (\theta_i - \xi) ; S^{-1}S^{-1}\}.
\end{align*}

The result follows from the probability density kernel as defined in definition 4.1. In particular, the main difference from Theorem 1 in \cite{Durante_2019} is we need $\gamma_{\text{post}} = S^{-1}(D_1\xi + D_2)$ as opposed to  $S^{-1}(D_1\xi)$. In addition, it's easy to see $\Omega_{post}^{*}$ is still a full rank correlation matrix, as we have the same expression of $\Gamma, \Delta, \bar \Omega$ as in the standard probit regression setting as described in \cite{Durante_2019}. \qedhere

\end{proof}

\subsection{Proof of Corollary 4.3}
\begin{proof}
    Note this corollary has almost the same form as in corollary 2 of \cite{Durante_2019}, except we need the truncation level for $V_1$ as $-S^{-1}(D_1 \xi +D_2)$ instead of $-S^{-1} D_1 \xi$. The proof still relies on equation 12 from  \cite{unified_skew_normal}, which states the distribution of $\operatorname{SUN}_{p, n}(\xi, \Omega, \Delta, \gamma, \Gamma)$ is equivalent to
    $$\xi + \omega (V_0 + \Delta \Gamma^{-1} V_{1, - \gamma}),$$
    where $V_0 \sim N(0, \bar \Omega - \Delta \Gamma^{-1} \Delta') \in \mathbbm R^{K}$ and $V_{1, -\gamma}$ is obtained by component-wise truncation below $-\gamma$ of a variate $N(0, \Gamma) \in R^{J}$. Plugging our results in Theorem 4.2 into the equation above yields the desired result. 
\end{proof}

\section{Sparse MIRT with Parameter Expansion} \label{sec:px-em}
\vspace{-0.2cm}

In the MIRT literature, PX-EM has not gained widespread adoption. Although Rubin and Thomas demonstrated a toy example for estimating a three-parameter logistic MIRT model \cite{px-em-irt}, they assume that each item only loads on a single factor and evaluates the likelihood using a normal measurement error approximation. This assumption is overly stringent for practical applications.

Rather than depending on a normal approximation, we can extend our parameter space following the idea in \cite{bfa}, incurring minimal computational cost. Their approach involves augmenting the likelihood portion of the posterior by introducing an additional auxiliary parameter, which serves to rotate the loading matrix towards a sparser configuration (PXL-EM). We therefore introduce a new model parameter $A$, allowing our expanded model to take on the following structure:
\begin{equation} \label{eq:complete_model}
    Y_{ij} | \theta_i, B_j', d_j , A\sim \text{Bernoulli}(\Phi(B_j' A_L^{-1}\theta_i+d_j)), \text{ } \theta_i | A \sim N(0, A), \text{ } A \sim \pi(A), 
\end{equation}
where $A_L$ is the lower Cholesky factor for A. Here we assume an improper prior for $A$ so that $\pi(A) \propto 1$. Define $\Delta^{*}:= (B^{*}, C, D, A)$. As suggested in \cite{px-em}, two essential conditions are needed to justify a proper parameter expansion strategy for the EM algorithm: (1) There exists a many-to-one reduction function $R$ that can preserve the observed data model: $Y_{\text{obs}}| \Delta^* \sim P(Y_{\text{obs}}| \Delta = R(\Delta^*))$ for any model parameter $\Delta^*$; (2) there exists a null value $A_0$ such that for all $\Delta$, the complete data model is preserved: $P(Y_{\text{com}}| B^{*}, C, D, A_0) = P(Y_{\text{comp}}| \Delta= B^{*}, C, D)$. To verify for condition (1), consider the reduction function $R(B^*, C, D, A)= (B, C, D)$, where $B=B^*A_L$. Note the observed-data likelihood is still invariant in this expanded model. In the original model, we have 
\vspace{-0.3cm}
$$Y_{ij} | B_j, d_j \sim  \text{Bernoulli}(\Phi(N(d_j, \|B_j\|^2)));$$
In the expanded PX-EM model, we have $$Y_{ij}| B_j^*, d_j, A \sim \text{Bernoulli}(\Phi( N(d_j, B_j'A_l^{-1} A (B_j'A_L^{-1})')) = \text{Bernoulli}(\Phi( N(d_j, \|B_j\|^2))).$$ Hence condition (1) is satisfied. To verify for condition (2), we can refer to the complete data model as demonstrated in equation \ref{eq:complete_model}, and observe by letting $A_0 = \mathbbm{I_K}$ would recover the original complete data model $Y_{ij}| \Delta, \theta_i$ since $B=B^*$. Hence our expanded version satisfies the conditions of the PX-EM algorithm as defined in \cite{px-em}.

\vspace{-0.3cm}
\subsection{Probit PXL-EM Algorithm}
\vspace{-0.2cm}
The derivation of the PXL-EM algorithm is straightforward. In summary, we now need to estimate model parameters $\Delta^{*}= (B^{*}, C, D, A)$, where we define $B^{*}= BA_L^{-1}$ and and $A_0= I_k$. The optimization step now decomposes into three independent optimization problems:
\vspace{-0.3cm}
$$Q(\Delta^{*}) = Q_1(B^*, D) + Q_2(C) + Q_3(A),$$ 
where $Q_1(B^*,D)$ and $Q_2(C)$ can be maximized using the same procedure as presented in Section 4.3 respectively. In addition, we have 
\begin{equation*}
Q_3(A)   \propto -\frac{1}{2} \sum_{i=1}^n \langle \theta_i'A^{-1} \theta_i \rangle -\frac{n}{2} \log |A| = -\frac{1}{2} \sum_{i=1}^{n} \Tr (A^{-1} \langle \theta_i \theta_i' \rangle) - \frac{n}{2} \log (|A|).
\end{equation*}

\vspace{-0.5cm}
\subsection{PXL-EM Updates}
\vspace{-0.2cm}
 The E-step of our PXL-EM algorithm is the same as in section 4.2.2, except that we use $B = B^{*}A_L$ instead of $B^{*}$ when computing conditional expectation. For the M-step, the updates for $(B^{*}, C, D)$ are the same as illustrated in Section 4.3. To update $A$, Consider the first-order condition:
\begin{equation*}
\begin{split}
\frac{\alpha Q_3(A)}{\alpha A} = 0  &\implies -\frac{1}{2} \sum_{i=1}^{N} (-A^{-1} \langle \theta_i \theta_i' \rangle A^{-1}) - \frac{n}{2} A^{-1} = 0\\ 
       & \implies A^{(m+1)} = \frac{1}{N} \sum_{i=1}^{N} \langle \theta_i \theta_i'\rangle  \approx \frac{1}{N\cdot M} \sum_{i=1}^{N} \sum_{m=1}^{M}  \theta_{i}^{(m)} \theta_i^{(m)'}.
\end{split} 
\end{equation*}

As mentioned before, while there exists a closed form for the second moment of $\langle \theta_i \theta_i'\rangle$, the computation would be too involved, so we propose to utilize the i.i.d. random samples from the E-step in order to update $A$ efficiently with Monte-Carlo estimation. This can be achieved with a simple matrix computation with no additional computational cost. 

\subsection{Remarks on Convergence Properties}
\vspace{-0.2cm}
Since our EM algorithm approximates the E-step through i.i.d. Monte Carlo draws, its convergence properties resemble those of the MCEM algorithm proposed by Wei and Tanner \cite{doi:10.1080/01621459.1990.10474930}. Consequently, the conventional guarantee of monotone convergence associated with the regular EM algorithm no longer holds, given that the E-step is not conducted in a precise manner. Another factor contributing to the potential non-monotonic convergence is the parameter expansion within the likelihood portion of the posterior. However, as noted by the authors in \cite{bfa}, the PXL-EM algorithm is expected to outperform the vanilla EM algorithm in terms of convergence speed.

Fort and Moulines \cite{10.1214/aos/1059655912} have established the convergence properties of MCEM for models within curved exponential families, of which our probit MIRT model is a specific instance. Notably, they demonstrate that point-wise convergence to the stationary point of the likelihood function can be achieved under certain regularity conditions for the MCEM algorithm with probability one. This is conditioned on the Monte Carlo sample size $M$ increasing over iterations, such that $\sum_{t} M_{(t)}^{-1} < \infty$.

\section{IBP Experiment: More Analysis} \label{sec:ibp_more}

\subsection{Visualizing Dynamic Posterior Exploration} \label{subsec:dpe_vis}
 Figure \ref{fig:dpe} showcases the dynamic posterior exploration process for our PXL-EM algorithm. Notably, PXL-EM stabilized when $\lambda_0=6$, and was eventually able to recover the true factor patterns.

 \begin{figure}[ht]
    \centering
    \includegraphics[width=0.8\textwidth, height=0.55\textwidth]{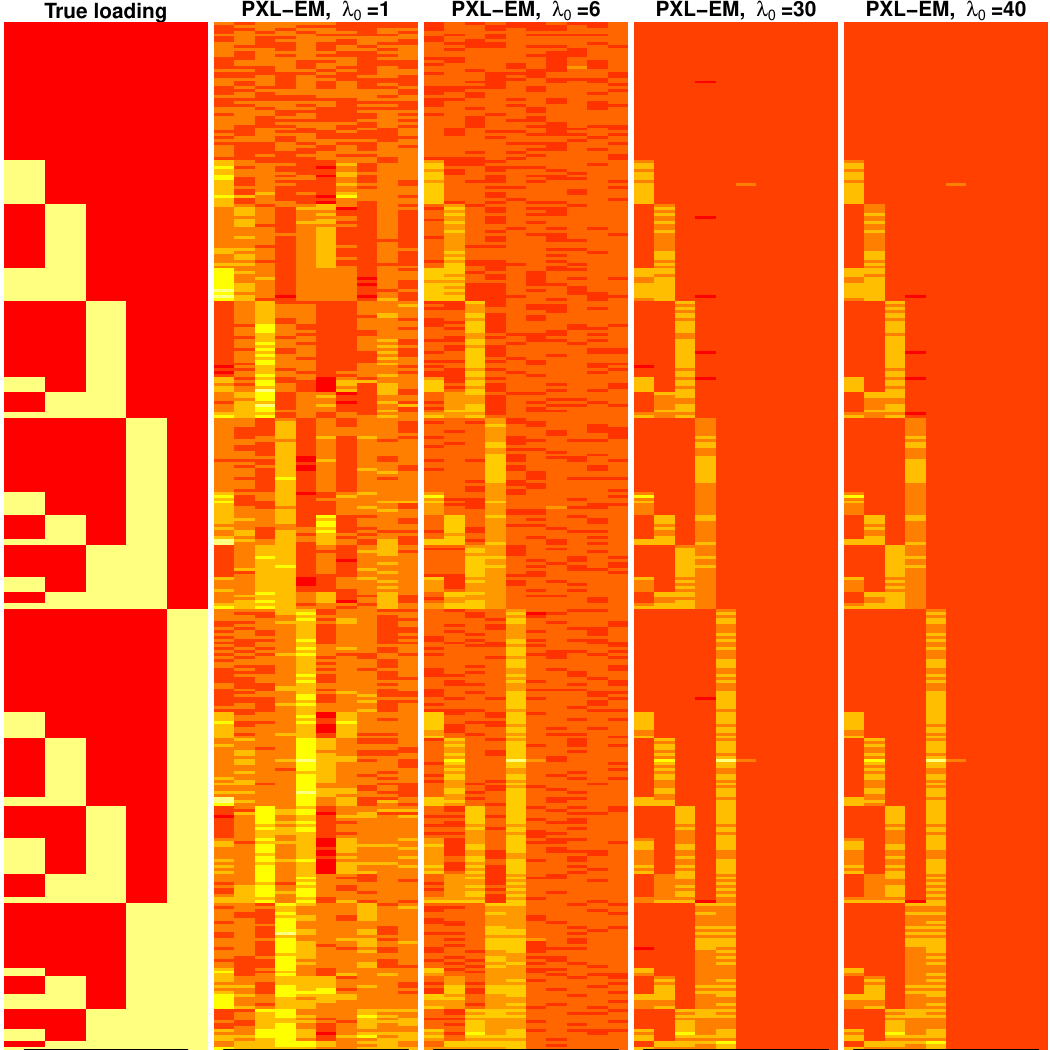}
    \caption{Dynamic Posterior Exploration of PXL-EM Algorithm}
    \label{fig:dpe}
\end{figure}

\subsection{IBP Consistency Experiment} \label{subsec:ibp_consis}

 We conduct a follow-up consistency experiment using the same IBP experimental setup described in Section 6.2. The loading matrix remains as visualized in the leftmost subplot of Figure 2, where $J=350$ and true number of factors $K^*=5$. Recall that $N=250$ in the original experiment; here, we plan to investigate the effect of sample size $N$ on the estimation mean square error of the true loading matrix, by varying $N$ from $100$ to $1,000$. For each $N$, we keep the same loading matrix, the same initialization procedure (uniformly generated between $[-0.02, 0.02]$), and the same dynamical posterior exploration process ($\lambda_1 = 0.5$ and $\lambda_0 \in \{0.5, 1, 3, 6, 10, 20 , 30, 40\}$). 

 Since the IBP loading matrix possesses a GLT structure, the probit IBP loading matrix must be identifiable by Theorem 3.3, provided that the intercept is also identifiable. Therefore, we should expect the MAP estimator produced by our PXL-EM algorithm would become increasingly accurate as $N$ increases. Table \ref{tab:ibp_consistency} empirically verifies our theory, as both the loading MSE and the intercept MSE steadily decrease as $N$ becomes larger.  

\begin{table}[htbp]
    \centering
    \renewcommand{\arraystretch}{1.2}
    \captionsetup{labelfont=bf, format=hang}
    \caption{Mean Squared Error of IBP Loading Matrix Estimation}
    \label{tab:ibp_consistency}
    \begin{tabular}{l|l| l}
        \toprule
        Number of Observations &  Loading Matrix MSE & Intercept MSE \\
        \midrule
        100 & 0.06078 &  0.07932\\
        150 & 0.02024 & 0.04724 \\
        200 & 0.01459 &  0.0315 \\
        250 & 0.00969 & 0.02185 \\
        300 &  0.0078 & 0.01703 \\
        500 & 0.00453 & 0.01088 \\
        750& 0.00322  & 0.0055 \\
        1000 & 0.00256 & 0.00616 \\
        \bottomrule
    \end{tabular}
\end{table}

\section{DESE Experiment: Loading Tables} \label{subsec:dese_loading}

 We report the first six factors of the estimated loading for the math items (table \ref{tab:m_item_loadings}), which could be referenced back to the exact items as appeared in the exam. We also attach the exact loading estimates of our PXL-EM algorithm for the english items (table \ref{tab:e_item_loadings}). Note how the second dimension of the english items is entirely sparse indicating math skills are not needed for English items as expected. The classification of these items can be found on the last page here:   \url{https://www.doe.mass.edu/mcas/2022/release/gr10-math.pdf}

\clearpage

\begin{table}[htbp]
    \centering
    \renewcommand{\arraystretch}{1.2}
    \captionsetup{labelfont=bf, format=hang}
    \caption{PXL-EM Loading Estimates of DESE MATH Items in 2022}
    \label{tab:m_item_loadings}
    \begin{tabular}{l|l|l|l|l|l|l|l|l}
        \toprule
        Item in Figure 3 & Item in Exam & Type & dim1 & dim2 & dim3 & dim4 & dim5 & dim6  \\
        \midrule
        m22 & mitem1 & Algebra & 0.49 & 0.54 & 0 & 0.04 & 0 & 0\\
        m23 & mitem2 & Algebra & 0.6 & 0.48 & 0 & 0.16 & 0 & 0\\
        m24 & mitem3 & Geometry & 0.42 & 0.4 & 0.16 & 0.18 & 0 & 0\\
        m25 & mitem4 & Algebra & 0.56 & 0.65 & 0 & 0 & 0 & 0\\
        m26 & mitem5 & Algebra & 0.45 & 0.52 & 0.12 & 0 & 0 & 0\\
        m27 & mitem7 & Geometry & 0.28 & 0.28 & 0.1 & 0 & 0 & 0\\
        m28 & mitem8 & Algebra & 0.43 & 0.51 & 0.15 & 0 & 0 & 0\\
        m29 & mitem9 & Geometry & 0.52 & 0.48 & 0.2 & 0.04 & 0 & 0\\
        m30 & mitem10 & Algebra & 0.38 & 0.51 & -0.02 & 0 & 0 & 0\\
        m31 & mitem12 & Number & 0.51 & 0.59 & 0 & 0 & 0 & 0\\
        m32 & mitem15 & Geometry & 0.45 & 0.46 & 0.28 & 0 & 0 & 0\\
        m33 & mitem16 & Algebra & 0.59 & 0.59 & 0.14 & 0 & 0 & 0\\
        m34 & mitem17 & Algebra & 0.38 & 0.48 & 0.08 & 0 & 0 & 0\\
        m35 & mitem19 & Geometry & 0.49 & 0.38 & 0.22 & 0.09 & 0 & 0\\
        m36 & mitem20 & Algebra & 0.35 & 0.38 & 0.13 & 0 & 0 & 0\\
        m37 & mitem21 & Geometry & 0.47 & 0.45 & 0.17 & 0 & 0 & 0\\
        m38 & mitem22 & Geometry & 0.4 & 0.4 & 0.21 & 0 & 0 & 0\\
        m39 & mitem23 & Number & 0.42 & 0.37 & 0.2 & 0 & 0 & 0\\
        m40 & mitem24 & Statistics & 0.49 & 0.34 & 0.2 & 0.22 & 0 & 0\\
        m41 & mitem25 & Algebra & 0.36 & 0.27 & 0.1 & 0 & 0 & 0\\
        m42 & mitem26 & Geometry & 0.28 & 0.28 & 0.27 & 0 & 0 & 0\\
        m43 & mitem28 & Geometry & 0.48 & 0.49 & 0.15 & 0 & 0 & 0\\
        m44 & mitem29 & Algebra & 0.51 & 0.39 & 0.22 & 0.29 & 0 & 0\\
        m45 & mitem30 & Geometry & 0.29 & 0.32 & 0.13 & 0 & 0 & 0\\
        m46 & mitem31 & Number & 0.45 & 0.31 & 0.13 & 0.1 & 0 & 0\\
        m47 & mitem33 & Statistics & 0.62 & 0.42 & 0.05 & 0.32 & 0 & 0\\
        m48 & mitem36 & Geometry & 0.31 & 0.28 & 0.1 & 0 & 0 & 0\\
        m49 & mitem37 & Statistics & 0.55 & 0.49 & 0.08 & 0.25 & 0 & 0\\
        m50 & mitem38 & Geometry & 0.35 & 0.22 & 0.14 & 0 & 0 & 0\\
        m51 & mitem40 & Algebra & 0.32 & 0.31 & 0.16 & 0 & 0 & 0\\
        m52 & mitem41 & Geometry & 0.29 & 0.36 & 0.15 & 0 & 0 & 0\\
        m53 & mitem42 & Algebra & 0.48 & 0.41 & 0.16 & 0.17 & 0 & 0\\
        \bottomrule
    \end{tabular}
\end{table}

\clearpage
\begin{table}[htbp]
    \centering
    \renewcommand{\arraystretch}{1.2}
    \captionsetup{labelfont=bf, format=hang}
    \caption{PXL-EM Loading Estimates of DESE English Items in 2022}
    \label{tab:e_item_loadings}
    \begin{tabular}{l|l|l|l|l|l|l|l}
        \toprule
        Item in Figure 3 & Item in Exam & dim1 & dim2 & dim3 & dim4 & dim5 & dim6 \\
        \midrule
        e1 & eitem1 & 0.37 & 0 & 0 & 0.06 & 0 & 0\\
        e2 & eitem2 & 0.54 & 0.12 & 0.05 & 0 & 0 & 0\\
        e3 & eitem3 & 0.38 & 0.06 & 0 & 0 & 0 & 0\\
        e4 & eitem4 & 0.57 & 0.1 & 0 & 0 & 0 & 0\\
        e5 & eitem5 & 0.47 & 0.09 & 0 & 0 & 0 & 0\\
        e6 & eitem6 & 0.47 & 0.11 & 0 & 0 & 0 & 0\\
        e7 & eitem10 & 0.49 & 0.13 & 0 & 0 & 0 & 0\\
        e8 & eitem11 & 0.5 & 0.07 & 0 & 0 & 0 & 0\\
        e9 & eitem12 & 0.54 & 0.18 & 0.07 & 0 & 0 & 0\\
        e10 & eitem14 & 0.43 & 0.08 & 0 & 0 & 0 & 0\\
        e11 & eitem15 & 0.39 & 0 & 0 & 0.04 & 0 & 0\\
        e12 & eitem16 & 0.63 & 0.04 & 0 & 0.08 & 0 & 0\\
        e13 & eitem17 & 0.64 & 0.12 & 0 & 0 & 0 & 0\\
        e14 & eitem18 & 0.65 & 0.15 & 0 & 0 & 0 & 0\\
        e15 & eitem19 & 0.27 & 0.05 & 0 & 0 & 0 & 0\\
        e16 & eitem23 & 0.57 & 0.12 & 0 & 0 & 0 & 0\\
        e17 & eitem24 & 0.59 & 0.08 & 0 & 0 & 0 & 0\\
        e18 & eitem25 & 0.62 & 0.1 & 0.06 & 0 & 0 & 0\\
        e19 & eitem26 & 0.76 & 0.15 & 0.07 & 0.05 & 0 & 0\\
        e20 & eitem27 & 0.67 & 0.17 & 0 & 0 & 0 & 0\\
        e21 & eitem28 & 0.78 & 0.15 & 0.09 & 0 & 0 & 0\\
        \bottomrule
    \end{tabular}
\end{table}

\section{QOL Experiment: More Analysis} \label{subsec:qol_loading}

\subsection{Loading Tables}

We report the factor loading estimates of the PXL-EM algorithm (table \ref{tab:qol_em_loadings}) and of the bifactor model (table \ref{tab:qol_bifactor_loadings}) for the QOL data. As depicted in Figure 4, the health-related items exhibit significant loading on both the fifth and the tenth factors for our EM algorithm. Among the six health items, only items $11$-$13$ load primarily on factor $5$, while all but item $13$ load primarily on factor $10$. Interestingly, items $11$-$13$ pertain to medicare accessibility, while the remaining health items are focused on general health status. \footnote{Items $11-13$ inquire about ``Medical care", ``Frequency of doctor visits", and ``Interactions with therapists". Items $10, 14, 15$ explore ``General health", ``Physical condition", and ``Emotional well-being''.} Consequently, we posit that factor $5$ signifies medical access, while factor $10$ captures general physical well-being. Notably, item $14$, which directly addresses ``physical condition'', exhibits the highest loading on factor $10$.

The ``Leisure'' factor is the only item category that loads solely on the primary component. One plausible reason is that many leisure items, as indicated in figure 4, do not display strong loadings on the secondary leisure dimensions, only with the exception of items $18$ and $19$. However, items $18$ and $19$ pose general queries about ``Chance to Enjoy time'' and ``Amount of fun", respectively, which may have been well-reflected by the underlying primary factor representing quality of life.  

\begin{table}[htbp]
    \centering
    \renewcommand{\arraystretch}{1.2}
    \captionsetup{labelfont=bf, format=hang}
    \caption{PXL-EM Loading Estimates of QOL Data}
    \label{tab:qol_em_loadings}
    \begin{tabular}{l|l|l|l|l|l|l|l|l|l|l}
        \toprule
        Item & dim1 & dim2 & dim3 & dim4 & dim5 & dim6 & dim7 & dim8 & dim9  & dim10 \\
        \midrule
        1-global & 0.802 & 0 & 0 & 0 & 0 & 0 & 0 & 0 & 0 & 0.244\\
        2-family & 0.578 & 0 & 0.702 & 0 & 0 & 0 & 0 & 0 & 0 & 0\\
        3-family & 0.588 & 0 & 0.51 & 0 & 0 & 0 & 0 & 0 & 0 & 0\\
        4-family & 0.628 & 0 & 0.69 & 0 & 0 & 0 & 0 & 0 & 0 & 0\\
        5-family & 0.665 & 0 & 0.64 & 0 & 0 & 0 & 0.145 & 0 & 0 & 0\\
        6-finance & 0.582 & 0 & 0 & 0.693 & 0 & 0 & 0 & 0 & 0 & 0\\
        7-finance & 0.554 & 0 & 0 & 0.552 & 0 & 0 & 0 & 0 & 0 & 0\\
        8-finance & 0.618 & 0 & 0 & 0.663 & 0 & 0 & 0 & 0 & 0 & 0\\
        9-finance & 0.644 & 0 & 0 & 0.634 & 0 & 0 & 0 & 0 & 0 & 0\\
        10-health & 0.462 & 0 & 0 & 0 & 0 & 0 & 0 & 0.222 & 0 & 0.601\\
        11-health & 0.548 & 0 & 0 & 0 & 0.422 & 0 & 0 & 0 & 0 & 0.35\\
        12-health & 0.52 & 0 & 0 & 0 & 0.571 & 0 & 0 & 0 & 0.206 & 0.224\\
        13-health & 0.606 & 0 & 0 & 0 & 0.609 & 0 & 0 & 0 & 0 & 0\\
        14-health & 0.593 & 0 & 0 & 0 & 0 & 0 & 0 & 0.147 & 0 & 0.669\\
        15-health & 0.677 & 0 & 0 & 0 & 0 & 0 & 0 & 0.191 & 0 & 0.314\\
        16-leisure & 0.817 & 0 & 0 & 0 & 0 & 0 & 0 & 0 & 0 & 0\\
        17-leisure & 0.68 & 0 & 0 & 0 & 0 & 0 & 0 & 0 & 0 & 0\\
        18-leisure & 0.772 & 0 & 0 & 0 & 0 & 0 & 0 & 0 & 0 & 0\\
        19-leisure & 0.859 & 0 & 0 & 0 & 0 & 0 & 0 & 0 & 0 & 0\\
        20-leisure & 0.748 & 0 & 0 & 0 & 0 & 0 & 0 & 0 & 0 & 0\\
        21-leisure & 0.554 & 0 & 0 & 0 & 0 & 0 & 0 & 0 & 0 & 0\\
        22-living & 0.593 & 0 & 0 & 0 & 0 & 0 & 0.595 & 0 & 0 & 0\\
        23-living & 0.497 & 0 & 0 & 0 & 0 & 0 & 0.489 & 0 & 0 & 0\\
        24-living & 0.515 & 0 & 0 & 0 & 0 & 0 & 0.663 & 0 & 0 & 0\\
        25-living & 0.516 & 0 & 0 & 0 & 0 & 0 & 0.685 & 0 & 0 & 0\\
        26-living & 0.564 & 0 & 0 & 0 & 0 & 0 & 0.549 & 0 & 0 & 0\\
        27-safety & 0.561 & 0 & 0 & 0 & 0 & 0 & 0 & 0.68 & 0 & 0\\
        28-safety & 0.551 & 0 & 0 & 0 & 0 & 0 & 0.331 & 0.589 & 0 & 0\\
        29-safety & 0.568 & 0 & 0 & 0 & 0 & 0 & 0 & 0.356 & 0 & 0\\
        30-safety & 0.498 & 0 & 0 & 0 & 0 & 0 & 0 & 0.611 & 0 & 0\\
        31-safety & 0.58 & 0 & 0 & 0 & 0 & 0 & 0 & 0.572 & 0 & 0\\
        32-social & 0.694 & 0 & 0 & 0 & 0 & 0 & 0 & 0 & 0.568 & 0\\
        33-social & 0.739 & 0 & 0 & 0 & 0 & 0 & 0 & 0 & 0.523 & 0\\
        34-social & 0.604 & 0 & 0 & 0 & 0 & 0 & 0 & 0 & 0.431 & 0\\
        35-social & 0.555 & 0 & 0 & 0 & 0 & 0 & 0 & 0 & 0 & 0\\
        \hline
        \bottomrule
    \end{tabular}
\end{table}

\begin{table}[htbp]
    \centering
    \renewcommand{\arraystretch}{1.2}
    \captionsetup{labelfont=bf, format=hang}
    \caption{Oracle Bifactor Loading Estimates of QOL Data}
    \label{tab:qol_bifactor_loadings}
    \begin{tabular}{l|l|l|l|l|l|l|l|l|l}
        \toprule
        Item & dim1 & dim2 & dim3 & dim4 & dim5 & dim6 & dim7 & dim8 & dim9   \\
        \midrule
        1-global & 0.853 & 0.01 & 0 & 0 & 0 & 0 & 0 & 0 & 0\\
        2-family & 0.604 & 0.00 & 0.692 & 0 & 0 & 0 & 0 & 0 & 0\\
        3-family & 0.631 & 0.00 & 0.480 & 0 & 0 & 0 & 0 & 0 & 0\\
        4-family & 0.646 & 0.00 & 0.687 & 0 & 0 & 0 & 0 & 0 & 0\\
        5-family & 0.703 & 0.00 & 0.622 & 0 & 0 & 0 & 0 & 0 & 0\\
        6-finance & 0.593 & 0.00 & 0 & 0.692 & 0 & 0 & 0 & 0 & 0\\
        7-finance & 0.559 & 0.00 & 0 & 0.554 & 0 & 0 & 0 & 0 & 0\\
        8-finance & 0.622 & 0.00 & 0 & 0.664 & 0 & 0 & 0 & 0 & 0\\
        9-finance & 0.641 & 0.00 & 0 & 0.643 & 0 & 0 & 0 & 0 & 0\\
        10-health & 0.581 & 0.00 & 0 & 0 & 0.251 & 0 & 0 & 0 & 0\\
        11-health & 0.584 & 0.00 & 0 & 0 & 0.522 & 0 & 0 & 0 & 0\\
        12-health & 0.564 & 0.00 & 0 & 0 & 0.596 & 0 & 0 & 0 & 0\\
        13-health & 0.608 & 0.00 & 0 & 0 & 0.501 & 0 & 0 & 0 & 0\\
        14-health & 0.711 & 0.00 & 0 & 0 & 0.257 & 0 & 0 & 0 & 0\\
        15-health & 0.755 & 0.00 & 0 & 0 & 0.081 & 0 & 0 & 0 & 0\\
        16-leisure & 0.787 & 0.00 & 0 & 0 & 0 & 0.211 & 0 & 0 & 0\\
        17-leisure & 0.632 & 0.00 & 0 & 0 & 0 & 0.243 & 0 & 0 & 0\\
        18-leisure & 0.692 & 0.00 & 0 & 0 & 0 & 0.434 & 0 & 0 & 0\\
        19-leisure & 0.766 & 0.00 & 0 & 0 & 0 & 0.562 & 0 & 0 & 0\\
        20-leisure & 0.689 & 0.00 & 0 & 0 & 0 & 0.339 & 0 & 0 & 0\\
        21-leisure & 0.531 & 0.00 & 0 & 0 & 0 & 0.182 & 0 & 0 & 0\\
        22-living & 0.631 & 0.00 & 0 & 0 & 0 & 0 & 0.553 & 0 & 0\\
        23-living & 0.519 & 0.00 & 0 & 0 & 0 & 0 & 0.491 & 0 & 0\\
        24-living & 0.542 & 0.00 & 0 & 0 & 0 & 0 & 0.635 & 0 & 0\\
        25-living & 0.555 & 0.00 & 0 & 0 & 0 & 0 & 0.682 & 0 & 0\\
        26-living & 0.586 & 0.00 & 0 & 0 & 0 & 0 & 0.539 & 0 & 0\\
        27-safety & 0.617 & 0.00 & 0 & 0 & 0 & 0 & 0 & 0.637 & 0\\
        28-safety & 0.642 & 0.00 & 0 & 0 & 0 & 0 & 0 & 0.501 & 0\\
        29-safety & 0.577 & 0.00 & 0 & 0 & 0 & 0 & 0 & 0.321 & 0\\
        30-safety & 0.554 & 0.00 & 0 & 0 & 0 & 0 & 0 & 0.548 & 0\\
        31-safety & 0.629 & 0.00 & 0 & 0 & 0 & 0 & 0 & 0.554 & 0\\
        32-social & 0.680 & 0.00 & 0 & 0 & 0 & 0 & 0 & 0 & 0.606\\
        33-social & 0.727 & 0.00 & 0 & 0 & 0 & 0 & 0 & 0 & 0.541\\
        34-social & 0.600 & 0.00 & 0 & 0 & 0 & 0 & 0 & 0 & 0.456\\
        35-social & 0.567 & 0.00 & 0 & 0 & 0 & 0 & 0 & 0 & 0.221\\
        \hline
        \bottomrule
    \end{tabular}
\end{table}

\subsection{Factor Inference} \label{subsec:factor_inference}

To showcase the potential of our approach in posterior factor inference, we conducted a comparative analysis of the estimated latent factors and their posterior variances between our proposed PXL-EM approach and the oracle bifactor model using the QOL dataset. With our PXL-EM algorithm, after successfully estimating the factor loading matrix $B$ and the intercepts $D$, we applied Theorem 4.2 to efficiently generate extensive posterior samples of latent factors for conducting inference on posterior means, posterior variances, and other objects of interest. Specifically, we drew 500 samples of latent factors for each subject $i$ and estimated their posterior means and standard errors. This entire inferential process was completed in just 9.9 seconds on our personal laptop.

In the case of the bifactor model, we obtained the estimated latent factors and their standard errors using the "mirt" R package, which required approximately 14.3 seconds for the entire estimation process. Table \ref{tab:qol_factor_inference} presents the correlations and mean squared errors of estimated latent factors and their standard errors when comparing our PXL-EM algorithm with the oracle bifactor model. Notably, our estimated latent factors and their standard errors exhibit a remarkably high degree of correlation with those derived from the oracle bifactor model for each factor. Furthermore, the mean squared errors between our estimation and the bifactor estimation are consistently minimal, affirming the accuracy and reliability of our estimation approach.

It's worth mentioning that the most significant estimation disparities between the PXL-EM algorithm and the bifactor model concentrate on the health component. This arises because our algorithm identified two latent factors for the health items, as discussed in Section 6.4. Please note that we did not consider the "Leisure" component in our analysis, as our PXL-EM did not recover this particular latent factor.

\begin{table}[htbp]
    \centering
    \renewcommand{\arraystretch}{1.2}
    \captionsetup{labelfont=bf, format=hang}
    \caption{Factor Inference: Oracle Bifactor v.s PXL-EM}
    \label{tab:qol_factor_inference}
    \begin{tabular}{l|l|l|l|l}
        \toprule
        Factor Name & Mean Correlation & Mean MSE & SE Correlation & SE MSE \\
        \midrule
        \hline
        Primary & 0.981 & 0.034 & 0.905 & 0.001\\
        Family & 0.981 & 0.022 & 0.862 & 0.006\\
        Finance & 0.983 & 0.019 & 0.815 & 0.006\\
        Health & 0.892 & 0.089 & 0.702 & 0.008\\
        Living & 0.951 & 0.055 & 0.742 & 0.009\\
        Safety & 0.953 & 0.052 & 0.756 & 0.008\\
        Social & 0.934 & 0.054 & 0.78 & 0.008\\
        \hline
        \bottomrule
    \end{tabular}
\end{table}

Finally we want to highlight the flexibility and efficiency of our approach in posterior factor inference. Thanks to our ability to efficiently derive extensive samples of latent factors from their posterior distribution without relying on time-consuming MCMC sampling, we gain the capability to conduct rapid inference on a wide range of objects of interest. This flexibility extends beyond just posterior mean or posterior variance, allowing us to explore and analyze various aspects of the latent factor distribution with ease.

\section{Cortical Thickness Experiment: Loading Tables} \label{sec:cor_loading}

For the ordinal bio-behavioral dataset, We provide the estimated factor loadings for our PXL-EM algorithm in table \ref{tab:cor_em_loadings} as below.  Additionally, table \ref{tab:bi_em_comp} compares our estimated loadings with the primary dimension of the confirmatory bifactor models as appeared in \cite{Stan2020}. In particular, table \ref{tab:bi_em_comp} highlights the striking similarities between our exploratory findings and the established scientific results.

\clearpage
\begin{table}[htbp]
    \centering
    \renewcommand{\arraystretch}{0.95}
    \captionsetup{labelfont=bf, format=hang}
    \caption{PXL-EM Loading Estimates of the Cortical Thickness Data}
    \label{tab:cor_em_loadings}
    {\scriptsize
    \begin{tabular}{l|l|l|l|l|l|l|l|l|l|l}
        \toprule
        Item & dim1 & dim2 & dim3 & dim4 & dim5 & dim6 & dim7 & dim8 & dim9  & dim10 \\
    \hline
    thick\_1 & 0.13 & 0.57 & 0.26 & 0.34 & 0 & 0 & 0 & 0 & 0 & 0\\
    thick\_2 & 0.1 & 0.5 & 0.45 & 0.41 & 0.1 & 0 & 0 & 0 & 0 & 0\\
    thick\_3 & 0.24 & 0.72 & 0 & 0.43 & 0 & 0 & 0 & 0 & 0 & 0\\
    thick\_4 & 0.23 & 0.62 & 0.16 & 0 & 0 & 0 & 0 & 0 & 0 & 0\\
    thick\_5 & 0.15 & 0.59 & 0.4 & 0.44 & 0.07 & 0 & 0 & 0 & 0 & 0\\
    thick\_6 & 0.28 & 0.73 & 0 & 0 & 0 & 0 & 0 & 0 & 0 & 0\\
    thick\_7 & 0.18 & 0.75 & 0 & 0.18 & 0 & 0 & 0 & 0 & 0 & 0\\
    thick\_8 & 0.13 & 0.7 & 0.34 & 0.16 & 0 & 0 & 0 & 0 & 0 & 0\\
    thick\_9 & 0.2 & 0.8 & 0 & 0.33 & 0 & 0 & 0 & 0 & 0 & 0\\
    thick\_10 & 0.11 & 0.65 & 0.36 & 0 & 0 & 0 & 0 & 0 & 0 & 0\\
    thick\_11 & 0.14 & 0.58 & 0.51 & 0.34 & 0.1 & 0 & 0 & 0 & 0 & 0\\
    thick\_12 & 0.31 & 0.68 & 0 & 0 & 0 & 0 & 0 & 0 & 0 & 0\\
    thick\_13 & 0.27 & 0.69 & 0 & 0 & 0 & 0 & 0 & 0 & 0 & 0\\
    thick\_14 & 0.16 & 0.68 & 0.46 & 0.15 & 0.07 & 0 & 0 & 0 & 0 & 0\\
    thick\_15 & 0.18 & 0.76 & 0 & 0.2 & 0 & 0 & 0 & 0 & 0 & 0\\
    thick\_16 & 0.12 & 0.67 & 0.36 & 0 & 0 & 0 & 0 & 0 & 0 & 0\\
    MADRS\_1 & 0.74 & -0.17 & 0 & 0 & -0.28 & -0.15 & -0.25 & 0 & 0 & 0\\
    MADRS\_2 & 0.73 & -0.21 & 0 & 0 & -0.32 & -0.3 & -0.22 & 0 & 0 & 0\\
    MADRS\_3 & 0.63 & -0.14 & 0 & 0 & 0 & -0.45 & 0 & 0 & 0 & 0\\
    MADRS\_4 & 0.43 & -0.2 & 0 & 0 & 0 & -0.37 & 0 & 0 & 0.14 & -0.63\\
    MADRS\_5 & 0.55 & -0.15 & 0 & 0 & 0 & -0.39 & 0 & 0 & 0 & 0\\
    MADRS\_6 & 0.6 & -0.17 & 0 & 0 & 0 & -0.31 & 0 & 0 & 0 & 0\\
    MADRS\_7 & 0.66 & -0.21 & 0 & 0 & -0.24 & -0.24 & 0 & 0 & 0 & 0\\
    MADRS\_8 & 0.69 & -0.17 & 0 & 0 & -0.19 & -0.22 & 0 & 0 & 0 & 0\\
    MADRS\_9 & 0.62 & -0.18 & 0 & 0 & -0.27 & -0.34 & 0 & 0 & 0 & 0\\
    MADRS\_10 & 0.67 & -0.21 & 0 & 0 & 0 & -0.31 & 0 & 0 & 0 & 0\\
    PANSS\_p1 & 0.55 & -0.27 & -0.14 & 0 & 0.31 & 0 & 0 & 0.56 & 0 & 0\\
    PANSS\_p2 & 0.53 & -0.23 & -0.16 & 0 & 0.49 & 0 & 0 & 0.19 & 0 & 0\\
    PANSS\_p3 & 0.48 & -0.38 & 0 & 0 & 0.24 & 0 & 0 & 0.42 & 0 & 0\\
    PANSS\_p4 & 0.4 & -0.15 & 0 & 0 & 0.55 & -0.43 & 0 & 0 & 0 & 0\\
    PANSS\_p5 & 0.31 & -0.13 & 0 & 0 & 0.46 & -0.25 & 0 & 0.3 & 0 & 0\\
    PANSS\_p6 & 0.62 & -0.29 & 0 & 0 & 0.23 & -0.18 & 0.17 & 0.23 & 0 & 0\\
    PANSS\_p7 & 0.58 & -0.14 & 0 & 0 & 0.43 & -0.39 & 0 & -0.14 & 0 & 0\\
    PANSS\_n1 & 0.56 & -0.13 & 0 & 0 & 0 & 0.58 & 0 & 0 & 0 & 0\\
    PANSS\_n2 & 0.79 & -0.22 & 0 & 0 & 0 & 0.21 & 0.24 & 0 & 0 & 0\\
    PANSS\_n3 & 0.68 & -0.12 & 0 & 0 & 0.28 & 0.36 & 0 & -0.25 & 0 & 0\\
    PANSS\_n4 & 0.65 & -0.25 & 0 & 0 & 0 & 0.17 & 0.4 & 0 & 0 & 0\\
    PANSS\_n5 & 0.29 & -0.2 & 0 & 0 & 0.24 & 0.2 & 0 & 0 & 0 & 0\\
    PANSS\_n6 & 0.61 & -0.1 & 0 & 0 & 0.11 & 0.48 & 0 & -0.17 & 0 & 0\\
    PANSS\_n7 & 0.66 & -0.23 & -0.22 & 0 & 0.45 & 0.14 & 0 & 0 & 0 & 0\\
    PANSS\_g1 & 0.5 & -0.19 & 0 & 0 & 0.15 & -0.18 & 0 & 0 & 0 & 0\\
    PANSS\_g2 & 0.65 & -0.17 & 0 & 0 & 0 & -0.38 & 0 & 0 & 0 & 0\\
    PANSS\_g3 & 0.44 & -0.16 & 0 & 0 & 0 & -0.25 & 0 & 0 & 0 & 0\\
    PANSS\_g4 & 0.66 & -0.12 & 0 & 0 & 0.25 & -0.29 & 0 & 0 & 0 & 0\\
    PANSS\_g5 & 0.63 & -0.12 & 0 & 0 & 0.37 & 0.21 & 0 & 0 & 0 & 0\\
    PANSS\_g6 & 0.76 & -0.24 & 0 & 0 & -0.21 & -0.26 & -0.15 & 0 & 0 & 0\\
    PANSS\_g7 & 0.69 & -0.17 & 0 & 0 & 0 & 0.37 & 0 & 0 & 0 & 0\\
    PANSS\_g8 & 0.58 & -0.13 & 0 & 0 & 0.51 & 0 & 0 & -0.31 & 0 & 0\\
    PANSS\_g9 & 0.48 & -0.26 & -0.12 & 0 & 0.38 & 0 & 0 & 0.53 & 0 & 0\\
    PANSS\_g10 & 0.5 & -0.17 & 0 & 0 & 0.35 & 0 & 0 & 0 & 0 & 0\\
    PANSS\_g11 & 0.57 & -0.17 & 0 & 0 & 0.49 & 0 & 0 & 0 & 0 & 0\\
    PANSS\_g12 & 0.4 & -0.12 & 0 & 0 & 0.47 & 0 & 0 & 0 & 0 & 0\\
    PANSS\_g13 & 0.68 & -0.24 & 0 & 0 & 0.26 & 0 & 0 & 0 & 0 & 0\\
    PANSS\_g14 & 0.58 & -0.19 & 0 & 0 & 0.43 & -0.3 & 0 & 0 & 0 & 0\\
    PANSS\_g15 & 0.69 & -0.21 & -0.23 & 0 & 0.39 & 0 & 0 & 0 & 0 & 0\\
    PANSS\_g16 & 0.7 & -0.27 & 0 & 0 & 0 & 0 & 0.44 & 0 & 0 & 0\\
    Young\_1 & 0.06 & 0 & 0 & 0 & 0.41 & -0.41 & 0 & 0 & 0.56 & 0\\
    Young\_2 & 0.11 & 0 & 0 & 0 & 0.36 & -0.42 & 0 & 0 & 0.63 & 0\\
    Young\_3 & 0 & 0 & 0 & 0 & 0.25 & -0.29 & 0 & 0 & 0.34 & 0\\
    Young\_4 & 0.33 & -0.16 & 0 & 0 & 0.17 & -0.38 & 0.05 & 0 & 0.25 & -0.71\\
    Young\_5 & 0.51 & -0.13 & 0 & 0 & 0.19 & -0.5 & 0 & -0.16 & 0 & 0\\
    Young\_6 & 0.18 & 0 & 0 & 0 & 0.36 & -0.47 & 0 & 0 & 0.47 & 0\\
    Young\_7 & 0.42 & 0 & 0 & 0 & 0.29 & -0.35 & 0 & 0 & 0.34 & 0\\
    Young\_8 & 0.54 & -0.19 & 0 & 0 & 0.25 & -0.24 & 0 & 0.28 & 0.23 & 0\\
    Young\_9 & 0.46 & 0 & 0 & 0 & 0.38 & -0.47 & 0 & -0.23 & 0 & 0\\
    Young\_10 & 0.26 & 0 & 0 & 0 & 0 & 0 & 0 & 0 & 0 & 0\\
    Young\_11 & 0.22 & 0 & 0 & 0 & 0.44 & 0 & 0 & 0 & 0 & 0\\
        \bottomrule
    \end{tabular}}
\end{table}

\clearpage

\begin{table}[htbp]
    \centering
    \renewcommand{\arraystretch}{0.95}
    \captionsetup{labelfont=bf, format=hang}
    \caption{Comparison of PXL-EM Findings Versus Confirmatory Bifactor Findings}
    \label{tab:bi_em_comp}
    {\scriptsize
    \begin{tabular}{l|l|l|l}
        \toprule
        Item & Description & Bifactor Primary Dimension & PXL-EM Second Dimension \\
\hline
        thick\_1 & L inferior parietal lobule & 0.73 & 0.57 \\
        thick\_2 & L pars opercularis & 0.74 & 0.5 \\
        thick\_3 & L precuneus & 0.73 & 0.72 \\
        thick\_4 & L supramarginal gyrus & 0.79 & 0.62 \\
        thick\_5 & R supramarginal gyrus & 0.75 & 0.59 \\
        thick\_6 & L lateral orbitofrontal gyrus & 0.69 & 0.73 \\
        thick\_7 & R pars opercularis & 0.72 & 0.75 \\
        thick\_8 & R pars triangularis & 0.7 & 0.7 \\
        thick\_9 & L banks of the superior temporal sulcus & 0.71 & 0.8 \\
        thick\_10 & L fusiform gyrus & 0.66 & 0.65 \\
        thick\_11 & L middle temporal gyrus & 0.76 & 0.58 \\
        thick\_12 & L insula & 0.78 & 0.68 \\
        thick\_13 & R middle temporal gyrus & 0.76 & 0.69 \\
        thick\_14 & R superior temporal gyrus & 0.85 & 0.68 \\
        thick\_15 & L superior temporal gyrus & 0.87 & 0.76 \\
        thick\_16 & R insula & 0.75 & 0.67 \\
        MADRS\_1 & Apparent sadness & -0.22 & -0.17 \\
        MADRS\_2 & Reported sadness & -0.23 & -0.21 \\
        MADRS\_3 & Inner tension & -0.15 & -0.14 \\
        MADRS\_4 & Reduced sleep & -0.16 & -0.2 \\
        MADRS\_5 & Reduced appetite & -0.12 & -0.15 \\
        MADRS\_6 & concentration difficulties & -0.17 & -0.17 \\
        MADRS\_7 & Lassitude & -0.17 & -0.21 \\
        MADRS\_8 & Inability to feel & -0.18 & -0.17 \\
        MADRS\_9 & Pessimistic thoughts & -0.18 & -0.18 \\
        MADRS\_10 & Suicidal thoughts & -0.21 & -0.21 \\
        PANSS\_p1 & Delusions & \textbf{-0.32} & \textbf{-0.27} \\
        PANSS\_p2 & Conceptual disorganization & -0.22 & -0.23 \\
        PANSS\_p3 & Hallucinory Behavior & \textbf{-0.37} & \textbf{-0.38} \\
        PANSS\_p4 & Excitement & -0.03 & -0.15 \\
        PANSS\_p5 & Grandiosity & -0.05 & -0.13 \\
        PANSS\_p6 & Suspiciousness & \textbf{-0.28} & \textbf{-0.29} \\
        PANSS\_p7 & Hostility & -0.01 & -0.14 \\
        PANSS\_n1 & Blunted affect & -0.17 & -0.13 \\
        PANSS\_n2 & Emotional withdrawl & -0.23 & -0.22 \\
        PANSS\_n3 & Poor rapport & -0.02 & -0.12 \\
        PANSS\_n4 & Passive-apathetic social withdrawl & \textbf{-0.28} & \textbf{-0.25} \\
        PANSS\_n5 & Difficulty in abstract thicking & -0.16 & -0.2 \\
        PANSS\_n6 & lack of spontaneity & -0.07 & -0.1 \\
        PANSS\_n7 & Stereotyped thicking & -0.2 & -0.23 \\
        PANSS\_g1 & Somatic concern & -0.22 & -0.19 \\
        PANSS\_g2 & Anxiety & -0.18 & -0.17 \\
        PANSS\_g3 & Guilt Feeling & -0.16 & -0.16 \\
        PANSS\_g4 & Tension & -0.08 & -0.12 \\
        PANSS\_g5 & Mannerisms and posturing & -0.05 & -0.12 \\
        PANSS\_g6 & Depression & \textbf{-0.25} & -0.24 \\
        PANSS\_g7 & Motor retardation & -0.19 & -0.17 \\
        PANSS\_g8 & Uncooperativeness & 0 & -0.13 \\
        PANSS\_g9 & Unusual thought content & \textbf{-0.26} & \textbf{-0.26} \\
        PANSS\_g10 & Disorientation & -0.05 & -0.17 \\
        PANSS\_g11 & Poor attention & -0.04 & -0.17 \\
        PANSS\_g12 & Lack of judgement and insight & -0.09 & -0.12 \\
        PANSS\_g13 & Disturbance of volition & -0.02 & -0.24 \\
        PANSS\_g14 & Poor impulse control & -0.04 & -0.19 \\
        PANSS\_g15 & Preoccupation & -0.17 & -0.21 \\
        PANSS\_g16 & Active social avoidance & \textbf{-0.28} & \textbf{-0.27} \\
        Young\_1 & Elevanted mood & 0.07 & 0 \\
        Young\_2 & Increased motor activity & 0.08 & 0 \\
        Young\_3 & Sexual interest & 0.06 & 0 \\
        Young\_4 & Sleep & -0.13 & -0.16 \\
        Young\_5 & Irritability & -0.09 & -0.13 \\
        Young\_6 & Speech - rate and amount & 0.06 & 0 \\
        Young\_7 & Language-thought disorder & -0.09 & 0 \\
        Young\_8 & Content & -0.22 & -0.19 \\
        Young\_9 & Disruptive-aggressive behavior & -0.01 & 0 \\
        Young\_10 & Appearance & -0.14 & 0 \\
        Young\_11 & Insight & -0.06 & 0 \\
        \bottomrule
    \end{tabular}}
\end{table}

\clearpage 

\section{Extension to Mixed Data Type} \label{sec:mixed}

\subsection{CDF Trick}

Subject $i$'s contribution to the ordinal portion of the likelihood function is embedded in their responses to these $J$ distinct ordinal items:
\begin{equation} \label{eq:ord_lik}
    \prod_{j=1}^{J} \left[\sum_{l=0}^{L_j} \mathbbm{I} \{Y_{ij} = l\} (\Phi(B_j'\theta_i + d_{j, l }) - \Phi(B_j'\theta_i + d_{j, l-1 }))\right].
\end{equation}

At first glance,  reducing Equation \ref{eq:ord_lik} into a probability kernel for a unified skew-normal distribution might appear challenging,  due to the presence of the difference term $(\Phi(B_j'\theta_i + d_{j, l }) - \Phi(B_j'\theta_i + d_{j, l-1 }))$. In the binary case, $- \Phi(B_j'\theta_i + d_{j, l-1 })$ is simply zero, so the likelihood can be written as the product of the form $\Phi((2(Y_{ij}-1)(B_j' \theta_i + d_j)))$. Here one potential trick would be recognizing the difference of two normal CDF functions can be closely approximated by a two-dimensional multivariate-normal CDF function:

\begin{lemma} \label{lem:cdf_trick}
For $l \neq \{0, L_{j}\}$, let $\Phi_2 \left\{ [B_j'\theta_i + d_{j, l }, -B_j'\theta_i - d_{j, l-1 }]' ; - \mathbbm{1}_2 \mathbbm{1}_2' + 2\mathbbm{I} \right\} $ represent the cumulative distribution function of a degenerate two-dimensional multivariate Gaussian with zero mean vector and covariance matrix  $- \mathbbm{1}_2 \mathbbm{1}_2' + 2\mathbbm{I}$ evaluated at vector $[B_j'\theta_i + d_{j, l }, -B_j'\theta_i - d_{j, l-1 }]' $. Then we have
\begin{equation*}
\Phi(B_j'\theta_i + d_{j, l }) - \Phi(B_j'\theta_i + d_{j, l-1 }) = \Phi_2 \left\{ [B_j'\theta_i + d_{j, l }, -B_j'\theta_i - d_{j, l-1 }]' ; - \mathbbm{1}_2 \mathbbm{1}_2' + 2\mathbbm{I} \right\} 
\end{equation*}
\end{lemma}

\begin{proof}
Define $a:= B_j' \theta_i + d_{j, l-1}$ , $b:= B_j' \theta_i + d_{j, l}$. let $Z$ represent a standard normal variable, and let $Z' = -Z$ whose value is determined by $Z$. It follows
\begin{equation}
\Phi(B_j'\theta_i + d_{j, l }) - \Phi(B_j'\theta_i + d_{j, l-1 })  = P(a \leq Z \leq b) = P(Z \leq b, Z' \leq -a ).
\end{equation}
Note the random vector $(Z, Z')'$ is a degenerated multivariate Gaussian as the matrix $\begin{bmatrix}
    1 & -1 \\
    -1 & 1 \\
\end{bmatrix}$ is positive semidefinite but not positive definite. 
\end{proof}

\subsection{Proof of Theorem 5.1}

\begin{proof}
 The likelihood for the continuous portion can be represented as:
    \begin{align*}
       L_c(\theta_i) &= (2\pi)^{-\frac{|J_c|}{2}} |\Sigma|^{-\frac{1}{2}} \exp \{ -\frac{1}{2} (y_{\raisebox{-0.5mm}{$\scriptscriptstyle J_c$}} - B_{J_c} \theta_i)' \Sigma^{-1} (y_{\raisebox{-0.5mm}{$\scriptscriptstyle J_c$}} - B_{J_c} \theta_i)\} \\
       & \propto \exp \{ -\frac{1}{2} (\Sigma^{-\frac{1}{2}}y_{\raisebox{-0.5mm}{$\scriptscriptstyle J_c$}} - \Sigma^{-\frac{1}{2}}B_{J_c} \theta_i)' (\Sigma^{-\frac{1}{2}}y_{\raisebox{-0.5mm}{$\scriptscriptstyle J_c$}} - \Sigma^{-\frac{1}{2}}B_{J_c} \theta_i)\} \\
        & \propto \exp \{-\frac{1}{2} (\hat{\theta_i} - \theta_i)' \tilde{B}_{J_c}' \tilde{B}_{J_c} (\hat{\theta_i} - \theta_i) \} \\
        & \propto \phi_k(\theta_i - \hat{\theta_i} ; (\tilde{B}_{J_c}' \tilde{B}_{J_c})^{-1}).
    \end{align*}

To derive the likelihood for the categorical responses, let $d_{j,l}$ be the shorthand for $d_{j, Y_j}$. For ordinal item $j \in J_{O_1}$, let $d_j$ represent $d_{j,0}$ when $Y_j=0$, but represent $d_{j, L_j-1}$ when $Y_j=L_j$. Using Lemma \ref{lem:cdf_trick}, the likelihood function of the graded response model for each observation $i$ can be expressed as follows:
\begin{align*}
  \prod_{j \in J_B} \Phi((2Y_j -1)(B_j' \theta_i + d_j)) \times \prod_{j \in J_{O_1}} \Phi_2 \{ \begin{bmatrix}
                        B_j'\theta_i + d_{j, l }  \\
                        -B_j'\theta_i - d_{j, l-1 }  \\
                    \end{bmatrix} ;
                    \begin{bmatrix}
                        1 & -1 \\
                        -1 & 1 \\
                    \end{bmatrix} \}  \\
                    \times \prod_{j \in J_{O_2}} \Phi((2\{Y_j=0\}-1)(B_j'\theta_i+d_j)).
\end{align*}
Using the notations above, the likelihood term can be further simplified:
\begin{align*}
    &\Phi_{|J_B|} \{(D_1 \theta_i + D_2) ; \mathbbm{I}_{|J_B|}\} \times \Phi_{2|J_{O_1}|} \{(\tilde B_{J_{O_1}} \theta_i + \tilde D_{J_{O_1}}) ;  \overline{\mathbbm{I_{J_{O_1}}}} \} \times \Phi_{|J_{O_2}|} \{ (D_3 \theta_i + D_4) ; \mathbbm{I}_{|J_{O_2}|} \} \\
    & = \Phi_{\bar J} \{ (\bar D \theta_i + \bar V) ; \overline{\mathbbm{I_{\bar{J}}}} \}.
\end{align*}
 Finally, recall the product of two multivariate Gaussian densities is again proportional to a new multivariate Gaussian density. Since $L_c(\theta_i)$ is proportional to a multivariate normal density in $\theta_i$, we should expect the its product with the prior density $\phi_k(\theta_i - \xi ; \Omega)$ is still multivariate Gaussian in $\theta_i$. This allows us to write the posterior distribution as follows:
\begin{align*}
\pi(\theta_i | y, B, D) & \propto \phi_k(\theta_i - \xi ; \Omega) \phi_k(\theta_i - \hat{\theta_i} ; (\tilde{B}_{J_c}' \tilde{B}_{J_c})^{-1})  \Phi_{\bar J} \{ (\bar D \theta_i + \bar V) ; \overline{\mathbbm{I_{\bar{J}}}} \} \\
& \propto \phi_k(\theta_i - \xi_{\text {post }} ; \Omega_{\text {post }}) \Phi_{\bar J} \{ (\bar D \theta_i + \bar V) ; \overline{\mathbbm{I_{\bar{J}}}} \} \\
& = \phi_k(\theta_i - \xi_{\text {post }} ; \Omega_{\text {post }}) \Phi_{\bar J} \{ S^{-1}(\bar D \theta_i + \bar V) ; S^{-1}\overline{\mathbbm{I_{\bar{J}}}} S^{-1} \}  \\
& = \phi_k(\theta_i - \xi_{\text {post }} ; \Omega_{\text {post }}) \Phi_{\bar J} \{ S^{-1}(\bar D \xi_{\text {post }} + \bar V) + S^{-1} \bar{D} (\theta_i - \xi_{\text {post }}) ; S^{-1} \overline{\mathbbm{I_{\bar{J}}}} S^{-1}\}.
\end{align*}
Note the posterior takes the form of the the probability density function as defined in definition 4.1. In particular, the expressions of $\xi_{\text{post }}, \Omega_{\text{post }}$, and $\gamma_{\text{post }}$ are immediate. Some simple algebra would lead to the claimed expressions for $\Delta_{\text{post }}$ and $\Gamma_{\text{post }}$.

It remains to show $\Omega_{\text{post }}^{*}$ is indeed a full-rank correlation matrix. To see this, note $\Omega_{\text{post }}^{*}$ can be decomposed as follows:
$$
\begin{bmatrix} \Gamma_{\text{post }} & \Delta_{\text{post }}' \\
                  \Delta_{\text{post }} & \bar \Omega \\
                  \end{bmatrix}  = \begin{bmatrix} S^{-1} & 0 \\
                  0 & \omega^{-1}\\
   \end{bmatrix}  \times \begin{bmatrix} \bar D \Omega \bar{D}' + \overline{\mathbbm{I_{\bar{J}}}}  & \bar{D} \Omega \\
                  \Omega \bar{D}' & \Omega\\
   \end{bmatrix} 
   \times \begin{bmatrix} S^{-1} & 0 \\
                  0 & \omega^{-1}\\
   \end{bmatrix}. 
$$
Observe the middle matrix on the right hand side is the covariance matrix of the random vector $[Z_1, Z_2]'$, where $Z_1 = \bar{D}Z_2 + \epsilon$ with $\epsilon$ independent of $Z_2$. Here $Z_2$ is a $K$-dimensional random vector with zero mean and covariance matrix $\Omega$, and $\epsilon$ is $\bar{J}$-dimensional random vector with zero mean and covariance matrix $\overline{\mathbbm{I_{\bar{J}}}}$. Note $\overline{\mathbbm{I_{\bar{J}}}}$ is positive semidefinite, hence a valid covariance matrix.
\end{proof}

\subsection{EM Algorithm for Mixed Data Type}

\subsubsection{E-Step}

Theorem 5.1 explicitly characterizes the latent factor posterior distribution under mixed data distributions as a similar form to the unified-skewed normal distribution. The following corollary makes the E-step more explicit.

\begin{corollary} \label{cor:mixed}
Under the mixed item response data types setup as illustrated in Theorem 5.1, suppose latent factor $\theta_i$ has prior distribution $N(0, \mathbbm{I}_k)$, then its posterior has the unified skew-normal distribution, and can be represented as
$$
(\theta_i | y, B, D, \Sigma) \stackrel{\mathrm{d}}{=} \Omega_{\text {post }}(B_{J_c}'\Sigma^{-1}B_{J_c} \hat{\theta_i})  +\omega_{\text {post }} \left\{V_{0}+\bar{\Omega}_{\text {post }} \omega_{\text {post }} \bar{D}' \left(\bar{D} \Omega_{\text {post }} \bar{D}'+ \overline{\mathbbm{I_{J}}}\right)^{-1} S V_{1}\right\},
$$
where 
\begin{itemize}
    \item $\Omega_{{\text {post }}} = \omega_{\text {post }} \bar{\Omega}_{\text {post }} \omega_{\text {post }} = (B_{j_c}' \Sigma^{-1} B_{J_c} + \mathbbm{I}_k)^{-1}, $
    \item $V_{0} \sim N \left\{0, \bar{\Omega}_{\text {post }}-\bar{\Omega}_{\text {post }} \omega_{\text {post }} \bar{D}'\left(\bar{D} \Omega_{\text {post }} \bar{D}'+ \overline{\mathbbm{I_{J}}} \right)^{-1} \bar{D} \omega_{\text {post }} \bar{\Omega}_{\text {post }}\right\}$, and 
    \item $V_{1}$ is a zero mean $\bar{J}$-variate truncated normal with covariance matrix $S^{-1}\left(\bar{D} \Omega_{\text{post}} \bar{D}'+ \overline{\mathbbm{I_{J}}}\right) S^{-1}$ truncated below $S^{-1}(\bar{D} \xi_{\text {post}} + \bar{V})$, and is independent of $V_0$. 
\end{itemize}

\end{corollary}

\begin{proof}
     This corollary is a direct consequence of Theorem 5.1 and equation 7.4 from \cite{azzalini_2013}, which states $\operatorname{SUN}_{p, n}(\xi, \Omega, \Delta, \gamma, \Gamma)$ has a stochastic representation of 
    $$\xi + \omega (V_0 + \Delta \Gamma^{-1} V_{1, - \gamma}),$$
    where $V_0 \sim N(0, \bar \Omega - \Delta \Gamma^{-1} \Delta') \in \mathbbm R^{K}$ and $V_{1, -\gamma}$ is obtained by component-wise truncation below $-\gamma$ of a variate $N(0, \Gamma) \in R^{J}$. Plugging our prior choice and the posterior parameters in Theorem 5.1 would yield the desired result. 
\end{proof}

When the number of ordinal items is large, it is possible to render the matrix $(\bar{D} \Omega_{\text{post}} \bar{D}' + \bar{\mathbbm{I}}_{\bar{J}})$ non invertible, as both $\bar{D} \Omega_{\text{post}} \bar{D}'$ and $\bar{\mathbbm{I}}_{\bar{J}}$ can become semi positive definite matrices. If this was the case, we may have to do a pseudo inverse. Alternatively, it is easy to see adding any $\epsilon>0$ on its diagonal would yield an invertible matrix, which can serve for a good approximation. In the bio-behavioral ordinal experiment, we added a small $\epsilon$ to the diagonal to ensure invertibility. When we are just mixing continuous and binary items, the matrix $(\bar{D} \Omega_{\text{post}} \bar{D}' + \bar{\mathbbm{I}}_{\bar{J}})$ is guaranteed to be positive definite, and hence invertible. 

\subsubsection{M-Step}

 Once the latent factors are known, the M-step would become straightforward. By the same argument as illustrated in Section 4.3, the M-step can be decomposed into $J$ independent optimization problems, depending on the item types. We can apply the same M-step strategy for binary items as discussed in Section 4.3. Note the same M-step for the probability inclusion parameter $C$ as illustrated in Section 4.3 remain the same. 
 
 For continuous items, consider maximizing equation below:
 \begin{equation} \label{eqn:q1c}
    Q_j(B_j, \sigma_j^2) = - \sum_{i=1}^{N} \frac{Y_{ij}^2 + B_j' \langle \theta_i \theta_i'  \rangle B_j - 2Y_{ij} B_j' \langle \theta_i \rangle }{2 \sigma_j^2}  - \sum_{k=1}^{K} \lambda_{jk}|B_{jk}| -  \frac{n+3}{2} \log \sigma_j^2 - \frac{1}{2\sigma_j^2},
\end{equation}  
 In the E-step, we have already approximated $\langle \theta_i \rangle$ as $\frac{1}{M} \sum_{m=1}^{M} \theta_i^{(m)}$, and $\langle \theta_i \theta_i'\rangle$ as $\frac{1}{M} \sum_{m=1}^{M} \theta_i^{(m)} \theta_i^{(m)'}$. We will maximize $B_j$ and $\sigma_j^2$ iteratively. When holding $\sigma_j^2$ fixed (using the value from previous iteration), the problem of obtaining update $B_j$ is equivalent to solving a linear Lasso regression, in which the response variables $Y_j$ is replicated $M$ times with corresponding design matrix $\Theta_i$ resulting from Monte-Carlo samples. Then the $l_1$ penalty $2\sigma_j^2 \lambda_{j_k}$ is applied for each component $B_{jk}$. After obtaining $B_j$, we can find optimal $\sigma_j^2$ using first-order condition, which yields: 
    $$\sigma_{j}^2 = \frac{ [\frac{1}{M} \sum_{i=1}^{N} \sum_{m=1}^{M} Y_{ij}^2 - B_j' \theta_i^{(m)} \theta_i^{(m)'} B_j- 2Y_{ij}B_j' \theta_i^{(m)}]+1}{n+3}.$$

For ordinal items, we need to maximize  $|J_O|$ independent objective functions as follows:
\begin{align*}
 Q_j(B_j,d_j) &= \sum_{i=1}^{N} \langle \log[\sum_{l=0}^{L} \mathbbm{I} \{Y_{ij} = l\} (\Phi(B_j'\theta_i + d_{j, l }) - \Phi(B_j'\theta_i + d_{j, l-1 }))] \rangle  \\
 & \quad \quad - \sum_{k=1}^{K} |B_{jk}|(\lambda_1  \langle \gamma_{jk} \rangle  + \lambda_0(1- \langle  \gamma_{jk} \rangle ))  \\
       & \approx \frac{1}{M} \sum_{i=1}^{N} \sum_{m=1}^{M} \log[\sum_{l=0}^{L} \mathbbm{I} \{Y_{ij} = l\} (\Phi(B_j'\theta_i^{(m)} + d_{j, l }) - \Phi(B_j'\theta_i^{(m)} + d_{j, l-1 }))]  \\
       & \qquad - \sum_{k=1}^{K} |B_{jk}|(\lambda_1  \langle \gamma_{jk} \rangle  + \lambda_0(1- \langle  \gamma_{jk} \rangle )). 
\end{align*}

Note this is equivalent to estimating a penalized ordered regression with a probit link, with distinct $l_1$-penalty $(\lambda_1  \langle \gamma_{jk} \rangle  + \lambda_0(1- \langle  \gamma_{jk} \rangle )$ applying to each element of the loading matrix. This penalized ordered regression again can be solved by standard off the shelf package.


\end{document}